\documentclass[lettersize,journal]{IEEEtran}

\usepackage{amsmath,amsfonts}
\usepackage{cite}
\usepackage{url}
\usepackage{hyperref}
\hypersetup{colorlinks=true, linkcolor=red, citecolor=red, urlcolor=blue} 
\usepackage{amssymb}
\usepackage{array}
\usepackage{amsthm}
\allowdisplaybreaks 
\usepackage{autobreak}
\usepackage{mathptmx}
\usepackage{mathtools, cuted}
\usepackage{algpseudocode}
\usepackage[ruled, vlined, linesnumbered]{algorithm2e}
\usepackage{graphicx}
\usepackage{bbm}
\usepackage[table, dvipsnames]{xcolor}
\usepackage{caption,subcaption}

\allowdisplaybreaks

\newtheorem{Theorem}{Theorem}
\newtheorem{Lemma}{Lemma}
\newtheorem{Corollary}{Corollary}
\newtheorem{Assumption}{Assumption}
\newtheorem{Remark}{Remark}

\newcommand{\rs}{\!\!}
\newcolumntype{C}[1]{>{\centering \arraybackslash}p{#1}}
\newcommand{\bblue}{\textcolor{black}}

\newcommand{\bquad}{\qquad\qquad\qquad\qquad\qquad}\newcommand{\mquad}{\qquad\qquad\qquad}
\newcommand{\squad}{\qquad\qquad}

\usepackage{acronym}
\acrodef{ml}[ML]{machine learning}
\acrodef{ai}[AI]{artificial intelligence}
\acrodef{fl}[FL]{federated learning}
\acrodef{hfl}[HFL]{hierarchical federated learning}
\acrodef{fedavg}[FedAvg]{federated averaging}
\acrodef{rl}[RL]{reinforcement learning}
\acrodef{drl}[DRL]{deep reinforcement learning}

\acrodef{bs}[BS]{base station}
\acrodef{isp}[ISP]{(wireless) internet service provider}
\acrodef{ue}[UE]{user equipment}
\acrodef{cs}[CS]{central server}
\acrodef{es}[ES]{edge server}
\acrodef{csp}[CSP]{content service provider}
\acrodef{rawhfl}[RawHFL]{\underline{r}esource-\underline{aw}are \underline{h}ierarchical \underline{f}ederated \underline{l}earning}
\acrodef{hfedavg}[H-FedAvg]{hierarchical federated averaging}
\acrodef{sgd}[SGD]{stochastic gradient descent}
\acrodef{cpu}[CPU]{central processing unit}
\acrodef{prb}[pRB]{physical resource block}
\acrodef{snr}[SNR]{signal-to-noise-ratio}
\acrodef{lp}[LP]{linear programming}
\acrodef{fpp}[FPP]{floating point precision}
\acrodef{cdf}[CDF]{cumulative distribution function}
\acrodef{uav}[UAV]{unmanned aerial vehicles}
\acrodef{ap}[AP]{access point}
\acrodef{hz}[Hz]{hertz}
\acrodef{csi}[CSI]{channel state information}
\acrodef{chr}[CHR]{cache hit ratio}
\acrodef{rnn}[RNN]{recurrent neural network}
\acrodef{qoe}[QoE]{quality of experience}
\acrodef{csi}[CSI]{channel state information}
\acrodef{mcs}[MCS]{modulation and coding scheme}
\acrodef{arq}[ARQ]{automatic repeat request}

\begin{document}

\title{Resource-Aware Hierarchical Federated Learning in Wireless Video Caching Networks} 
\author{Md Ferdous Pervej, \IEEEmembership{Member, IEEE}, and Andreas F. Molisch, \IEEEmembership{Fellow, IEEE}
\thanks{Part of this work was presented at the $2024$ IEEE International Conference on Communications (ICC) \cite{pervej2024resource}.}
\thanks{This work was supported by NSF-IITP Project $2152646$.}
\thanks{The authors are with the Ming Hsieh Department of Electrical and Computer Engineering, University of Southern California, Los Angeles, CA $90089$ USA (e-mail: pervej@usc.edu; molisch@usc.edu).}
}

\markboth{Draft} %
{Pervej \MakeLowercase{\textit{et al.}}: Resource-Aware HFL in Wireless Video Caching Networks}

\maketitle

\begin{abstract}
Backhaul traffic congestion caused by the video traffic of a few popular files can be alleviated by storing the to-be-requested content at various levels in wireless video caching networks. 
Typically, content service providers (CSPs) own the content, and the users request their preferred content from the CSPs using their (wireless) internet service providers (ISPs). 
As these parties do not reveal their private information and business secrets, traditional techniques may not be readily used to predict the dynamic changes in users' future demands. 
Motivated by this, we propose a novel \underline{r}esource-\underline{aw}are \underline{h}ierarchical \underline{f}ederated \underline{l}earning (RawHFL) solution for predicting user's future content requests.  
A practical data acquisition technique is used that allows the user to update its local training dataset based on its requested content.
Besides, since networking and other computational resources are limited, considering that only a subset of the users participate in the model training, we derive the convergence bound of the proposed algorithm. 
Based on this bound, we minimize a weighted utility function for jointly configuring the controllable parameters to train the RawHFL energy efficiently under practical resource constraints. 
Our extensive simulation results validate the proposed algorithm's superiority, in terms of test accuracy and energy cost, over existing baselines.
\end{abstract}


\section{Introduction}

\noindent 
\IEEEPARstart{V}{ideo} streaming is the dominant source of data traffic in wireless networks: approximately $3$ out of $5$ video views are generated from wireless devices \cite{loh2022youtube}, accounting for some $70 \%$ of bits sent over the air.
This traffic not only burdens the over-the-air transmission but also the backhaul: since the videos are owned by some \acp{csp}, e.g., YouTube, Netflix, Amazon Prime Video, HBO, Hulu, to name a few, and stored in the cloud servers used by the \acp{csp}, the requested videos need to be obtained from the cloud by the \ac{isp} and transported over its backhaul in order to be delivered to the requester node. 
As such, continuous deliveries of the requested content can overwhelm the backhaul network due to large bandwidth requirements engendered by the combination of large data amount and latency constraints required for the \ac{qoe}.
These problems can be greatly mitigated by caching at (or near) the wireless edge \cite{golrezaei2013femtocaching}, i.e., wireless networks can store the to-be-requested content at the \ac{bs}, reducing repetitive extractions and saving considerable backhaul resources.
Therefore, if the \ac{csp} and the \ac{isp} collaborate, they can design a practical solution to ameliorate the load for the backhaul.

Among many design parameters of a video caching network, two main elements are the content {\em placement} and {\em delivery} of the requested content \cite{liu2016caching}.
Typically, the storage size required to store {\em all} content available from the \acp{csp} far exceeds the capacity of a local cache. 
Consequently, it is critical to know the most likely to-be-requested content sets during the content placement phase.
Unfortunately, determination of these to-be-requested content is challenging: they may deviate considerably from the most recently used ones and may also differ significantly from the global or even regional most-popular files. 
This is particularly true when the users have individual preferences and do not always request the most popular content. 
Thanks to \ac{ml}, such changes can be accurately captured by carefully designing the \ac{ml} model and its training algorithm.
Nonetheless, the need for large amounts of training data samples to accurately learn the \ac{ml} model parameters remains a challenge.


Additionally, content caching becomes inevitably difficult due to the need to protect privacy and business secrets.
In order to get the requested content from the \ac{csp}, the user must place the request using its serving \ac{bs}.
However, neither the \ac{ue} nor the \ac{csp} wants to reveal the exact content ID to the serving \ac{bs}: the \ac{ue} as a matter of privacy, and the \ac{csp} because the user requests are business secrets that it does not disclose to the \ac{isp}, which often is a competitor. 
Conversely,  while the \ac{ue} can place an encrypted content request to the \ac{csp} via the serving \ac{bs}, the \ac{bs} does not wish to reveal the exact spatial information of the requester \ac{ue} to the \ac{csp}, since that is {\em its} business secret.
The above factors must be considered in designing any \ac{ml} algorithm for wireless video caching networks.
Fortunately, these challenges can be handled with \ac{fl} \cite{mcmahan2017communication}, which enables distributed model training on users' devices, thus protecting the users' data privacy.
While traditional \ac{fl} follows a parameter server paradigm, where a \ac{cs} coordinates the training process for the distributed clients\footnote{A user is usually named by the terms {\em \ac{ue}} and {\em client} in wireless networks and \ac{fl}, respectively. We use these terms interchangeably in this paper.}, almost all practical wireless networks have hierarchical structures \cite{wang2022demystifying,pervej2024hierarchical}.
Therefore, it is critical to devise an efficient \ac{hfl} considering the cooperation among the \ac{ue}, \ac{isp} and \ac{csp}.

\subsection{State-of-the-Art}

\noindent
Content placement is widely investigated to maximize \ac{chr} and/or minimize content delivery delay.
Recently, Javedankherad \textit{et al.} leveraged a weighted vertex graph coloring and a greedy algorithm for mobility-aware content placement to maximize the \ac{chr} \cite{javedankherad2022mobility}.
Malak \textit{et al.} optimized cache placement and transmission power to minimize transmission delays in multi-hop wireless heterogeneous networks \cite{malak2023joint}. 
Qian \textit{et al.} first derived closed-form solutions for cache placement and transcoding, followed by designing trading contracts between a communication, caching and computing resource provider and mobile network operators in order to minimize personalized content delivery delay \cite{qian20233c}.
Considering users' long-term request patterns, Ning \textit{et al.} jointly optimized cache placement and content recommendations in order to minimize content delivery costs \cite{ning2023optimizing}.
While these studies show the benefits of content caching, dynamic changes in user preferences, leading to frequent changes in global/regional content popularity, still cause significant challenges in wireless video caching networks.


\ac{ml} is widely used in literature \cite{kang2023content, li2023deep, nguyen2021user, lin2023joint,fang2023drl,zhou2023recommendation,xiong2023reinforcement} to apprehend the dynamic changes in content popularity.
These works broadly fall into two \ac{ml} domains, namely ($1$) supervised learning and ($2$) \ac{rl}.
In the supervised \ac{ml} category, existing works \cite{kang2023content,li2023deep,nguyen2021user} mainly used some variants of \ac{rnn} to predict future content popularity or cache placement.
Besides, \ac{drl} is also widely used \cite{lin2023joint,fang2023drl,zhou2023recommendation,xiong2023reinforcement} to learn long-term cache placement policy. 
However, these works mainly adopted a centralized \ac{ml} strategy, considering training datasets are centrally available.
This assumption differs in many practical cases where training data are only available to the end users who want to protect their data privacy.


Many recent studies \cite{qiao2022adaptive,wang2019edge,wang2023popularity,jiang2021federated,li2023community,khanal2022route,lin2023feedback,maale2023deepfesl,cao2023mobility,feng2023proactive} leveraged \ac{fl}-based solutions for different content caching applications in order to protect users' data privacy.
Qiang \textit{et al.} proposed an adaptive \ac{fl}-based proactive content caching solution \cite{qiao2022adaptive}, where two sets of \ac{drl} agents were used to learn two tasks: subset client selection and clients' local training rounds selection. 
They then leveraged \ac{fl} to aggregate the models of these \ac{drl} agents.
\cite{wang2019edge} also used two separate \ac{drl} agent sets --- one at the user level for learning computation offloading policy and one at the edge server level for learning content placement policy --- and applied federated aggregation of these \ac{drl} models. 
A similar strategy was also adopted by Wang \textit{et al.}, where clients' mobility and preferences were first utilized to cluster the clients \cite{wang2023popularity}.
They then let each cluster head train a \ac{drl} model to learn the cache placement.
In \cite{jiang2021federated}, users first learned their content preferences, followed by sharing their context information and learned preferences with their associated \acp{ap}.
The \acp{ap} then partitioned the users and derived popularity scores of the content that they leveraged to train stochastic variance reduced gradient models to predict the global content popularity.
An attention-weighted \ac{fl} algorithm was proposed by Li \textit{et al.} for predicting content popularity in a device-to-device network, leveraging a similar user partitioning strategy based on mobility and social behavior \cite{li2023community}.
Khanal \textit{et al.} also proposed a self-attention-based \ac{fl} algorithm to predict cell-wise content popularity in vehicular networks \cite{khanal2022route}.

Lin \textit{et al.} proposed a proactive caching scheme using stacked autoencoders for global content popularity prediction \cite{lin2023feedback}.
Maale \textit{et al.} considered a \ac{uav}-assisted caching network \cite{maale2023deepfesl}, where popular content of the ground users was placed in the \ac{uav}'s limited cache to ensure ground users' \acp{qoe}.
Cao \textit{et al.} proposed a similar mobility-aware content caching solution \cite{cao2023mobility}, where FL was used to predict users' trajectories and content popularity. 
These prediction results were then used to place the contents to minimize content delivery costs.
Feng \textit{et al.} also considered a similar approach \cite{feng2023proactive}, which first predicted the residence time of moving vehicles and then utilized that to predict content popularity in each \ac{bs}.

It is worth noting that \ac{fl} has many open problems (see \cite{kairouz2021advances} and the references therein).
Efficient client sampling has recently emerged as one of the key solutions to resource constraints \cite{cho202towards,chen2022optimal,wang2024a}.  
These works advocate for a subset of clients selection from the entire client pool in order to mitigate the straggler effect \cite{yang2021achieving}.
Owing to the limited resources in wireless communication, many recent works also suggest client scheduling based on the available resources of the clients and the \ac{bs} that acts as the \ac{cs} \cite{pervej2023Resource, du2023gradient,chen2023exploring,ni2024joint}.
These works, however, deal with general wireless networks and do not consider a wireless video caching platform.

\subsection{Our Contributions}

\noindent
While \cite{kang2023content, li2023deep, nguyen2021user, lin2023joint, fang2023drl, zhou2023recommendation, xiong2023reinforcement} showed \ac{ml} improves caching performance, these centralized \ac{ml} methods are not suitable due to privacy concerns.
Besides, although some studies \cite{qiao2022adaptive, wang2023popularity, jiang2021federated, li2023community, khanal2022route, lin2023feedback, maale2023deepfesl, cao2023mobility, feng2023proactive} only considered users' data privacy, collaboration among the three entities involved, i.e., \ac{ue}, \ac{bs} and \ac{csp}, was not addressed. 
Furthermore, the above works do not consider sporadic content request patterns and the important fact that the users' training datasets are not readily available: they need to process their requested content's information to prepare their training datasets.
Additionally, an \ac{hfl} based solution is preferable due to the hierarchical architecture of practical networks.
Moreover, while clients have limited computational power and energy resources, the \ac{isp} has limited radio resources.
As such, coordination among these parties is needed to ensure an efficient video caching platform.
Due to these reasons, we propose a novel \ac{rawhfl} algorithm where the clients train their local model to predict their to-be-requested content.
More specifically, our key contributions are summarized as follows:

\begin{itemize}
    \item Considering a realistic multi-cell wireless network consisting of multiple \acp{bs} and clients, where clients request content from a \ac{csp} using encrypted tagged IDs, we leverage privacy-preserving collaboration among these three parties to bring an efficient \ac{fl} solution for predicting users' to-be-requested content in wireless video caching networks. 
    \item Due to resource scarcity, the proposed \ac{rawhfl} only selects a subset of the clients for model training. 
    Furthermore, acknowledging the well-known system and data heterogeneity, we derive \ac{rawhfl}'s convergence bound, which reveals that the global gradient norm depends on the successful reception of the selected clients' trained accumulated gradients and their local training rounds.
    \item To that end, we jointly optimize client selection and clients' (i) local training rounds, (ii) \ac{cpu} cycles and (iii) transmission power to minimize a weighted utility function that facilitates \ac{rawhfl}'s convergence and minimizes energy expense under delay, energy and radio resource constraints.  
    As the original problem is non-convex, we transform it into a relaxed difference of convex programming problems and use a low-complexity iterative solution.
    \item Our extensive simulation results validate that the proposed solution outperforms existing baselines in terms of test accuracy and energy expense in a resource-constrained setting\footnote{Note that the predicted content information can be utilized to design content placement and evaluate other metrics, such as cache hit ratio, latency, revenue, etc., to name a few. 
    However, the evaluation of such metrics depends on system parameters and is beyond the scope of this paper.
    These metrics are generally directly proportional to the prediction accuracy.}.
    Moreover, the proposed \ac{rawhfl} performs nearly identically to the ideal case performance of \ac{hfedavg} \cite{liu2020client} in all examined scenarios. 
\end{itemize}

The rest of the paper is organized as follows. 
Section \ref{sysModel} presents our system model, content request model and preliminaries of \ac{hfl}.
In Section \ref{rawHFLconvergence}, we first summarize the proposed \ac{rawhfl} and then present our detailed theoretical analysis.
Besides, Section \ref{jointProblemFormulation} describes our joint problem formulation, followed by problem transformations and solutions.
Section \ref{simulationResults} presents our extensive simulation results, followed by the concluding remarks in Section \ref{conclusion}.
Moreover, Table \ref{tableofVariable} summarizes the important notations used in this paper.

\begin{table}[!t]
\caption{Summary of important variables}
\centering
\fontsize{8}{8} \selectfont 
\begin{tabular}{|C {1.3cm}|C{6.5cm}|}
\hline
\textbf{Parameter}  & \textbf{Definitions}  \\ \hline 
$u$, $\mathcal{U}$ & User $u$, all user set \\ \hline 
    $b$, $\mathcal{B}$ & base station $b$, all BS set \\ \hline 
    $l$, $\mathrm{L}$  & $l^{\mathrm{th}}$ SGD round, upper bound for local SGD round \\ \hline 
    $e$, $E$   & $e^{\mathrm{th}}$ edge round, total edge round  \\ \hline 
    $k$, $K$  & $k^{\mathrm{th}}$ global round, total global round \\ \hline 
    $t$ & $t^{th}$ discrete slot at which UE may request content \\\hline
    $\mathcal{U}_b$, $\mathcal{U}_b^{k,e}$ & BS $b$'s UE set; selected UE/client set of BS $b$ during edge round $e$ of global round $k$ \\ \hline 
    $\mathrm{L}_u^{k,e}$ & UE/client $u$'s local SGD round during edge round $e$ of global round $k$ \\ \hline 
    $g$, $G$ & Genre $g$; total genres \\ \hline
    $c_g$, $\mathcal{C}_g$, $\mathcal{C}$ & $c^{\mathrm{th}}$ content of genre $g$; all content set in genre $g$; entire content catalog \\ \hline
    $\Bar{C}_g$, $C$ & Total content in genre $g$; total content in the catalog \\ \hline
    $z$, $Z$, $\mathcal{Z}$ & $z^{\mathrm{th}}$ pRB; total pRBs; pRB set \\ \hline 
    $\mathcal{D}_{u,\mathrm{raw}}^0$ & UE $u$'s initial historical dataset \\ \hline
    $p_{u,\mathrm{ac}}$ & UE $u$'s probability of being active (making a content request) \\ \hline
    $\mathrm{1}_{u,c_g}^t$ & Binary indicator function that defines whether $u$ requests content $c_g$ during slot $t$ \\ \hline
    $p_{u,g}$ & UE $u$'s preference to genre $g$ \\ \hline
    $\Upsilon$ & Dirichlet distribution's concentration parameter for the genre preference \\ \hline
    $\mathcal{D}_{u,\mathrm{proc}}^{t}$, $\mathrm{D}_{u,\mathrm{proc}}^{t}$ & UE $u$'s processed dataset; total samples in UE $u$'s processed dataset \\ \hline   
    $f_u(\cdot)$; $f_b(\cdot)$, $f(\cdot)$ & UE $u$'s loss function; BS $b$'s loss function; global loss function \\ \hline
    $\mathbf{w}_u^{k,e,l}$, $\mathbf{w}_b^{k,e}$, $\mathbf{w}^{k}$ & UE/Client $u$'s local model during SGD round $l$ of edge round $e$ of global round $k$; BS $b$'s edge model during edge round $e$ of global round $k$; global model during round $k$ \\ \hline 
    $g_u(\mathbf{w}_{u}^{k,e,l})$ & UE $u$'s gradient during $l^{\mathrm{th}}$ local round of $e^{\mathrm{th}}$ edge round of $k^{\mathrm{th}}$ global round \\ \hline
    $\tilde{g}_u^{k,e}$ & UE $u$'s accumulated gradients during $e^{\mathrm{th}}$ edge round of $k^{\mathrm{th}}$ global round \\ \hline
    $\eta$ & Learning rate \\ \hline
    $\alpha_u$, $\alpha_b$ & UE $u$'s trained model/accumulated gradients' weight; BS $b$'s model/accumulated gradients' weight \\ \hline
    $\mathrm{1}_{u, \mathrm{sl}}^{k,e}$; $\mathrm{p}_{u, \mathrm{sl}}^{k,e}$ & Binary indicator function to define whether $u$ is selected in edge round $e$ of global round $k$; success probability of $\mathrm{1}_{u,\mathrm{sl}}^{k,e}$ \\ \hline
    $\mathrm{t}_{u,\mathrm{cp}}^{k,e}$; $\mathrm{e}_{u,\mathrm{cp}}^{k,e}$ & UE $u$'s local model training time and energy overheads during edge round $e$ of global round $k$ \\ \hline
    $\mathrm{t}_{u,\mathrm{up}}^{k,e}$; $\mathrm{e}_{u,\mathrm{up}}^{k,e}$ & UE $u$'s accumulated gradient offloading time and energy overheads during edge round $e$ of global round $k$ \\ \hline
    $\mathrm{t_{th}}$ & Deadline threshold to finish one edge round \\ \hline
    $P_{u,\mathrm{tx}}^{k,e}$; $P_{u, \mathrm{max}}$ & Transmission power of $u$; maximum allowable transmission power of $u$ \\ \hline 
    $\mathrm{e}_{u,\mathrm{bud}}$ & Energy budget of $u$ for each edge round \\ \hline
    $\mathrm{1}_{u,\mathrm{sc}}^{k,e}$; $\mathrm{p}_{u,\mathrm{sc}}^{k,e}$ & Binary indicator function to define whether accumulated gradient of $u$ is received by the \ac{bs}; success probability of $\mathrm{1}_{u,\mathrm{sc}}^{k,e}$ \\ \hline
    $\phi$; $s$ & Floating point precision; client's uplink payload size for the accumulated gradients \\ \hline
    $\beta$ & Smoothness of the loss function \\ \hline
    $\sigma^2$ & Bounded variance of the gradients \\ \hline 
    $\epsilon_{0}^2$ & Bounded divergence between the local and the edge loss functions \\ \hline
    $\epsilon_{1}^2$ & Bounded divergence between the edge and the global loss functions \\ \hline
    $\varphi_b^{k,e}$ & Utility function for the joint optimization problem \\ \hline
\end{tabular}
\label{tableofVariable}
\end{table}

\begin{figure}[!t]
    \centering
    \includegraphics[trim=72 10 120 10, clip,
    width=0.4\textwidth]{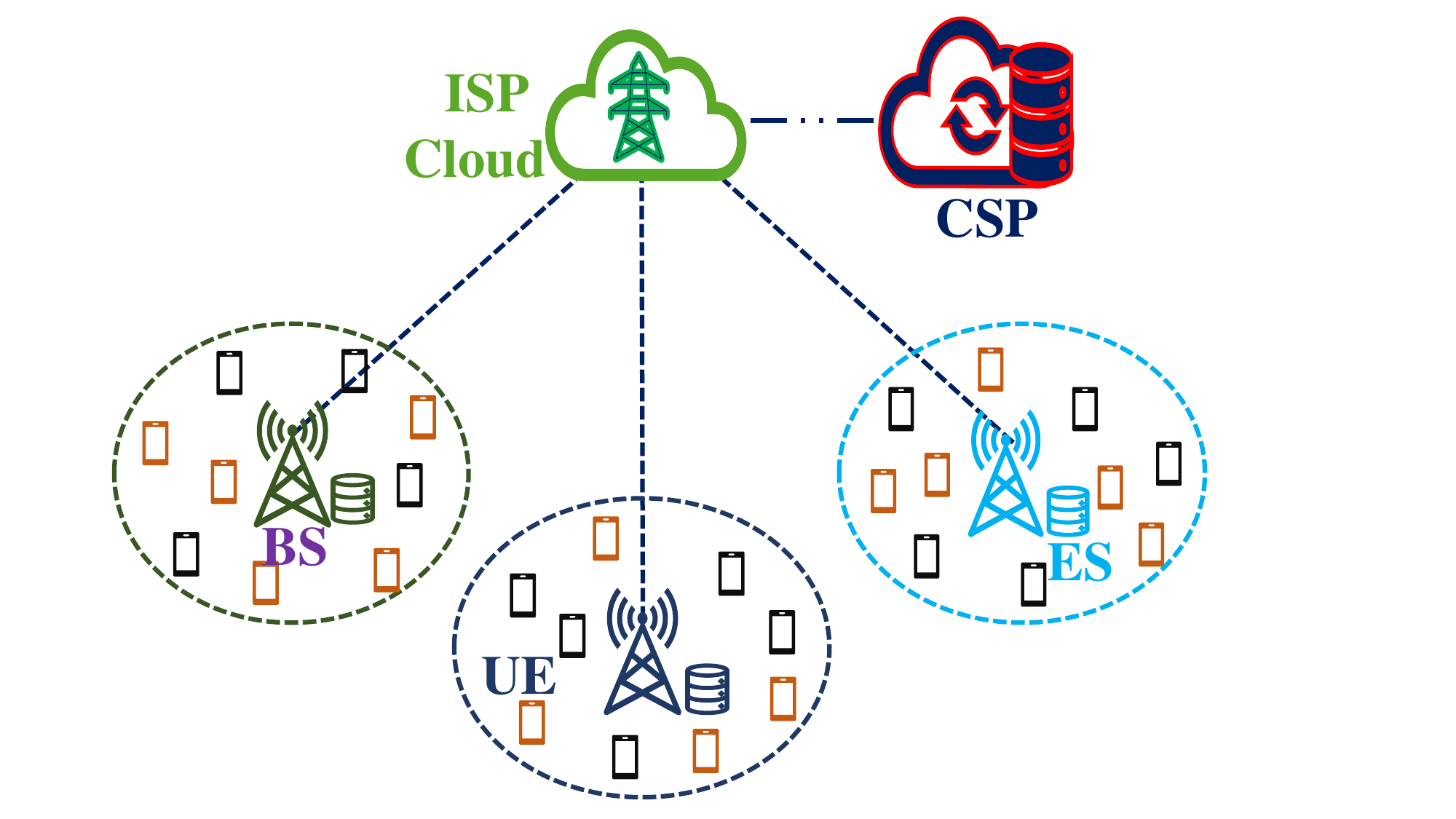}
    \caption{Network system model: \acp{ue} are connected to their serving \acp{bs}, and these \acp{bs} are connected to their \ac{isp} cloud network, while the \ac{csp} has an agreement with the \ac{isp} that allows placing one \ac{es} to each \ac{bs} } 
    \label{netSysMod}
\end{figure}

\section{System Model}
\label{sysModel}

\subsection{Network Model} 
\noindent
This paper considers a video caching wireless network, where distributed \acp{ue}, denoted by $\mathcal{U}=\{u\}_{u=0}^{U-1}$, can request content from a \ac{csp} using their serving \acp{bs}.
Let us denote the \ac{bs} set by $\mathcal{B}=\{b\}_{b=0}^{B-1}$ and the user set associated with BS $b$ by $\mathcal{U}_b$.
Besides, a \ac{ue} is only associated with a single \ac{bs} and $\mathcal{U} = \bigcup_{b=0}^{B-1} \mathcal{U}_b$.
For simplicity, we assume that all \acp{bs} are operated and controlled by the same \ac{isp}.
Denote the content catalog of the \ac{csp} by $\mathcal{C}=\{\mathcal{C}_g\}_{g=0}^{G-1}$, where $\mathcal{C}_g = \{c_g\}_{c=0}^{\bar{C}_g-1}$ represents the content set of genre $g$.
Therefore, the content catalog has a fixed set of $C = \sum_{g=0}^{G-1} \Bar{C}_g$ content\footnote{While new content may arrive/generate at the \ac{csp}, given that the \ac{csp} has a limited storage capacity, following standard practice \cite{maale2023deepfesl, cao2023mobility, feng2023proactive}, we assume the number of content that the \ac{csp} can store remains fixed.}.
The \ac{csp} and the \ac{isp} have mutual agreements that enable the \ac{csp} to install one \ac{es} to every \ac{bs} of the \ac{isp}. 
This lets the \ac{csp} leverage the \ac{es}' limited storage and computation power to devise an efficient \ac{fl}-based caching solution.
The considered system model is shown in Fig. \ref{netSysMod}. 
Furthermore, we assume that the \ac{csp} and the \ac{isp} are entirely different entities under two different business organizations and, therefore, always keep their operations secrets from each other.
As such, the \ac{csp} neither reveals the information of the \ac{ue}'s requested content nor shares what content it stores in the \acp{es} with the \acp{bs}.
The \acp{ue}, on the other hand, know the exact content information and can place their content requests using encrypted temporal tag IDs via their associated \acp{bs}.
Fig. \ref{privacyProtectionSysMod} shows the privacy protection procedures at different nodes.
Furthermore, the \ac{csp} can frequently change these encrypted tag IDs and their mapping with the original content to further enhance its operational secrets.

\begin{figure}[!t]
\centering
    \includegraphics[trim=50 8 20 8, clip, width=0.5\textwidth]{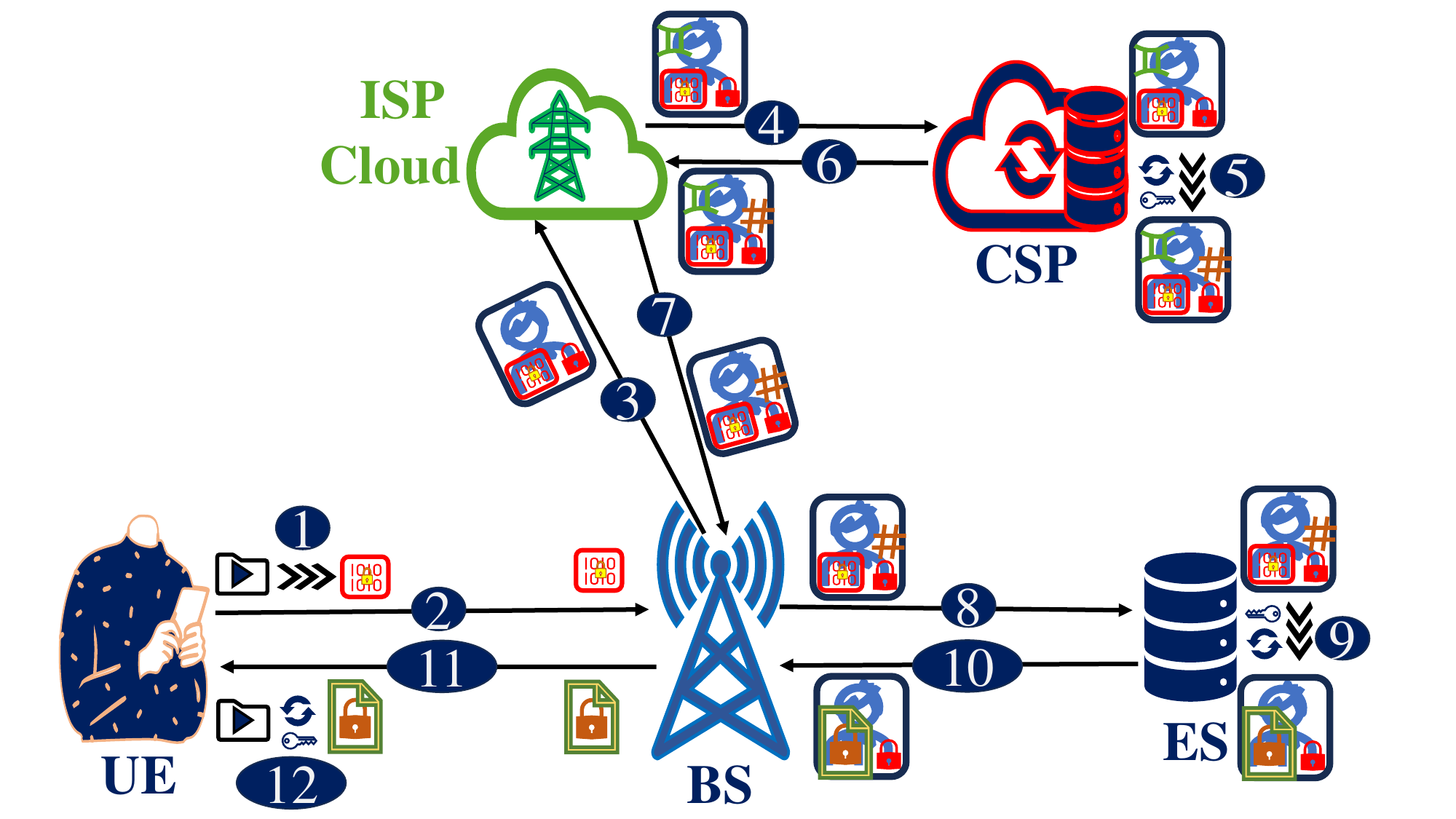}
    \caption{Privacy protection in different nodes: ($1$) \ac{ue} uses tag ID, ($2$) \ac{ue} sends encrypted content request to the serving \ac{bs}, ($3$) \ac{bs} sends the encrypted content request to its cloud for charging/authentication, ($4$) \ac{isp} cloud sends encrypted content request information to \ac{csp}, ($5$) \ac{csp} decode actual content ID from the tagged ID, ($6$) \ac{csp} sends encrypted information for the \ac{es} to the \ac{isp} cloud, ($7$) \ac{isp} sends the encrypted information to serving \ac{bs}, ($8$) serving \ac{bs} forwards encrypted information to the \ac{es}, ($9$) \ac{es} decode the \ac{csp}'s information, ($10$) \ac{es} sends encrypted video payload to serving \ac{bs}, ($11$) \ac{bs} sends the packet to the \ac{ue}, and ($12$) \ac{ue} decrypt the packet}
    \label{privacyProtectionSysMod}
\end{figure}

The network operates in discrete time slots, denoted by $t$, and the duration between two consecutive slots is $\kappa$ seconds. 
The \ac{isp} dedicates $\Bar{\mathrm{Z}}$ \ac{hz} bandwidth to perform the \ac{fl} training.
More specifically, each \ac{bs} of the \ac{isp} utilizes $\Bar{\mathrm{Z}}$ \ac{hz} bandwidth, which they divide into $Z$ orthogonal \acp{prb}, for sharing the \ac{fl}-related payloads with the participating \acp{ue} and the \acp{es}.
Denote the pRB set by $\mathcal{Z}=\{z\}_{z=0}^{Z-1}$.
We assume that node association, radio resource allocations and \ac{csi} are known at the \ac{bs}.
Besides, the neighboring \acp{bs} can utilize $\Bar{\mathrm{Z}}$ \ac{hz} bandwidth for the \ac{fl} task from different portions of their entire operating bandwidth to avoid inter-cell interference\footnote{If the \acp{bs} have to use the same portion of the bandwidth, they can coordinate and assign the \acp{prb} in neighboring cells and/or use advanced beamforming techniques such as zero-forcing for interference-free communication.} \cite[Chap. 21]{molisch2023wireless}.
Moreover, since orthogonal \acp{prb} are used within each cell, there is no intra-tier interference.
Furthermore, we consider that the wireless channel between the \ac{ue} and the \ac{bs} is dominated by large-scale fading since users are stationary in our system model, and practical networks can exploit enough diversity to mitigate small-scale fading.
As such, we model the large-scale path loss and log-Normal shadowing following $3$GPP's urban macro model \cite[Section $7.4$]{3GPPTR38_901}.

\subsection{Content Request Model and Dataset Acquisition}
\label{contReqModSection}


\subsubsection{Content Request Model}
\noindent 
We assume that each \ac{ue}'s content request model follows two steps. 
Firstly, a \ac{ue} may or may not request content during a slot $t$, depending on its need.
Let $p_{u,\mathrm{ac}}$ be the probability that the $u^{\mathrm{th}}$ \ac{ue} makes a content request, which can also be interpreted as the \textit{activity level} of the \ac{ue}.
In other words, the content request arrival is essentially modeled as a Bernoulli random variable with success probability $p_{u,\mathrm{ac}}$, which is widely used in the literature \cite{somuyiwa2018reinforcement,bharath2018caching,pervej2024efficient}.
Given that the \ac{ue} is \emph{active}, in the next step, we model \emph{which} content it wants to request. 
To model this, we consider a popularity-preference tradeoff owing to the fact that both individual preference and global content popularity can influence the \ac{ue} in choosing the content to request.
In particular, we assume that each \ac{ue} has a genre preference, denoted by $p_{u,g}$, such that $0\leq p_{u,g} \leq 1$ and $\sum_{g=0}^{G-1} p_{u,g}=1$.
These genre preferences are modeled following a symmetric Dirichlet distribution $\mathrm{Dir}(\pmb{\Upsilon})$, where the $\pmb{\Upsilon}$ is the concentration parameter \cite{pervej2023Resource}.

We assume that the \ac{ue} requests the most popular content $c_g$ from its preferred genre $g$ when it makes the first content request. 
Let $\mathrm{1}_{u,c_g}^{t}$ be a binary variable that takes value $1$ if the \ac{ue} requests content $c_g$ in slot $t$ and value $0$ otherwise. 
In the subsequent content request, the \ac{ue} can either exploit the most similar content\footnote{Each content has its own distinctive feature set that can be used to calculate its similarity with other content in the same genre.} from the same genre with probability $\upsilon_u$ or explore the most popular content from a different genre, $g' \neq g$, with probability $(1-\upsilon_u)$. 
It is worth noting that the genre preference $p_{u,g}$ does not influence $p_{u,\mathrm{ac}}$ in the above content request model\footnote{However, other content request strategies can also be easily incorporated, which will only change the clients' datasets. Our \ac{rawhfl} solution is general, and our theoretical analysis in the sequel will still hold.}.

\subsubsection{Dataset Acquisition Method}
We assume that each \ac{ue} has a small initial historical raw dataset, denoted by $\mathcal{D}_{u,\mathrm{raw}}^0$ that the client updates based on the information of its requested content.
Following the above content request model, for $t>0$, each \ac{ue} updates its local raw dataset using the requested content's information as 
\begin{equation}
\label{datasetEvol}
\begin{aligned}
    \mathcal{D}_{u,\mathrm{raw}}^t &\coloneqq 
    \begin{cases}
        \mathcal{D}_{u,\mathrm{raw}}^{t-1} \bigcup \big\{ \mathbf{x}(\mathrm{1}_{u,c_g}^t), \bblue{y(\mathrm{1}_{u,c_g}^t)} \big\}, \rs\rs\rs  & \text{if $u$ is active}, \\
        \mathcal{D}_{u,\mathrm{raw}}^{t-1}, & \text{otherwise},
    \end{cases}\rs,\rs\rs
\end{aligned}
\end{equation}
where $\mathbf{x}(\mathrm{1}_{u,c_g}^t)$ and \bblue{$y(\mathrm{1}_{u,c_g}^t)$} are the feature \bblue{vector and label index} of the requested content, respectively.
Note that (\ref{datasetEvol}) makes \ac{ue}'s dataset size time-varying and essentially captures natural data sensing methods in many practical applications \cite{pervej2022mobility, hosseinalipour2023parallel}. 
The \acp{ue} then process their local raw datasets acquired by (\ref{datasetEvol}) to train the \ac{ml} model locally.
Let us denote the processed dataset by $\mathcal{D}_{u,\mathrm{proc}}^t = \{\mathbf{x}_{u}^a, \bblue{y_u^a}\}_{a=0}^{\mathrm{D}_{u,\mathrm{proc}}^t - 1}$, where $(\mathbf{x}_{u}^a, \bblue{y_u^a})$ is the $a^{\mathrm{th}}$ (processed) training samples\footnote{The exact processing of the raw data samples shall depend on the application/simulation. More details are provided in the Section \ref{simulationResults}.}, and $\mathrm{D}_{u,\mathrm{proc}}^t$ is the total training samples.
Note that the probability that a \ac{ue} is active, i.e., $p_{u,\mathrm{ac}}$, directly governs its data acquisition policy in (\ref{datasetEvol}).
Particularly, the \ac{ue} receives a new training sample with probability $p_{u,\mathrm{ac}}$.
\bblue{It is worth noting that while a larger $p_{u,\mathrm{ac}}$ means the client has more training samples, the impact of $p_{u,\mathrm{ac}}$ may not be directly measured in the theoretical analysis since all clients select $\mathrm{n}$ mini-batches (discussed in the sequel) from their respective datasets to perform the local model training.}

\subsection{Hierarchical Federated Learning: Preliminaries}

\noindent
In this paper, we consider that the \ac{cs}, \acp{es}, \acp{bs} and the distributed clients collaboratively train an \ac{ml} model, parameterized by $\mathbf{w} \in \mathbb{R}^d$.
The task is to predict the to-be-requested content of these \acp{ue} so that the video caching network can make the most out of the limited storage of the \acp{es}\footnote{Our solution is general and can easily be extended for other tasks.}. 
We consider an \ac{hfl} \cite{liu2020client} framework consisting of two tiers: ($1$) client-\ac{es} and ($2$) \ac{es}-\ac{cs}.
The \acp{ue}/clients perform local rounds, while the \ac{es} and \ac{cs} conduct \emph{edge rounds} and \emph{global rounds}, respectively.   
In each \emph{local round}, the \acp{ue} train the \ac{ml} model using their processed datasets $\mathrm{D}_{u,\mathrm{proc}}^t$'s to minimize the following local loss function.
\begin{equation}
\label{localLossFunc}
\begin{aligned}
    f_u (\mathbf{w} | \mathcal{D}_{u,\mathrm{proc}}^t) \coloneqq [1/\mathrm{D}_{u,\mathrm{proc}}^t] \sum\nolimits_{(\mathbf{x}_{u}^a, \bblue{y_u^a}) \in \mathcal{D}_{u,\mathrm{proc}}^t} \mathrm{l}\left(\mathbf{w} | (\mathbf{x}_{u}^a, \bblue{y_u^a}) \right),
\end{aligned}
\end{equation}
where $\mathrm{l}\left(\mathbf{w} | (\mathbf{x}_{u}^a, \bblue{y_u^a})\right)$ is the loss associated with the $a^{\mathrm{th}}$ data sample.

Besides, each \ac{es} minimizes the following loss function in each \emph{edge round}.
\begin{equation}
\label{esLossAllClientsParticipate}
\begin{aligned}
    f_{b} (\mathbf{w} | \mathcal{D}_{b}^t) \coloneqq \sum\nolimits_{u \in \mathcal{U}_b} \alpha_u f_u (\mathbf{w} | \mathcal{D}_{u,\mathrm{proc}}^t),
\end{aligned}
\end{equation}
where $\alpha_u$ is the weight of the $u^{\mathrm{th}}$ client in \ac{bs} $b$ and $\mathcal{D}_b^t \coloneqq \bigcup_{u \in \mathcal{U}_b} \mathcal{D}_{u,\mathrm{proc}}^t$.
Note that we use the same notation $b$ in the subscript to represent the \ac{es} of the $b^{\mathrm{th}}$ \ac{bs} for brevity. 
Furthermore, the \ac{cs} wants to minimize the following global loss function in each \emph{global round}.
\begin{equation}
\label{centralLossAllClietnsParticipate}
\begin{aligned}
    f(\mathbf{w} | \mathcal{D}^t) \coloneqq \sum\nolimits_{b=0}^{B-1} \alpha_b f_b (\mathbf{w} | \mathcal{D}_{b}^t), 
\end{aligned}
\end{equation}
where $\alpha_b$ is the weight of the \ac{es} of the $b^{\mathrm{th}}$ \ac{bs} at the \ac{cs} and $\mathcal{D}^t \coloneqq \bigcup_{b=0}^{B-1} \mathcal{D}_b^t$.
Recall that clients' datasets are not stationary and keep evolving following (\ref{datasetEvol}).
Therefore, the global optimal model $\mathbf{w}^{*}$ may not necessarily remain stationary \cite{hosseinalipour2023parallel}.
The \ac{cs} thus wants to find the optimal $\mathbf{w}^{*}$ in each $t$ as follows.
\begin{equation}
\begin{aligned}
    \mathbf{w}^{*} &= \underset{\mathbf{w} }{\text{arg min}} \quad f(\mathbf{w} | \mathcal{D}^t), \quad \forall t.
\end{aligned}
\end{equation}

\section{Resource-Aware Hierarchical Federated Learning: Convergence Analysis}
\label{rawHFLconvergence}
\noindent

\begin{figure}[!t]
    \centering
    \includegraphics[trim=175 260 150 130, clip, width=0.5\textwidth]{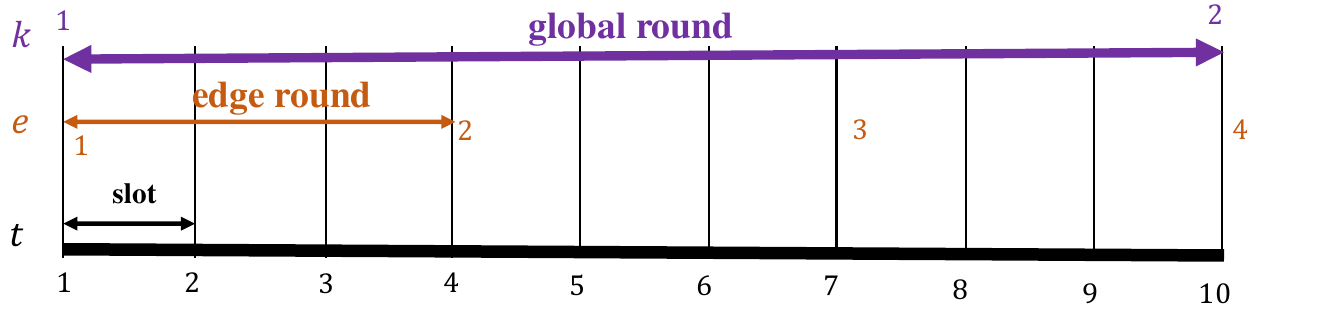}
    \caption{Time flow in the proposed system model: each edge round can have multiple content request slots, while local training happens in between two consecutive edge rounds, and each global round consists of multiple edge rounds}
    \label{timeFlow}
\end{figure}

\subsection{Resource-Aware Hierarchical Federated Learning Model}

\noindent
The proposed \ac{rawhfl} is designed on top of the general \ac{hfl} \cite{liu2020client,pervej2024hierarchical,wang2022demystifying} framework under explicit consideration of the available computational, radio and energy resources.   
In the proposed \ac{rawhfl}, the clients perform local training that consists of taking multiple local \ac{sgd} steps within an edge round. 
Similarly, the \acp{es} perform $E>1$ edge rounds in each global round. 
Fig. \ref{timeFlow} describes the flow of the local, edge, and global rounds. 
Besides, we consider a fully synchronous updating procedure at the upper tiers, i.e., in the \ac{es} and the \ac{cs}, which leads to a deadline-constrained case.
More specifically, this necessitates receiving the trained model parameters of the client at the \ac{es} during the \ac{es}' model update.
Our proposed \ac{rawhfl} has $6$ key steps, which are summarized below.

\subsubsection*{\textbf{Step $1$ - Global Round $k$ Initialization}}
At the beginning of each global round $k$, the \ac{cs} broadcasts the latest available global model, denoted by $\mathbf{w}^k$, to all \acp{es} via the \ac{isp}'s infrastructure.

\subsubsection*{\textbf{Step $2$ - Edge Round $e$ Initialization}}
Let us denote the edge round by $e$ and the edge model of the $b^{\mathrm{th}}$ \ac{es} by $\mathbf{w}_b^{k,e}$. 
During the first edge round, i.e., $e=0$, of every global round $k$, each \ac{es} synchronizes its edge model with the latest model received from the \ac{cs}, i.e., $\mathbf{w}_b^{k,e} \gets \mathbf{w}^k$, $\forall b \in \mathcal{B}$.
Due to limited resources, \ac{rawhfl} considers partial client participation by only selecting $\bar{\mathcal{U}}_b^{k,e} \subseteq \mathcal{U}_b$ clients for each \ac{es} (of each \ac{bs}) in all edge rounds.
It is worth noting that we assume the \acp{es} and the \acp{bs} collaborate to optimize these subset client selections and selected clients' ($1$) local training rounds and ($2$) transmit powers jointly in the sequel. 
Let $\mathrm{1}_{u,\mathrm{sl}}^{k,e}$ denote client's participation during $e^{\mathrm{th}}$ edge round of the $k^{\mathrm{th}}$ global round, which is defined as
\begin{equation}
    \begin{aligned}
        \mathrm{1}_{u,\mathrm{sl}}^{k,e} &= 
        \begin{cases}
            1, &\text{with probability } \mathrm{p}_{u, \mathrm{sl}}^{k,e},\\
            0, &\text{otherwise},
        \end{cases}.
    \end{aligned}
\end{equation}
It is worth noting that the objective of each \ac{es} in \ac{rawhfl} differs from (\ref{esLossAllClientsParticipate}) when $\bar{\mathcal{U}}_b^{k,e} \subset \mathcal{U}_b$.
More specifically, each \ac{es} minimizes the following loss in \ac{rawhfl}.
\begin{equation}
\label{esLossRawHFL}
    f_{b} (\mathbf{w} | \cup_{u \in \bar{\mathcal{U}}_b^{k,e}} \mathcal{D}_{u, \mathrm{proc}}^t ) \coloneqq \sum\nolimits_{u \in \bar{\mathcal{U}}_b^{k,e}} \alpha_u f_u (\mathbf{w} | \mathcal{D}_{u,\mathrm{proc}}^t),
\end{equation}
where $\sum_{u \in \bar{\mathcal{U}}_b^{k,e}} \alpha_u = 1$ and $t$ is the time slot at which the $e^{\mathrm{th}}$ edge round starts.
The \ac{es} then forwards its current local model to its \ac{bs}, which broadcasts\footnote{Following existing literature \cite{pervej2024hierarchical,pervej2023Resource, hosseinalipour2023parallel}, downlink communication overheads are ignored in this paper.} the model to the selected clients.  

\subsubsection*{\textbf{Step 3 - Local Model Training}}
Firstly, each $u \in \bar{\mathcal{U}}_b^{k,e}$ synchronizes its local model as
\begin{equation}
    \mathbf{w}_u^{k,e,0} \gets \mathbf{w}_b^{k,e}.
\end{equation}
The client then trains its local model $\mathbf{w}_u^{k,e,0}$ for $\mathrm{L}_{u}^{k,e}$ \ac{sgd} rounds to minimize the loss function defined in (\ref{localLossFunc}).
As such, the local updated model is 
\begin{equation}
\label{localModUp}
\begin{aligned}
    \mathbf{w}_{u}^{k,e,\mathrm{L}_{u}^{k,e}} 
    &= \mathbf{w}_{u}^{k, e, 0} - \eta \sum\nolimits_{l=0}^{\mathrm{L}_{u}^{k,e}-1} g_u(\mathbf{w}_{u}^{k,e,l}),
\end{aligned}
\end{equation} 
where $\eta$ is the learning rate and $g_u(\mathbf{w}_{u}^{k,e,l})$ is the stochastic gradient.
Note that $1 \leq \mathrm{L}_{u}^{k,e} \leq \mathrm{L}$, where $\mathrm{L}$ is the maximum local \ac{sgd} rounds.
Besides, $\mathrm{L}_u^{k,e}$ can be different for different clients because each client has a deadline of $\mathrm{t_{th}}$ seconds and an energy budget of $\mathrm{e}_{u,\mathrm{bd}}$ Joules during each edge round.

Note that, in each \ac{sgd} step, we let each client randomly sample $\mathrm{n}$ mini-batches, which yields the following local computation time overhead 
\begin{equation}
\label{localCompTime}
    \mathrm{t}_{u,\mathrm{cp}}^{k,e} = \mathrm{L}_{u}^{k,e} \times \mathrm{n} \bar{\mathrm{n}} \mathrm{c}_u \mathrm{D}_u / f_u^{k,e},
\end{equation}
where $\bar{\mathrm{n}}$, $\mathrm{c}_u$, $\mathrm{D}_u$ and $f_u^{k,e}$ are the batch size, number of \ac{cpu} cycles to compute $1$-bit data, data sample size in bits, and the \ac{cpu} frequency.
The corresponding energy expense is calculated as \cite{qiao2022adaptive,pervej2024hierarchical,hosseinalipour2023parallel}
\begin{equation}
\begin{aligned}
    \mathrm{e}_{u,\mathrm{cp}}^{k,e} 
    &= \mathrm{L}_{u}^{k,e} \times 0.5\zeta \mathrm{n} \bar{\mathrm{n}} \mathrm{c}_u \mathrm{D}_u (f_u^{k,e})^2,
\end{aligned}
\end{equation}
where $0.5\zeta$ is the effective capacitance of the CPU chip.

\subsubsection*{\textbf{Step $4$ - Trained Accumulated Gradients Offloading}}
Once the local training is finished, each client sends its accumulated gradients $\tilde{g}_u^{k,e} \coloneqq \sum_{l=0}^{\mathrm{L}_u^{k,e}-1} g_u(\mathbf{w}_{u}^{k,e,l})$ to its associated \ac{bs} as an encrypted payload.
This uplink communication causes time and energy overheads for the clients.
The required time to complete the uplink transmission is calculated as
\begin{equation}
\label{gradientOffload}
    \mathrm{t}_{u,\mathrm{up}}^{k,e} = \mathrm{s}/ \big(\omega \log_2 \big[1 + \gamma_{u}^{k,e}\big] \big),
\end{equation}
where $\mathrm{s} = d \times (\phi+1)$ bits \cite{pervej2023Resource} is the payload size, $\phi$ is the \ac{fpp}, $\omega$ is the bandwidth of a \ac{prb} and $\gamma_u^{k,e}$ is the \ac{snr}, which is derived as\footnote{The \emph{small-scale} fading channel $h_u^{k,e}$ can also be accommodated in the \ac{snr} calculation by modifying (\ref{snr}) as $\gamma_{u}^{k,e} = \beta_{u}^{k,e} \zeta_u^{k,e} P_{u,\mathrm{tx}}^{k,e} \left\vert h_u^{k,e} \right\vert^2 / (\omega \varsigma^2)$. 
Moreover, when the \ac{csi} is imperfect, the UE can use a lower \ac{mcs} with proper \ac{arq} that may scale down the data rate.}
\begin{equation}
\label{snr}
    \gamma_{u}^{k,e} = \beta_{u}^{k,e} \zeta_u^{k,e} P_{u,\mathrm{tx}}^{k,e} / (\omega \varsigma^2), 
\end{equation}
where $P_{u,\mathrm{tx}}^{k,e}$ is the uplink transmission power, $\beta_u^{k,e}$ is the large-scale path loss, $\zeta_u^{k,e}$ is the log-Normal shadowing and $\varsigma^2$ is the noise variance. 
Besides, we calculate the energy expense for the client's uplink transmission as
\begin{equation}
    \mathrm{e}_{u,\mathrm{up}}^{k,e} = \mathrm{s} P_{u,\mathrm{tx}}^{k,e} / \big(\omega \log_2 \big[1 + \gamma_{u}^{k,e}\big]\big).
\end{equation}

\subsubsection*{\textbf{Step $5$ - Edge Round Completion and Aggregation}}
Once the \ac{bs} receives client's $\tilde{g}_u^{k,e}$, it forwards the payload to its \ac{es}. 
Note that the overheads in forwarding the received $\tilde{g}_u^{k,e}$ to the \ac{es} are ignored since the \ac{es} is embedded in the same \ac{bs}.
When the deadline $\mathrm{t_{th}}$ expires, the \ac{es} aggregates the received gradients and updates its edge model as 
\begin{equation}
\label{edgeUpdateRule}
\begin{aligned}
    \mathbf{w}_b^{k, e+1} 
    &= \mathbf{w}_{b}^{k,e} - \eta \sum\nolimits_{u \in \bar{\mathcal{U}}_b^{k,e}} \alpha_u \big[\mathrm{1}_{u,\mathrm{sc}}^{k,e}/\mathrm{p}_{u,\mathrm{sc}}^{k,e}\big]
    \tilde{g}_u^{k,e},
\end{aligned}
\end{equation}
where $\alpha_u \coloneqq 1/|\bar{\mathcal{U}}_b^{k,e}|$ and $\mathrm{1}_{u,\mathrm{sc}}^{k,e}$ is defined as 
\begin{equation}
\begin{aligned}
    \mathrm{1}_{u,\mathrm{sc}}^{k,e} &= 
    \begin{cases}
        1, &\text{with probability } \mathrm{p}_{u, \mathrm{sc}}^{k,e},\\
        0, &\text{otherwise},
    \end{cases}.
\end{aligned}
\end{equation}
We use $\mathrm{1}_{u,\mathrm{sc}}^{k,e}$ to capture whether $\tilde{g}_u^{k,e}$ is received by the \ac{bs} successfully within the allowable deadline.
It is worth noting that the client selection strategy can allow us to select similar computationally capable clients who can perform nearly identical local \ac{sgd} rounds.
Besides, we consider mini-batch \ac{sgd}, and each client selects $\mathrm{n}$ mini-batches.
As such, putting equal weight on the client's accumulated gradients is practical\footnote{Different weighting strategies can also be easily incorporated.}.

To that end, each \ac{es} checks whether they have completed $E$ edge rounds, i.e., if $(e+1)=E$. 
If the \acp{es} do not complete $E$ edge rounds, they repeat \emph{\textbf{Step} $2$} to \emph{\textbf{Step} $5$}. 
Upon completing $E$ edge rounds, each \ac{es} forwards its updated edge model to the \ac{cs}\footnote{The transmission between the \ac{es} and the \ac{cs} happens in the network's backhaul, and the corresponding time and energy overheads incur on the \ac{es} and the \ac{bs}. 
We kept our focus on the overheads for the clients and thus ignored these overheads.} using the \ac{isp}'s infrastructure.

\subsubsection*{\textbf{Step $6$ - Global Round Completion and Aggregation}}
Upon receiving the updated edge models, the \ac{cs} performs global aggregation and updates the global model as 
\begin{equation}
\label{globalUpdateRule}
\begin{aligned}
    \rs\rs \rs \rs \mathbf{w}^{k+1} 
    \rs= \mathbf{w}^k \rs - \eta \sum\nolimits_{e=0}^{E-1} \sum\nolimits_{b=0}^{B-1} \rs \alpha_b \sum\nolimits_{u \in \bar{\mathcal{U}}_b^{k,e}} \alpha_u \big[\mathrm{1}_{u,\mathrm{sc}}^{k,e} / \mathrm{p}_{u,\mathrm{sc}}^{k,e}\big] \tilde{g}_u^{k,e} \rs. \rs\rs\rs\rs
\end{aligned}   
\end{equation}
It is worth noting that the \ac{cs} minimizes the following global loss function in our proposed \ac{rawhfl} algorithm.
\begin{equation}
\label{centralLossRawHFL}
    f(\mathbf{w}^k | \cup_{b=0}^{B-1} \cup_{u \in \bar{\mathcal{U}}_b^{k,e}} \mathcal{D}_{u,\mathrm{proc}}^t) \coloneqq \sum\nolimits_{b=0}^{B-1} \alpha_b f_{b} \big(\mathbf{w}^k | \cup_{u \in \bar{\mathcal{U}}_b^{k,e}} \mathcal{D}_{u, \mathrm{proc}}^t \big).
\end{equation}

\emph{\textbf{Step 1}} to \emph{\textbf{Step 6}} are repeated for $K$ global rounds.
These steps are summarized in Algorithm \ref{rawHFLAlg}.
Note that from hereon onward, we have used $f_u(\mathbf{w})$, $f_b(\mathbf{w})$ and $f(\mathbf{w})$ to represent $f_u(\mathbf{w}| \mathcal{D}_{u,\mathrm{proc}}^t)$, $f_b(\mathbf{w} | \cup_{u \in \bar{\mathcal{U}}_b^{k,e}} \mathcal{D}_{u,\mathrm{proc}}^t)$ and $f(\mathbf{w} | \cup_{b=0}^{B-1} \cup_{u \in \bar{\mathcal{U}}_b^{k,e}} \mathcal{D}_{u,\mathrm{proc}}^t)$, respectively, for notational simplicity.

\begin{algorithm}[!t]
\small
\SetAlgoLined 
\DontPrintSemicolon
\KwIn{Global model: $\mathbf{w}^0$; total global round $K$, number of edge rounds $E$ }
\For{$k=0$ to $K-1$}{
    \For {$b=1$ to $B$ in parallel}{
        Receives updated global model $\mathbf{w}_b^{k,0} \gets \mathbf{w}^k$ \;
        \For {$e=0$ to $E-1$} {
            BS $b$ receives optimized $\mathrm{1}_{u, \mathrm{sl}}^{k,e}, \mathrm{L}_{u}^{k,e}$ and $f_{u}^{k,e}$ \tcp*{\textit{c.f.} (\ref{originalOptimProb})}
            \For{$u$ in $\bar{\mathcal{U}}_b^{k,e}$ in parallel}{
                Get updated edge model $\mathbf{w}_{u}^{k,e,0} \gets \mathbf{w}_b^{k,e}$ \;
                Perform $\mathrm{L}_u^{k,e}$ mini-batch SGD rounds and get updated local model based on (\ref{localModUp}) \;
                Offload accumulated gradients $\tilde{g}_u^{k,e}$ to the \ac{bs}\;
            }
            Update edge model $\mathbf{w}_b^{k, e+1}$ using update rule in (\ref{edgeUpdateRule})
        }
    }
    Update global model $\mathbf{w}^{k+1}$ based on the update rule in (\ref{globalUpdateRule}) \;
}
\KwOut{Trained global model $\mathbf{w}^K$}
\caption{RawHFL Algorithm}
\label{rawHFLAlg}
\end{algorithm}


\subsection{Convergence of RawHFL}
\label{convAnalysis_Section}

\noindent
We make the following assumptions for our theoretical analysis, which are standard in the literature \cite{ pervej2024hierarchical, wang2022demystifying, hosseinalipour2023parallel,pervej2023Resource}.
\begin{Assumption}
\label{betaSmoothAssump}
    $\beta$-smoothness: The loss functions in all nodes are $\beta$-smooth, i.e., $\Vert \nabla f(\mathbf{w}) - \nabla f(\mathbf{w}')\Vert \leq \beta \Vert \mathbf{w} - \mathbf{w}' \Vert$, where $\Vert \cdot \Vert$ is the $L_2$ norm. 
\end{Assumption}
\begin{Assumption}
    Unbiased \ac{sgd}: mini-batch gradients are unbiased, i.e., $\mathbb{E}_{\xi \sim \mathcal{D}_{u,\mathrm{proc}}} \left[ g_u(\mathbf{w}) \right] = \nabla f_u(\mathbf{w})$, where $\mathbb{E} [\cdot]$ is the expectation operator, and $\xi$ is client's randomly sampled mini-batch. 
\end{Assumption}
\begin{Assumption}
    Bounded variance: variance of the gradients is bounded, i.e., $\Vert g_u(\mathbf{w}) - \nabla f_u(\mathbf{w}) \Vert^2 \leq \sigma^2$.
\end{Assumption}
\begin{Assumption}
    Independence: a) the \textit{stochastic gradients} are independent of each other in different episodes and b) \textit{accumulated gradient offloading} is independent of client selection and each other in each edge round $e$. 
\end{Assumption}
\begin{Assumption}
    Bounded divergence: divergence between the a) local and edge and b) edge and global loss functions are bounded, i.e., for all $u$, $b$ and $\mathbf{w}$
    \begin{align}
        &\sum\nolimits_{u \in \bar{\mathcal{U}}_b^{k,e}} \alpha_u \Vert \nabla f_u(\mathbf{w}) - \nabla f_b(\mathbf{w}) \Vert^2 \leq \epsilon_0^2,\\
        &\sum_{b=0}^{B-1} \alpha_b \Big\Vert \sum_{u \in \bar{\mathcal{U}}_b^{k,e}} \rs \rs \alpha_u \nabla \tilde{f}_{u} (\mathbf{w}) - \sum_{b'=0}^{B-1} \alpha_{b'} \rs \rs \sum_{u' \in \bar{\mathcal{U}}_{b'}^{k,e}} \rs\rs \alpha_{u'} \nabla \tilde{f}_{u'} (\mathbf{w}) \Big\Vert^2 \leq \epsilon_1^2, \label{ES_Central_Loss_Divergence}
    \end{align}
    where $\nabla \tilde{f}_u (\mathbf{w}) \coloneqq \sum_{l=0}^{\mathrm{L}_u^{k,e} - 1} \nabla f_u (\mathbf{w})$.
\end{Assumption}

Note that our proposed \ac{rawhfl} algorithm acknowledges data and system heterogeneity. 
While partial client selection may handle the system heterogeneity, clients' data heterogeneity still exists\footnote{Note that tackling heterogeneous data distributions requires special measures (see \cite{zhu2021datafree} and the references therein), which are beyond the scope of this paper.}.
Owing to these two types of heterogeneity, we derive the converge bound of the proposed \ac{rawhfl} algorithm below.


\begin{Theorem}
\label{theorem1}
Suppose $\eta < \mathrm{min}\left\{\frac{1}{2\sqrt{5} \beta \mathrm{L}}, \frac{1}{\beta E \mathrm{L}} \right\}$ and the above assumptions hold.
Then, the average global gradient norm from $K$ global rounds of \ac{rawhfl} is upper-bounded as
\begin{align}
\label{theorem1_eqn}
    & \frac{1}{K} \sum_{k=0}^{K-1} \mathbb{E} \big[\Vert \nabla f(\mathbf{w}^k) \Vert^2 \big]
    \leq \frac{2}{\eta K} \sum_{k=0}^{K-1} \frac{1}{\Omega^k} \Big\{ \mathbb{E} [ f(\mathbf{w}^k) ] - \mathbb{E} [ f(\mathbf{w}^{k+1}) ] \Big\} \nonumber\\
    & + \frac{2 \beta \eta \mathrm{L} \sigma^2}{K} \sum_{k=0}^{K-1} \frac{\mathrm{N}_1^k}{\Omega^k} + \frac{18 E \beta^2 \epsilon_0^2 \eta^2 \mathrm{L}^3}{K} \sum_{k=0}^{K-1} \frac{\mathrm{N}_2}{\Omega^k} + \nonumber\\
    &\frac{20 \mathrm{L} \beta^2 \epsilon_1^2 \eta^2 E^3}{K} \sum_{k=0}^{K-1} \frac{1}{\Omega^k} + \frac{2 \beta \eta \mathrm{L}}{K} \sum_{k=0}^{K-1} \frac{1}{\Omega^k} \sum_{e=0}^{E-1} \sum_{b=0}^{B-1} \alpha_b \times \nonumber\\
    & \rs \rs \sum\nolimits_{u \in \bar{\mathcal{U}}_b^{k,e}} \alpha_u \mathrm{N}_u  \big[(1/\mathrm{p}_{u,\mathrm{sc}}^{k,e}) - 1 \big] \sum\nolimits_{l=0}^{\mathrm{L}_u^{k,e} - 1} \mathbb{E} \big[\big\Vert g_u (\mathbf{w}_u^{k, e, l}) \big\Vert^2 \big],
\end{align}
where the expectations depend on clients' randomly selected mini-batches and $\mathrm{1}_{u,\mathrm{sc}}^{k,e}$'s. 
Besides, $\Omega^k \coloneqq \sum_{e=0}^{E-1} \sum_{b=0}^{B-1} \alpha_b \sum_{u \in \bar{\mathcal{U}}_b^{k,e}} \alpha_u \mathrm{L}_{u}^{k,e}$, $\mathrm{N}_1^k \coloneqq  60 \beta^3 \eta^3 E^3 \mathrm{L}^3 + 3 \beta \eta E \mathrm{L} +  \sum_{e=0}^{E-1} \sum_{b=0}^{B-1} \alpha_b  \left(\alpha_b + 4 E \mathrm{L} \beta \eta \right) \sum_{u \in \bar{\mathcal{U}}_b^{k,e}} \left(\alpha_u\right)^2$, $\mathrm{N}_2 \coloneqq 1 + 20 \beta^2 \eta^2 E^2 \mathrm{L}^2$ and $\mathrm{N}_{u} \coloneqq  E + 3 \beta \eta \mathrm{L} + 4 \beta \eta E \left(\alpha_u + 15 E \beta^2 \eta^2 \mathrm{L}^3 \right)$.
\end{Theorem}

\begin{proof}[\textbf{Sketch of Proof}]
Using the aggregation rule defined in (\ref{globalUpdateRule}) and $\beta$-smoothness assumption, we start with 
\begin{align}
    &f(\mathbf{w}^{k+1}) 
    \leq f(\mathbf{w}^k) - \nonumber \\
    & \eta \bigg<\nabla f(\mathbf{w}^k),  \sum_{e=0}^{E-1} \sum_{b=0}^{B-1} \alpha_b\sum_{u \in \bar{\mathcal{U}}_b^{k,e}} \alpha_u \frac{\mathrm{1}_{u,\mathrm{sc}}^{k,e}}{\mathrm{p}_{u,\mathrm{sc}}^{k,e}} \sum_{l=0}^{\mathrm{L}_{u}^{k,e}-1} g_u\big(\mathbf{w}_u^{k, e,l}\big) \bigg> + \nonumber\\
    & \frac{\beta \eta^2}{2} \bigg\Vert \sum_{e=0}^{E-1} \sum_{b=0}^{B-1} \alpha_b\sum_{u \in \bar{\mathcal{U}}_b^{k,e}} \alpha_u \frac{\mathrm{1}_{u,\mathrm{sc}}^{k,e}}{\mathrm{p}_{u,\mathrm{sc}}^{k,e}} \tilde{g}_u^{k,e} \bigg\Vert^2,
\end{align}
where $\left< \mathbf{a}, \mathbf{b} \right>$ is the inner product of the vectors $\mathbf{a}$ and $\mathbf{b}$.

Then, we derive the upper bounds of the second and third terms. 
Plugging these results and assuming $\eta \leq \frac{1}{\beta E \mathrm{L}}$, we get the following after doing some algebraic manipulations.
\begin{align}
\label{theorem1_Mid_Bound}
    &\frac{1}{K} \sum_{k=0}^{K-1} \mathbb{E} \left[\left\Vert \nabla f(\mathbf{w}^k)\right\Vert^2 \right]
    \leq \frac{2}{\eta K} \sum_{k=0}^{K-1} \left[ \frac{\mathbb{E} [f(\mathbf{w}^k)] - \mathbb{E} \left[ f(\mathbf{w}^{k+1}) \right] }{\Omega^k} \right] \nonumber\\
    &+\frac{2 \beta \eta \sigma^2}{K} \sum_{k=0}^{K-1} \left[\frac{\sum_{e=0}^{E-1} \sum_{b=0}^{B-1} \left(\alpha_b\right)^2 \sum_{u \in \bar{\mathcal{U}}_b^{k,e}} \left(\alpha_u \right)^2 \mathrm{L}_u^{k,e} }{\Omega^k} \right] + \nonumber\\
    & \frac{2 \beta \eta E}{K} \sum_{k=0}^K \frac{1}{\Omega^k} \sum_{e=0}^{E-1} \sum_{b=0}^{B-1} \alpha_b \rs\rs \sum_{u \in \bar{\mathcal{U}}_b^{k,e}} \rs\rs\rs \alpha_u \mathrm{L}_u^{k,e} \sum_{l=0}^{\mathrm{L}_u^{k,e} - 1} \rs \bigg[\frac{1} {\mathrm{p}_{u,\mathrm{sc}}^{k,e}} - 1 \bigg] \times \nonumber\\
    &\mathbb{E} \big[\big\Vert g_u (\mathbf{w}_u^{k, e, l}) \big\Vert^2 \big] +  \frac{2 \mathrm{L} \beta^2}{K} \sum_{k=0}^{K-1} \frac{1}{\Omega^k} \sum_{e=0}^{E-1} \sum_{b=0}^{B-1} \alpha_b \mathbb{E} \Big[ \big\Vert \mathbf{w}^k - \mathbf{w}_{b}^{k,e} \big\Vert^2 \Big]  \nonumber\\
    &+ \frac{2 \beta^2}{K} \sum_{k=0}^{K-1} \frac{1}{\Omega^k} \sum_{e=0}^{E-1} \sum_{b=0}^{B-1} \rs \alpha_b \rs \rs \rs \sum_{u \in \bar{\mathcal{U}}_b^{k,e}} \rs \rs \rs \alpha_u \rs \sum_{l=0}^{\mathrm{L}_u^{k,e} - 1} \rs \mathbb{E} \Big[ \big\Vert \mathbf{w}_{b}^{k,e} - \mathbf{w}_{u}^{k,e,l} \big\Vert^2 \Big]. \rs
\end{align}

When learning rate $\eta \leq \mathrm{min} \left\{\frac{1}{\beta E \mathrm{L}}, \frac{1}{3\sqrt{2} \beta \mathrm{L}} \right\}$, we can derive the difference between the client's model and the \ac{es}'s model in (\ref{theorem1_Mid_Bound}) as
\begin{align}
\label{lemama1_Eqn}
    &\frac{1}{K} \sum_{k=0}^{K-1} \frac{1}{\Omega^k} \sum_{e=0}^{E-1} \sum_{b=0}^{B-1} \alpha_b\sum_{u \in \bar{\mathcal{U}}_b^{k,e}} \alpha_u \sum_{l=0}^{\mathrm{L}_u^{k,e} - 1} \mathbb{E} \left[ \left\Vert \mathbf{w}_{b}^{k,e} - \mathbf{w}_{u}^{k,e,l} \right\Vert^2 \right] \nonumber\\
    &\leq \frac{3 E \mathrm{L}^2 \eta^2 \sigma^2}{K} \rs \sum_{k=0}^{K-1} \rs \frac{1}{\Omega^k} + \frac{9 E \epsilon_0^2 \eta^2 \mathrm{L}^3}{K} \rs \sum_{k=0}^{K-1} \rs \frac{1}{\Omega^k} + \frac{3 \mathrm{L}^2 \eta^2}{K} \rs \sum_{k=0}^{K-1} \rs \frac{1}{\Omega^k} \times \nonumber\\
    &\sum_{e=0}^{E-1} \sum_{b=0}^{B-1} \rs \alpha_b \rs \rs \sum_{u \in \bar{\mathcal{U}}_b^{k,e}} \rs\rs  \alpha_u \rs \bigg[\frac{1}{\mathrm{p}_{u,\mathrm{sc}}^{k,e}} - 1 \bigg]\sum_{l=0}^{\mathrm{L}_u^{k,e} - 1} \mathbb{E} \left[ \left\Vert g_u (\mathbf{w}_u^{k, e, l})  \right\Vert^2 \right].
\end{align}
Similarly, when $\eta < \mathrm{min}\left\{\frac{1}{2\sqrt{5} \beta \mathrm{L}}, \frac{1}{\beta E \mathrm{L}} \right\}$, we have 
\begin{align}
\label{lemma2_eqn}
    &\frac{1}{K} \sum_{k=0}^{K-1} \frac{1}{\Omega^k} \sum_{e=0}^{E-1} \sum_{b=0}^{B-1} \alpha_b \mathbb{E} \left[ \left\Vert \mathbf{w}^k - \mathbf{w}_{b}^{k,e} \right\Vert^2 \right] \nonumber\\
    &\leq \frac{4 E \eta^2 \sigma^2} {K} \sum_{k=0}^{K-1} \frac{1}{\Omega^k} \sum_{e=0}^{E-1} \sum_{b=0}^{B-1} \alpha_b \rs\rs \sum_{u \in \bar{\mathcal{U}}_b^{k,e}} \rs \left(\alpha_u\right)^2 \mathrm{L}_u^{k,e} + \nonumber\\
    &\frac{60 \beta^2 \sigma^2 E^3 \mathrm{L}^3 \eta^4}{K} \sum_{k=0}^{K-1} \frac{1}{\Omega^k} + \frac{180 \beta^2 \epsilon_0^2 E^3 \mathrm{L}^4 \eta^4}{K} \sum_{k=0}^{K-1} \frac{1}{\Omega^k} + \nonumber\\
    &\frac{10 \epsilon_1^2 \eta^2 E^3}{K} \sum_{k=0}^{K-1} \frac{1}{\Omega^k} +\frac{4 E \eta^2}{K} \sum_{k=0}^{K-1} \frac{1}{\Omega^k} \sum_{e=0}^{E-1} \sum_{b=0}^{B-1} \rs \alpha_b \rs\rs \sum_{u \in \bar{\mathcal{U}}_b^{k,e}} \alpha_u \big(\alpha_u + \nonumber\\
    &15 E \beta^2 \eta^2 \mathrm{L}^3 \big) \rs \bigg[\frac{1}{\mathrm{p}_{u,\mathrm{sc}}^{k,e}} - 1 \bigg] \sum_{l=0}^{\mathrm{L}_u^{k,e} - 1} \mathbb{E} \left[\left\Vert g_u (\mathbf{w}_u^{k, e, l}) \right\Vert^2\right].
\end{align}
Finally, we plug (\ref{lemama1_Eqn}) and (\ref{lemma2_eqn}) into (\ref{theorem1_Mid_Bound}) to find the final bound in (\ref{theorem1_eqn}).  
The detailed proof is left to the supplementary material.
\end{proof}


\begin{Remark}
\label{remTheorem}
The first diminishing term with $\{\cdot\}$ on the right-hand side of (\ref{theorem1_eqn}) captures the change in the global loss function in two consecutive global rounds. 
The second term with $\sigma^2$ arises due to the bounded variance assumption of the stochastic gradients.
The third term stems from the assumption that the divergence between the local loss function and \ac{es}' loss function due to statistical data heterogeneity is upper bounded by $\epsilon_0^2$. 
Similarly, the fourth term arises due to the bounded divergence assumption between the \ac{es}' loss function and the global loss function due to data heterogeneity. 
Moreover, the fifth term stems from the wireless links between the clients and their serving \acp{bs}.
\end{Remark}


\begin{Corollary}
\label{corol_1}
When all clients perform the same number of \ac{sgd} rounds, i.e., $\mathrm{L}_u^{k,e}=\mathrm{L}$, $\forall u, k$ and $e$, the convergence bound in Theorem \ref{theorem1} boils down to 
\begin{align*}
\label{corr0}
    & \frac{1}{K} \sum_{k=0}^{K-1} \mathbb{E} \big[\Vert \nabla f(\mathbf{w}^k) \Vert^2 \big]
    \leq \frac{2 ( \mathbb{E} [f(\mathbf{w}^0)] - \mathbb{E} [ f(\mathbf{w}^{K}) ])} {\eta E K \mathrm{L}} + \nonumber\\
    &[2 \beta \eta \sigma^2 / (E K) ] \sum\nolimits_{k=0}^{K-1} \mathrm{N}_1^k + 18 \beta^2 \epsilon_0^2 \eta^2 \mathrm{L}^2 \mathrm{N}_2 + 20 \beta^2 \epsilon_1^2 \eta^2 E^2 +\nonumber\\
    & \frac{2 \beta \eta}{E K} \sum_{k=0}^K \sum_{e=0}^{E-1} \sum_{b=0}^{B-1} \rs \alpha_b \rs \rs \rs \sum_{u \in \bar{\mathcal{U}}_b^{k,e}} \rs\rs\rs \alpha_u \mathrm{N}_u \rs\rs \sum_{l=0}^{\mathrm{L}_u^{k,e} - 1} \rs \bigg[\frac{1}{\mathrm{p}_{u,\mathrm{sc}}^{k,e}} - 1 \bigg] \mathbb{E} \big[\big\Vert g_u (\mathbf{w}_u^{k, e, l}) \big\Vert^2 \big],
\end{align*}
Moreover, when $\mathrm{p}_{u,\mathrm{sc}}^{k,e}$'s are $1$'s, i.e., all selected clients' accumulated gradients are received successfully, we have 
\begin{equation}
\label{cor1Eqn}
\begin{aligned}
    & \frac{1}{K} \sum_{k=0}^{K-1} \mathbb{E} \big[\Vert \nabla f(\mathbf{w}^k) \Vert^2 \big]
    \leq \frac{2 ( \mathbb{E} [f(\mathbf{w}^0)] - \mathbb{E} [ f(\mathbf{w}^{K}) ])} {\eta E K \mathrm{L}} + \\
    & \frac{2 \beta \eta \sigma^2}{E K} \sum\nolimits_{k=0}^{K-1} \mathrm{N}_1^k + 18 \beta^2 \epsilon_0^2 \eta^2 \mathrm{L}^2 \mathrm{N}_2 + 20 \beta^2 \epsilon_1^2 \eta^2 E^2,
\end{aligned}
\end{equation}
which is analogous to existing studies \cite{pervej2024hierarchical, wang2022demystifying}.
\end{Corollary}

\begin{Remark}
    The second, third and fourth terms of (\ref{cor1Eqn}) have some constant terms that may not depend on the global round $k$.
    However, the learning rate $\eta$ is typically significantly smaller than $1$.
    Since these terms are multiplied by the higher-order terms of $\eta$, they are negligible. 
    Therefore, the global gradient may converge to a neighborhood of the optimal solution.
\end{Remark}

\begin{Remark}[Choice of the learning rate $\eta$]
Based on the above analysis, we write the dominating terms of the bound from Corollary 1 as follows by further assuming $\Bar{\mathcal{U}}_{b}^{k,e} = \mathcal{U}_{b}$.
\begin{align}
\label{solve4convrate}
    \varpi 
    &\coloneqq \frac{2 F^{inf}} {\eta E K \mathrm{L}} + 2 \beta \eta \sigma^2 \sum\nolimits_{b=0}^{B-1} \left(\alpha_b \right)^2 \sum\nolimits_{u \in \mathcal{U}_b} \left(\alpha_u\right)^2,
\end{align}
where $F^{inf} \coloneqq \mathbb{E} [f(\mathbf{w}^0)] - \mathbb{E} [ f(\mathbf{w}^{K})]$.
This indicates that the optimal learning rate $\eta^{*} = \sqrt{\frac{F^{inf}}{EKL \times \beta \sigma^2 \sum_{b=0}^{B-1} \left(\alpha_b \right)^2 \sum_{u \in \mathcal{U}_b} \left(\alpha_u\right)^2}}$.
Besides, using equal weights and equal number of clients in each \ac{es}, i.e., $\alpha_u=\frac{1}{|\mathcal{U}_b|}=\frac{1}{\bar{U}}$ and $\alpha_b = \frac{1}{B}$, we have $\eta^{*} = \sqrt{\frac{F^{inf}}{EKL \times \beta \sigma^2 \sum_{b=0}^{B-1} \left(\frac{1}{B} \right)^2 \sum_{u \in \mathcal{U}_b} \left(\frac{1}{\bar{U}}\right)^2}} = \sqrt{\frac{B \bar{U}}{EKL} \times \frac{F^{inf}}{\beta \sigma^2}}$.
Plugging $\eta^{*}$ into (\ref{solve4convrate}) gives the convergence rate $\mathcal{O}\left(\frac{1}{\sqrt{EKL \times B \bar{U}}}\right)$, which indicates that we achieve linear speed-up with respect to the number of clients in all \acp{es} when $T \coloneqq EKL$ is sufficiently large.
This is due to the fact that we need $\mathcal{O}\left(\frac{1}{\tilde{\epsilon}^2 \times B\bar{U}} \right)$ steps as opposed to $\mathcal{O}\left(\frac{1}{\tilde{\epsilon}^2} \right)$ steps to achieve $\tilde{\epsilon}$ accuracy with algorithms that have the convergence rates of $\mathcal{O}\left(\frac{1}{\sqrt{T \times B\bar{U}} }\right)$ and $\mathcal{O}\left(\frac{1}{\sqrt{T}}\right)$, respectively.
\end{Remark}

\begin{Remark}
From the above Theorem and Corollary, it is quite evident that the convergence bound depends on $\bar{\mathcal{U}}_b^{k,e}$'s, $\Omega^k$'s and $\mathrm{p}_{u,\mathrm{sc}}^{k,e}$'s. 
\bblue{Besides, since the client's local loss function depends on the local dataset $\mathcal{D}_{u,\mathrm{proc}}^t$ by definition, the theoretical bounds also rely on the evolution of the dataset, which is governed by the activation probability $p_{u,\mathrm{ac}}$. 
However, given the fact that the servers have no control over the data distribution of the clients, we cannot govern clients' data acquisitions.}
As such, in the sequel, we optimize the intertwined \ac{fl} and wireless networking parameters jointly in order to facilitate \ac{rawhfl}'s convergence.
\end{Remark}

\section{RawHFL: Joint Problem Formulation and Solutions}
\label{jointProblemFormulation}


\subsection{Joint Problem Formulation}
\noindent
Based on the theoretical analysis in Section \ref{convAnalysis_Section}, we want to focus on the controllable terms of the convergence bound in (\ref{theorem1_eqn}), which are $\bar{\mathcal{U}}_b^{k,e}$'s, $\Omega^k$'s and $\mathrm{p}_{u,\mathrm{sc}}^{k,e}$'s, such that the average global gradient norm after $K$ global rounds is minimized.
Intuitively, some clients may have some predicaments in performing local model training due to multifarious factors, such as residing at the cell edge, poor wireless channel quality, limited energy budget and low computational resources, to name a few.
Besides, while a large $\mathrm{L}_u^{k,e}$ may increase $\Omega^k$, which improves the convergence rate, this is more complex in a constrained wireless network for the following reasons.
More \ac{sgd} rounds mean more local training overheads, leaving less time and energy for the uplink transmission once the local training is over.
Therefore, if the clients do not have sufficient time and energy left to transmit their trained accumulated gradients, the \ac{bs} and, hence, the \ac{es} cannot utilize these clients' contributions in the edge model aggregation.
Moreover, if $\mathrm{p}_{u,\mathrm{sc}}^{k,e}=0$, the last term of (\ref{theorem1_eqn}) can blow up, making the theoretical bound $\infty$. 
As such, it is critical to jointly optimize client selection and clients' local \ac{sgd} rounds, \ac{cpu} frequencies and transmission powers to make the convergence rate faster.

\begin{Remark}
In order to minimize the right-hand side of (\ref{theorem1_eqn}), we need to optimize the controllable parameters for all edge rounds in every global round before \ac{rawhfl} training begins. 
However, since the wireless links between the clients and the \acp{bs} vary and are unknown beforehand, it is impossible to know the optimal parameters for all edge rounds of all global rounds at the beginning of model training.
Owing to these complexities, a myopic perspective is necessary to sub-optimally minimize the right-hand side of (\ref{theorem1_eqn}). 
\end{Remark}

Concretely, we want to optimize the controllable parameters in every edge round so that the weighted summation of the local \ac{sgd} steps $\mathrm{L}_u^{k,e}$'s of the selected clients in $\bar{\mathcal{U}}_b^{k,e}$ is maximized. 
In other words, we want to maximize $\sum_{b=0}^{B-1} \alpha_b \sum_{u \in \bar{\mathcal{U}}_b^{k,e}} \alpha_u \mathrm{L}_u^{k,e} \equiv \sum_{b=0}^{B-1} \alpha_b \sum_{u \in \mathcal{U}_b^{k,e}} \mathrm{1}_u^{k,e} \cdot \big[\alpha_u \mathrm{L}_u^{k,e} \big]$.
To that end, each client transmits its accumulated gradients $\tilde{g}_u^{k,e}$ as a single wireless packet. 
Besides, as modern wireless networks have countermeasures for packet loss\footnote{Almost all practical wireless networks have cyclic redundancy check, error correction coding, and (hybrid) \ac{arq} in place \cite[Chap. $13$]{molisch2023wireless}.}, we assume that the transmitted packet can be decoded error-free if received during the aggregation time.
In particular, we define $\mathrm{p}_{u,\mathrm{sc}}^{k,e} = \mathrm{Pr}\{[ \mathrm{t}_{u,\mathrm{cp}}^{k,e} + \mathrm{t}_{u,\mathrm{up}}^{k,e} \leq \mathrm{t_{th}}\}$, which is also a common practice in the literature \cite{zeng2022federated,pervej2024hierarchical}. 
Then, we enforce $\mathrm{t}_{u,\mathrm{cp}}^{k,e} + \mathrm{t}_{u,\mathrm{up}}^{k,e} \leq \mathrm{t_{th}}$ as a constraint to ensure $\mathrm{p}_{u,\mathrm{sc}}^{k,e} \approx 1$ so that the last term in (\ref{theorem1_eqn}) becomes $0$.
Moreover, from (\ref{localCompTime}) and (\ref{gradientOffload}), it is evident that $\mathrm{p}_{u,\mathrm{sc}}^{k,e}$ depends on $\mathrm{L}_{u}^{k,e}$, $f_u^{k,e}$ and $P_{u,\mathrm{tx}}^{k,e}$.  
Due to these facts, we pose the following optimization problem that can be solved in each edge round to jointly configure the intertwined optimization variables.
\begin{subequations}
\label{originalOptimProb}
\begin{align}
    &\tag{\ref{originalOptimProb}}
    \underset{\pmb{\mathrm{1}}_{\mathrm{sl}}^{k,e}, \pmb{\mathrm{L}}^{k,e}, \pmb{f}^{k,e}, \pmb{P}_{\mathrm{tx}}^{k,e}}{ \mathrm{min} } \quad \varphi_b^{k,e} \coloneqq \frac{1}{\sum\nolimits_{b=0}^{B-1}  \alpha_b \sum\nolimits_{u \in \mathcal{U}_b} \mathrm{1}_{u, \mathrm{sl}}^{k,e} \cdot [\alpha_u \mathrm{L}_u^{k,e} ] } \\
    &\mathrm{s.t.} \quad C_1:  \quad \mathrm{1}_{u,\mathrm{sl}}^{k,e} \in \{0,1\}, \quad \forall u, k, e\\
    &\quad C_2: \quad \sum\nolimits_{u \in \mathcal{U}_b} \mathrm{1}_{u,\mathrm{sl}}^{k,e} = Z, \quad \forall b, k, e \\
    &\quad C_3: \begin{cases}
        \left(\pmb{\mathrm{1}}_{b,\mathrm{sl}}^{k,e}\right)^T \pmb{\mathrm{1}}_{b, \mathrm{sl}}^{k,e-1} \leq \Psi, & \text{if } k=0 ~\&~ 1\leq e \leq E-1, \\
        \left(\pmb{\mathrm{1}}_{b, \mathrm{sl}}^{k,e}\right)^T \pmb{\mathrm{1}}_{b,\mathrm{sl}}^{k-1,E-1} \leq \Psi, & \text{if } k \neq 0 ~\&~ (e+1)=E,
    \end{cases} \label{cons_C3} \\ 
    &\quad C_4: \quad 1 \leq \mathrm{L}_u^{k,e} \leq \mathrm{L}, ~ \mathrm{L}_u^{k,e} \in \mathbb{Z}^{+}, \quad \forall u, k, e \\
    &\quad C_5: \quad f_u^{k,e} \leq f_{u, \mathrm{max}}, 
    \quad \forall u, k, e\\
    &\quad C_6: \quad 0 \leq P_{u,\mathrm{tx}}^{k,e} \leq P_{u, \mathrm{max}}, \quad \forall u, k, e\\
    &\quad C_7: \quad \mathrm{1}_{u, \mathrm{sl}}^{k,e} \cdot [ \mathrm{t}_{u,\mathrm{cp}}^{k,e} + \mathrm{t}_{u,\mathrm{up}}^{k,e} ] \leq \mathrm{1}_{u, \mathrm{sl}}^{k,e} \cdot \mathrm{t_{th}}, \quad \forall u, k, e \\
    &\quad C_8: \quad \mathrm{1}_{u, \mathrm{sl}}^{k,e} \cdot [\mathrm{e}_{u,\mathrm{cp}}^{k,e} + \mathrm{e}_{u,\mathrm{up}}^{k,e} ] \leq \mathrm{1}_{u, \mathrm{sl}}^{k,e} \cdot \mathrm{e}_{u, \mathrm{bd}}, \forall u, k, e
\end{align}
\end{subequations}
where $\pmb{\mathrm{1}}_{b,\mathrm{sl}}^{k,e} = \big\{\mathrm{1}_{u,\mathrm{sl}}^{k,e}\big\}_{u\in\mathcal{U}_b}$, $\pmb{\mathrm{1}}_{\mathrm{sl}}^{k,e}=\big\{\pmb{\mathrm{1}}_{b,\mathrm{sl}}^{k,e}\big\}_{b=0}^{B-1}$, $\pmb{\mathrm{L}}^{k,e}=\Big\{\big\{\mathrm{L}_u^{k,e}\big\}_{u \in \mathcal{U}_b}\Big\}_{b=0}^{B-1}$, $\pmb{f}^{k,e}=\Big\{\big\{f_u^{k,e}\big\}_{u \in \mathcal{U}_b}\Big\}_{b=0}^{B-1}$, $\pmb{P}_{\mathrm{tx}}^{k,e}=\Big\{\big\{P_{u,\mathrm{tx}}^{k,e}\big\}_{u \in \mathcal{U}_b}\Big\}_{b=0}^{B-1}$ and $\Psi$ is a positive constant.
Note that $C_1$ is the binary client selection constraint, while $C_2$ ensures only $Z$ clients are selected in each \ac{bs}.
Besides, $C_3$ ensures that at least $\Psi$ new clients are selected in two consecutive edge rounds.
Furthermore, constraints $C_4$, $C_5$ and $C_6$ confirm that the client selects its local iteration, \ac{cpu} frequency and transmission power within the upper bounds, respectively. 
Finally, constraints $C_7$ and $C_8$ satisfy the deadline and energy constraints, respectively.

\begin{Remark}
The clients can share their $P_{u, \mathrm{max}}$'s, $f_{u,\mathrm{max}}$'s and $\mathrm{e}_{u,\mathrm{bd}}$'s with their associated \ac{bs}.
Then, the \acp{bs} can collaborate since they are under the same \ac{isp}, which allows the \ac{isp} to solve (\ref{originalOptimProb}) centrally.
Once the problem is solved, each \ac{bs} can broadcast the optimized parameters to the selected clients.
However, (\ref{originalOptimProb}) is a binary mixed-integer non-linear optimization problem and is NP-hard.
Therefore, in the following, we first transform this complex problem into a weighted utility minimization problem, followed by approximating the non-convex constraints and solving the transformed problem using an iterative solution.
It is worth noting that as \ac{ai}/\ac{ml} is deemed a potential solution to many problems in wireless communication (see \cite{du2020machine} and the reference therein), problem (\ref{originalOptimProb}) may also be solved using \ac{ai}/\ac{ml}. 
However, since the problem is NP-hand, an \ac{ai}/\ac{ml}-based solution may also only lead to a sub-optimal solution.
\end{Remark}


\subsection{Problem Transformations and Solution}
\noindent
Intuitively, one may equivalently choose $\varphi_b^{k,e} \coloneqq - \sum\nolimits_{b=0}^{B-1}\alpha_b \sum\nolimits_{u \in \mathcal{U}_b} \mathrm{1}_{u, \mathrm{sl}}^{k,e} \cdot [\alpha_u \mathrm{L}_u^{k,e} ]$ as the objective function in (\ref{originalOptimProb}), which should be minimized if the clients are selected to maximize the weighted summation of their $\mathrm{L}_u^{k,e}$'s. 
However, neither the equivalent objective nor the original one in (\ref{originalOptimProb}) seeks an energy-efficient solution, as the clients with the maximum possible $\mathrm{L}_u^{k,e}$'s are likely to be selected with these two objective functions.
Besides, it is practical to consider that clients may request new content only after watching the previously requested video content.
As such, a new training sample may only be available at the client's local dataset after a while. 
Due to these factors, we consider that the interval between two edge rounds, i.e., $\mathrm{t_{th}}$, is at least as long as the duration of the video.
Therefore, we concentrate on the energy costs.
More specifically, we revise $\varphi_b^{k,e}$ to balance the weight between $\mathrm{L}_{u}^{k,e}$'s and the corresponding energy expenses of the clients as
\begin{align}
\label{utilFunc}
    \varphi_b^{k,e} \coloneqq & - \theta \sum\nolimits_{b=0}^{B-1}\alpha_b \sum\nolimits_{u \in \mathcal{U}_b} \mathrm{1}_{u, \mathrm{sl}}^{k,e} \cdot [\alpha_u \mathrm{L}_u^{k,e} ] + \nonumber\\
    &\qquad (1-\theta)\sum\nolimits_{b=0}^{B-1}\alpha_b \sum\nolimits_{u \in \mathcal{U}_b} \mathrm{1}_{u, \mathrm{sl}}^{k,e} \cdot [\alpha_u \mathrm{e}_{u,\mathrm{tot}}^{k,e} ], 
\end{align}
where $\mathrm{e}_{u, \mathrm{tot}}^{k,e} = \mathrm{e}_{u, \mathrm{cp}}^{k,e} + \mathrm{e}_{u, \mathrm{up}}^{k,e}$ and $\theta \in [0,1]$.


Since the problem is still non-convex and NP-hard, we introduce relaxations that allow the problem to be solved more efficiently.
At first, we drop the integer constraint on $\mathrm{L}_{u}^{k,e}$.
Let $\bar{\mathrm{L}}_{u}^{k,e} \coloneqq \mathrm{1}_{u,sl}^{k,e} \cdot \mathrm{L}_{u}^{k,e}$ be a new variable.
Then, we introduce the following constraints to replace the multiplication of the binary and continuous variables.
\begin{align}
    1 \cdot \mathrm{1}_{u, \mathrm{sl}}^{k,e} &\leq \bar{\mathrm{L}}_{u}^{k,e} \leq \mathrm{L} \cdot \mathrm{1}_{u, \mathrm{sl}}^{k,e}; \quad 0 \leq \bar{\mathrm{L}}_{u}^{k,e} \leq \mathrm{L}. \label{sgdIterC1}\\
    1 \cdot (1 - \mathrm{1}_{u, \mathrm{sl}}^{k,e}) &\leq \mathrm{L}_{u}^{k,e} - \bar{\mathrm{L}}_{u}^{k,e} \leq \mathrm{L} \cdot (1 - \mathrm{1}_{u,\mathrm{sl}}^{k,e}) \label{sgdIterC2}.\\
    \bar{\mathrm{L}}_{u}^{k,e} &\leq \mathrm{L}_u^{k,e} + (1 - \mathrm{1}_{u,\mathrm{sl}}^{k,e}) \mathrm{L} \label{sgdIterC3}.
\end{align}
After this, we equivalently transform the binary client selection constraint as 
\begin{align}
    &\sum_{b=0}^{B-1}\sum_{u \in \mathcal{U}_b} \mathrm{1}_{u,\mathrm{sl}}^{k,e} - \sum_{b=0}^{B-1}\sum_{u \in \mathcal{U}_b} (\mathrm{1}_{u,\mathrm{sl}}^{k,e} )^2 \leq 0, \label{transformBinaryClientSel} \\ 
    & 0 \leq \mathrm{1}_{u,\mathrm{sl}}^{k,e} \le 1.
\end{align}

We then use (\ref{transformBinaryClientSel}) in objective function (\ref{utilFunc}) to add a penalty when the client selection decision is not $0$ or $1$ as  
\begin{align}
    &\tilde{\varphi}_b^{k,e} \coloneqq - \beta \sum\nolimits_{b=0}^{B-1} \rs \alpha_b \rs \rs \sum_{u \in \mathcal{U}_b} \rs\rs \alpha_u \bar{\mathrm{L}}_u^{k,e} + (1-\beta)\sum\nolimits_{b=0}^{B-1} \rs \alpha_b \rs \rs \sum_{u \in \mathcal{U}_b} \rs \rs \alpha_u [\Tilde{\mathrm{e}}_{u,\mathrm{cp}}^{k,e} + \nonumber\\
    &~ \Tilde{\mathrm{e}}_{u,\mathrm{up}}^{k,e} ] + \varrho \big( \sum\nolimits_{b=0}^{B-1} \sum\nolimits_{u \in \mathcal{U}_b} \mathrm{1}_{u,\mathrm{sl}}^{k,e} - \sum\nolimits_{b=0}^{B-1}\sum\nolimits_{u \in \mathcal{U}_b} (\mathrm{1}_{u,\mathrm{sl}}^{k,e} )^2 \big), 
\label{utilFuncWithPenalty}
\end{align}
where $\Tilde{\mathrm{e}}_{u,\mathrm{cp}}^{k,e} = \bar{\mathrm{L}}_u^{k,e} \times 0.5\zeta \mathrm{n} \bar{\mathrm{n}} \mathrm{c}_u \mathrm{D}_u (f_u^{k,e})^2$, $\Tilde{\mathrm{e}}_{u,\mathrm{up}}^{k,e} = \mathrm{s} P_{u,\mathrm{tx}}^{k,e} \mathrm{1}_{u,\mathrm{sl}}^{k,e} / \big(\omega \log_2 \big[1 + \gamma_{u}^{k,e}\big]\big)$ and $\varrho > 0$ is a positive constant.
However, the last quadratic term in (\ref{utilFuncWithPenalty}) is still a problem and is thus approximated using a first-order Taylor series, resulting in the following modified objective function. 
\begin{align}
\label{objFunc}
    &\rs\rs \rs \Bar{\varphi}_b^{k,e} \coloneqq - \beta \sum_{b=0}^{B-1} \alpha_b \rs\rs \sum_{u \in \mathcal{U}_b} \rs \rs \alpha_u \bar{\mathrm{L}}_u^{k,e} + (1-\beta)\sum_{b=0}^{B-1} \rs \alpha_b \rs \rs \sum_{u \in \mathcal{U}_b} \rs \rs \alpha_u [\Bar{\mathrm{e}}_{u,\mathrm{cp}}^{k,e} + \Bar{\mathrm{e}}_{u,\mathrm{up}}^{k,e} ] + \nonumber\\
    & \varrho \big( \sum\nolimits_{b=0}^{B-1} \sum_{u \in \mathcal{U}_b} \rs [1 - 2 \cdot \mathrm{1}_{u,\mathrm{sl}}^{k,e,(i)} ] \mathrm{1}_{u,\mathrm{sl}}^{k,e} + \rs \sum\nolimits_{b=0}^{B-1} \sum_{u \in \mathcal{U}_b} \rs (\mathrm{1}_{u,\mathrm{sl}}^{k,e, (i)} )^2 \big), \rs\rs \rs
\end{align}
where $\bar{\mathrm{e}}_{u,\mathrm{cp}}^{k,e} = \zeta \mathrm{n} \bar{\mathrm{n}} \mathrm{c}_u \mathrm{D}_u f_u^{k,e,(i)} \big[0.5 f_u^{k,e,(i)} \bar{\mathrm{L}}_u^{k,e} + \bar{\mathrm{L}}_u^{k,e,(i)} f_u^{k,e} - f_u^{k,e,(i)} \bar{\mathrm{L}}_u^{k,e,(i)} \big]$ and $\mathrm{1}_{u,\mathrm{sl}}^{k,e,(i)}$, $\bar{\mathrm{L}}_u^{k,e,(i)}$ and $f_u^{k,e,(i)}$ are some initial feasible points.

We now turn the focus on the non-convex constraints. 
First, we approximate the energy requirement for the accumulated gradient offloading as 
\begin{align}
\label{offloadEnergyApprox}
    \Tilde{\mathrm{e}}_{u,\mathrm{up}}^{k,e} 
    &\approx \frac{\mathrm{s} \log(2)}{\omega} \bigg[ P_{u,\mathrm{tx}}^{k,e,(i)} \mathrm{1}_{u,\mathrm{sl}}^{k,e} \Big/ \log\bigg[1 + \frac{\beta_u^{k,e} \zeta_u^{k,e} P_{u,\mathrm{tx}}^{k,e,(i)} }{\omega \varsigma^2} \bigg] + \nonumber\\
    &\frac{\mathrm{1}_{u,\mathrm{sl}}^{k,e,(i)} \log\bigg[1 + \frac{\beta_u^{k,e} \zeta_u^{k,e} P_{u,\mathrm{tx}}^{k,e,(i)} }{\omega \varsigma^2} \bigg] - \frac{\beta_u^{k,e} \zeta_u^{k,e} P_{u,\mathrm{tx}}^{k,e,(i)} \mathrm{1}_{u,\mathrm{sl}}^{k,e,(i)} }{\omega\varsigma^2 + \beta_u^{k,e}\zeta_u^{k,e} P_{u,\mathrm{tx}}^{k,e,(i)} } }{\left(\log\bigg[1 + \frac{\beta_u^{k,e} \zeta_u^{k,e} P_{u,\mathrm{tx}}^{k,e,(i)} }{\omega \varsigma^2} \bigg] \right)^2} \times \nonumber\\
    & \left(P_{u,\mathrm{tx}}^{k,e} - P_{u,\mathrm{tx}}^{k,e,(i)}\right) \bigg]
    \coloneqq \Bar{\mathrm{e}}_{u,\mathrm{up}}^{k,e}
\end{align}
Similarly, we approximate the non-convex local computation time as 
\begin{align}
    \bar{\mathrm{t}}_{u,\mathrm{cp}}^{k,e} 
    \rs \approx [\mathrm{n} \bar{\mathrm{n}} \mathrm{c}_u \mathrm{D}_u/f_u^{k,e,(i)}] \big(\bar{\mathrm{L}}_u^{k,e,(i)} \rs - \bar{\mathrm{L}}_{u}^{k,e,(i)} f_u^{k,e} / f_u^{k,e,(i)} \rs + \bar{\mathrm{L}}_u^{k,e}  \big)\!,\rs\rs 
\end{align}
Furthermore, the non-convex offloading time constraint is approximated as follows:
\begin{align}
    &\tilde{\mathrm{t}}_{u,\mathrm{up}}^{k,e} \coloneqq \mathrm{1}_{u,\mathrm{sl}}^{k,e} \cdot \mathrm{s} / \bigg(\omega \log_2 \bigg[1 + \frac{\beta_{u}^{k,e} \zeta_u^{k,e} P_{u,\mathrm{tx}}^{k,e}}{\omega \varsigma^2} \bigg] \bigg) \nonumber\\
    &\approx \frac{\mathrm{s} \log(2)}{\omega}\bigg[  \mathrm{1}_{u,\mathrm{sl}}^{k,e}/ \log \bigg[1 + \frac{\beta_{u}^{k,e} \zeta_u^{k,e} P_{u,\mathrm{tx}}^{k,e,(i)}}{\omega \varsigma^2} \bigg]  - \nonumber\\
    & \frac{\mathrm{1}_{u, \mathrm{sl}}^{k,e,(i)} \beta_{u}^{k,e} \zeta_u^{k,e} }{\left(\omega\varsigma^2 + \beta_{u}^{k,e} \zeta_u^{k,e} P_{u,\mathrm{tx}}^{k,e,(i)} \right) \left( \log \bigg[1 + \frac{\beta_{u}^{k,e} \zeta_u^{k,e} P_{u,\mathrm{tx}}^{k,e,(i)}}{\omega \varsigma^2} \bigg] \right)^2} \times \nonumber\\
    & \left(P_{u,\mathrm{tx}}^{k,e} - P_{u,\mathrm{tx}}^{k,e,(i)} \right) \bigg] 
    \coloneqq \bar{\mathrm{t}}_{u,\mathrm{up}}^{k,e}
\end{align}



To that end, we pose the transformed problem as 
\begin{subequations}
\label{transFormedOptimProb}
\begin{align}
    &\tag{\ref{transFormedOptimProb}}
    \underset{\pmb{\mathrm{1}}_{\mathrm{sl}}^{k,e}, \bar{\pmb{\mathrm{L}}}^{k,e}, \pmb{\mathrm{L}}^{k,e}, \pmb{f}^{k,e}, \pmb{P}^{k,e} }{ \mathrm{min} } \quad \bar{\varphi}_b^{k,e} \\
    &\mathrm{s.t.} \quad  \quad 0 \leq \mathrm{1}_{u,\mathrm{sl}}^{k,e} \leq 1, \quad \forall u,k,e \\
    &\qquad \quad C_2, C_3, ~ 1 \leq \mathrm{L}_u^{k,e} \leq \mathrm{L}, ~ (\ref{sgdIterC1}), (\ref{sgdIterC2}), (\ref{sgdIterC3}), C_5, C_6 \\ 
    &\qquad \quad \bar{\mathrm{t}}_{u,\mathrm{cp}}^{k,e} + \bar{\mathrm{t}}_{u,\mathrm{up}}^{k,e} \leq \mathrm{1}_{u, \mathrm{sl}}^{k,e} \cdot \mathrm{t_{th}}, \quad \forall u, k, e \\
    &\qquad \quad \bar{\mathrm{e}}_{u,\mathrm{cp}}^{k,e} + \bar{\mathrm{e}}_{u,\mathrm{up}}^{k,e} \leq \mathrm{1}_{u, \mathrm{sl}}^{k,e} \cdot \mathrm{e}_{u, \mathrm{bd}}, \forall u, k, e
\end{align}
\end{subequations}
where the constraints are similar as in (\ref{originalOptimProb}).


\begin{algorithm}[!t]
\small
\SetAlgoLined 
\DontPrintSemicolon
\KwIn{Initial points $\mathrm{1}_{u,\mathrm{sl}}^{k,e,(i)}$'s, $\bar{\mathrm{L}}_{u}^{k,e,(i)}$, $\mathrm{L}_{u}^{k,e,(i)}$, $f_u^{k,e,(i)}$'s, $P_{u,\mathrm{tx}}^{k,e,(i)}$'s, $i=0$, total iteration $I$, precision level $\bar{\varepsilon}$ and $\varrho$ \;}
\nl{\textbf{Repeat}: \;} 
\Indp { $i \gets i+1$ \;
        Solve (\ref{transFormedOptimProb}) using $\mathrm{1}_{u,\mathrm{sl}}^{k,e,(i-1)}$'s, $\mathrm{L}_{u}^{k,e,(i-1)}$, $\bar{\mathrm{L}}_{u}^{k,e,(i-1)}$, $f_u^{k,e,(i-1)}$'s, $P_{u,\mathrm{tx}}^{k,e,(i-1)}$'s and $\varrho$, and get optimized $\mathrm{1}_{u,\mathrm{sl}}^{k,e}$'s, $\mathrm{L}_{u}^{k,e}$'s, $\bar{\mathrm{L}}_{u}^{k,e}$'s, $f_u^{k,e}$'s and $P_{u,\mathrm{tx}}^{k,e}$'s \;
        $\mathrm{1}_{u,\mathrm{sl}}^{k,e,(i)} \gets \mathrm{1}_{u,\mathrm{sl}}^{k,e}$; $\mathrm{L}_{u}^{k,e,(i)} \gets \mathrm{L}_u^{k,e}$; $\bar{\mathrm{L}}_{u}^{k,e,(i)} \gets \bar{\mathrm{L}}_{u}^{k,e}$; $f_u^{k,e,(i)} \gets f_u^{k,e}$; $P_{u,\mathrm{tx}}^{k,e,(i)} \gets P_{u,\mathrm{tx}}^{k,e}$ \;
        }
\Indm \textbf{Until} converge with precision $\bar{\varepsilon}$ or $i=I$\;
\KwOut{Optimal $\mathrm{1}_{u,\mathrm{sl}}^{k,e}$'s, $\mathrm{L}_{u}^{k,e}$'s, $\bar{\mathrm{L}}_{u}^{k,e}$'s, $f_u^{k,e}$'s and $P_{u,\mathrm{tx}}^{k,e}$'s}
\caption{Iterative Solution for (\ref{transFormedOptimProb})}
\label{iterAlg}
\end{algorithm}

This transformed problem (\ref{transFormedOptimProb}) now belongs to the class of ``difference of convex programming" problems.
Starting with some initial feasible points, we can iteratively solve the problem using existing tools such as CVX \cite{diamond2016cvxpy}.  
In particular, we use Algorithm \ref{iterAlg} to solve this problem.  
Note that the above iterative solution is well known to converge to a local stationary point of the original problem in polynomial time \cite{sun2019optimal}. 
Problem (\ref{transFormedOptimProb}) has $5U$ decision variables and $(2B+12U)$ constraints.
As such, the worst-case time complexity of running Algorithm \ref{iterAlg} for $I$ iterations is $\mathcal{O}\left(I \times 125U^3 [12U+2B]\right)$ \cite{pervej2023Resource}.
It is worth noting that while the proposed algorithm yields a sub-optimal solution, the empirical performance of the proposed \ac{rawhfl} is nearly identical to the performance of the ideal case \ac{hfedavg} \cite{liu2020client}. 


\section{Simulation Results and Discussions}
\label{simulationResults}

\begin{figure}
\centering
    \includegraphics[trim=10 10 45 10, clip, width=0.35\textwidth]{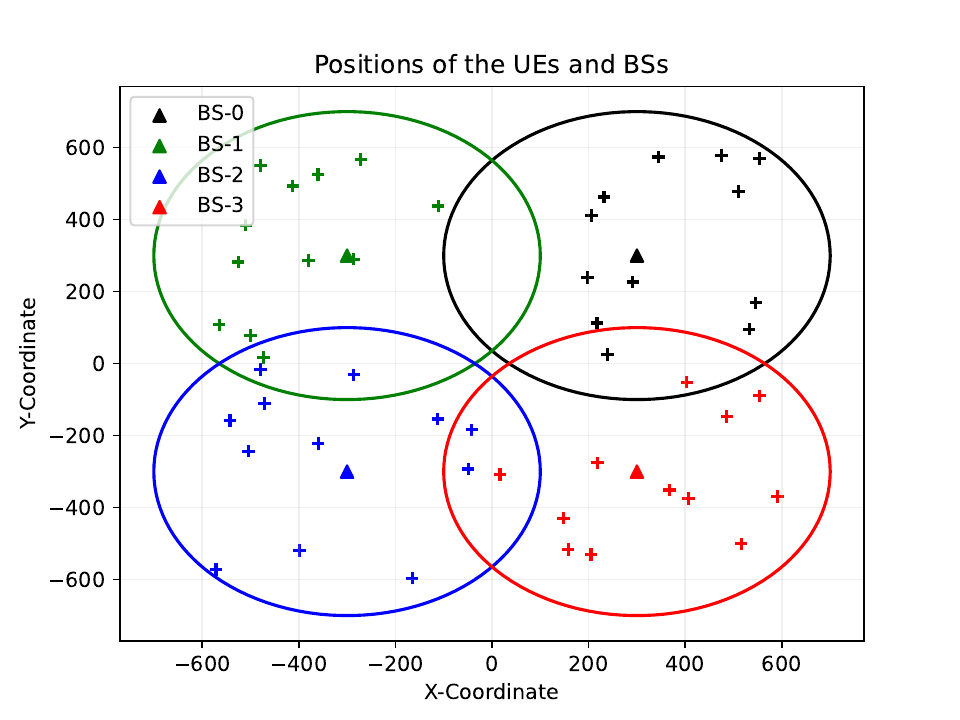}
    \caption{Locations of the \acp{ue} and \acp{bs} in one realization}
    \label{nodeLoc}
\end{figure}

\subsection{Simulation Setting}
\noindent
We perform extensive simulations with different parameter settings to show the performance of the proposed solution. 
We consider a total $B=4$ \acp{bs}, each serving $|\mathcal{U}_b|=12$ \acp{ue} with a coverage radius of $400$ meters, as shown in Fig. \ref{nodeLoc}. 
Therefore, there are total $U=48$ clients.
The \ac{isp}'s carrier frequency is $2.4$ GHz.
The \ac{prb} size is $\omega=3 \times 180$ kHz, while we vary the number of \acp{prb} based on $|\Bar{\mathcal{U}}_b^{k,e}|$'s. 
We model the path loss and line-of-sight probabilities following the $3$GPP urban macro model \cite[Section $7.4$]{3GPPTR38_901}.
Besides, $G=8$, $\bar{C}_g=32, \forall g$, $C=256$, $I=50$, $\varrho=1$, $\theta=0.4$, $\phi=32$ \cite{pervej2024hierarchical}, $\zeta=2\times10^{-28}$ \cite{liu2020client} and $\kappa=\mathrm{t_{th}}=150$ seconds.
The $c_u$'s, $f_{u,\mathrm{max}}$'s, $\mathrm{e}_{u,\mathrm{bd}}$'s and $P_{u, \mathrm{max}}$'s are uniformly randomly drawn from $[25, 40]$ cycles, $[1.2, 2.0]$ GHz, $[0.8, 1.5]$ Joules and $[20,30]$ dBm, respectively. 
Furthermore, we uniformly randomly select the activity levels $p_{u,\mathrm{ac}}$'s and the probability of exploiting similar content in the same genre $\upsilon_u$'s from $[0.2, 0.8]$ and $[0.1, 0.8]$, respectively.
Moreover, for the genre preferences $p_{u,g}$'s, we use $\mathrm{Dir}(\pmb{0.3})$.

To our best knowledge, \bblue{there exists no extensive} real-world video-caching-related dataset \bblue{containing} spatial and temporal information about the users and their requested content set that is suitable for \ac{rawhfl}. 
As such, we use the procedures described in Section \ref{contReqModSection} to let the \acp{ue} make content requests and acquire their (synthetic) dataset. 
Given that the \ac{ue} requests content $c_g$ during slot $t$, the extracted features for this content, i.e., $\mathbf{x}(\mathrm{1}_{u,c_g}^t)$ contains the following information: $\left\{\upsilon_u, \{p_{u,g}\}_{g=0}^{G-1}, \frac{g}{G}, \pmb{\mathrm{sim}}(c_g), \frac{c_g}{\Bar{C}_g} \right\}$, where $\pmb{\mathrm{sim}}(c_g)$ is a vector of \emph{cosine} similarities between the features of the requested content $c_g$ and the features of the rest of the content in genre $g$. 
Besides, \bblue{the label $y(\mathrm{1}_{u,c_g}^t)$ is the index of the requested content ID}\footnote{The number of labels equals the number of total files. 
We can stack the files from each genre one after another to prepare the label indices.}.

We use a simple fully connected ({\tt{FC}}) neural network\footnote{Our proposed \ac{rawhfl} solution is general.
Other \ac{ml} models, e.g., \acp{rnn} or transformers, can also be used.} with rectified linear ({\tt{ReLU}}) activation functions.
More specifically, the model has the following architecture: ${\tt{FC}} (\mathrm{input\_shape}, 512)  \rightarrow {\tt{ReLU}} () \rightarrow  {\tt{FC}}(512, 256) \rightarrow {\tt{ReLU}}() \rightarrow {\tt{FC}}(256, \mathrm{output\_shape})$. 
This model has $d=219648$ trainable parameters in our implementation, which makes the wireless payload size $s=7248384$ bits.
Moreover, a sliding window technique is used to prepare the processed datasets $\mathcal{D}_{u,\mathrm{proc}}^t$'s for model training. 
Particularly, we consider that each client processes its raw dataset such that the {\em previous} slot's requested content's feature set and the {\em current} slot's requested content ID are used as the feature and label, respectively.
The clients use batch size $\bar{\mathrm{n}}=32$, $\mathrm{n}=10$ mini-batches, at max $\mathrm{L}=50$ local rounds and {\tt\ac{sgd}} optimizer with a learning rate of $0.01$ for model training.
\bblue{Besides, cross-entropy loss function\footnote{\url{https://pytorch.org/docs/stable/generated/torch.nn.CrossEntropyLoss.html}} is used for model training. 
The accuracy is determined by calculating the fraction of total correctly predicted labels over the total labels.}
Furthermore, we use $E=4$ edge rounds in every global round and train for a total of $K=100$ global rounds\footnote{\bblue{We only performed computer simulations, which were run on our HPC cluster (\url{https://www.carc.usc.edu/}) with different $32$ gigabytes RAM and GPU-enabled compute nodes.}}.

Finally, as the clients are randomly dropped, and their system configurations are generated randomly, we completely repeat our simulations $10$ times in order to get an average performance.
In the following, unless mentioned otherwise, the average performance from these $10$ trials is reported.
It is worth noting that although the wireless channel between a client and the serving \ac{bs} can vary significantly, the choice of the \ac{isp}'s \ac{prb} size should generally impact the accumulated gradient offloading time when the model has a large number of trainable parameters.
As such, the \ac{csp} and \ac{isp} should collaborate to choose the model size and \ac{prb} size accordingly in order to facilitate the \ac{fl} process in video caching networks under deadline and resource constraints.

\begin{figure*}
\centering
\begin{subfigure}{0.25\textwidth}
    \centering
    \includegraphics[trim=3 08 30 28, clip, width=\textwidth, height=0.22\textheight]{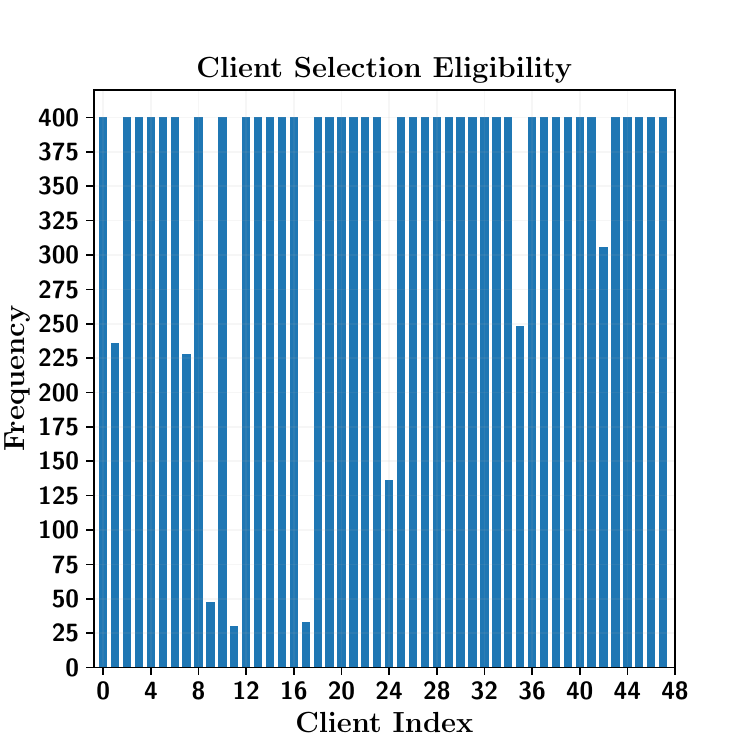}
    \caption{Client's eligible frequency}
    \label{clientSelElig}
\end{subfigure}
\begin{subfigure}{0.74\textwidth}
    \centering
    \includegraphics[trim=45 5 30 20, clip, width=\textwidth, height=0.22\textheight]{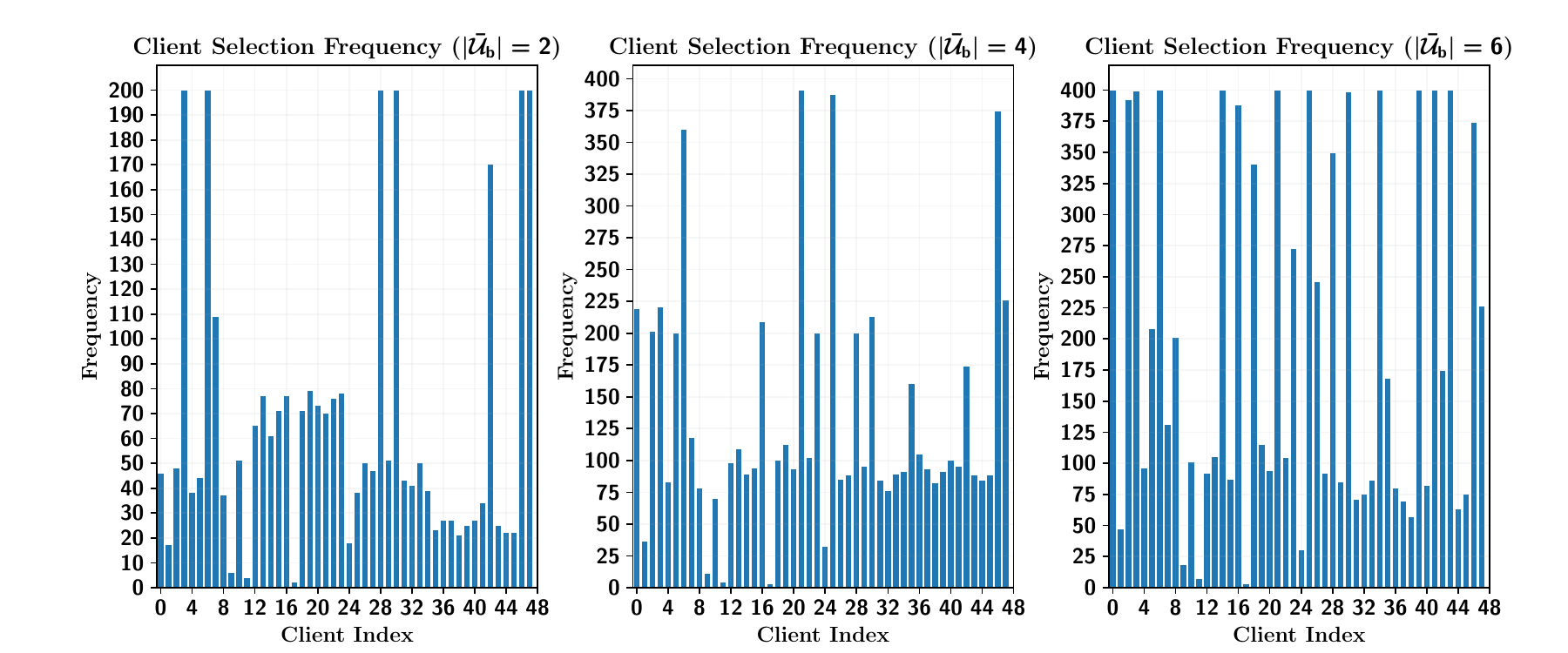}
    \caption{Optimized client selection frequencies}
    \label{clientSelFreq}
\end{subfigure}
\caption{\bblue{Clients' participation eligibility and number of selections in $400$ edge rounds}}
\end{figure*}

\subsection{Performance Analysis}


\noindent 
In a resource-constrained environment, client selection is necessary due to the following factors. 
Firstly, the wireless network has limited radio resources. 
In many cases, $Z < |\mathcal{U}_b|$ is realistic as a practical network has many other tasks to perform simultaneously.
Secondly, since wireless channels vary, some clients may have poor link quality that can prolong the accumulated gradient offloading. 
Thirdly, some clients may not have sufficient energy budgets to offload their accumulated gradients.
Finally, depending on the system configurations, some clients may need extended time to complete the local training even though the link quality is good. 
As such, appropriate client selection based on the available resources is essential.
Moreover, since we have a constrained environment, some clients may not even be eligible to participate in model training when their required time and/or energy overheads are beyond their budgets.

We observe a similar trend in our simulation results. 
In Fig. \ref{clientSelElig}, we show the eligible clients, i.e., the clients who can perform at least one local \ac{sgd} round, for $K=100$ global rounds in $1$ of the $10$ simulation runs.
Recall that client selection happens in each edge round, and each global round has $E=4$ edge rounds, leading to a total $4 \times 100$ edge rounds in Fig. \ref{clientSelElig}.
As expected, some clients are only eligible to be selected in a few edge rounds.
For example, client \bblue{$11$} is eligible to be selected only in \bblue{$28$} edge rounds.
Moreover, we plotted the client selection frequency with our proposed solution of (\ref{transFormedOptimProb}) in Fig. \ref{clientSelFreq} for different $|\mathcal{U}_b|$'s.
It is worth pointing out that since Algorithm \ref{iterAlg} yields a sub-optimal solution, the client selection strategies may not be optimal. 
Our results in Fig. \ref{clientSelFreq} validate that Algorithm \ref{iterAlg} finds the subset client set (and other decision variables) to minimize the objective function to satisfy all constraints.
As such, in the following, we utilize Algorithm \ref{iterAlg} and report the average performance of the $10$ simulation repeats.

Now, we investigate how the convergence performance of \ac{rawhfl} varies with different $|\bar{\mathcal{U}}_b^{k,e}|$'s.
Intuitively, the number of available training data should affect the model performance.
When only a few clients are selected, the model is trained on the few samples of these clients.
This can be detrimental if appropriate clients are not selected, especially when data heterogeneity is severe.
As such, \ac{rawhfl} may require more training to achieve a certain level of accuracy when $|\bar{\mathcal{U}}_b^{k,e}|$ is small.
Moreover, a client's content request probability from the same genre, $\upsilon_u$, determines how frequently the client requests similar content from the same genre.
A high $\upsilon_u$ means less exploration, i.e., more exploitation from the same genre.
As such, a high $\upsilon_u$ is expected to yield better prediction performance since the client has fewer variations in its content requests.


Our simulation results also show similar trends in Figs. \ref{flRoundVsTestLossAcc_0.1_0.8} - \ref{flRoundVsTestLossAcc_0.7}.
Note that these test results are the average of all clients' individual test datasets.
We observe that the value of $\upsilon_u$ governs the test accuracy and test loss.
For example, when all clients randomly select their respective $\upsilon_u$ from $0.1$ to $0.8$, the test accuracy after $K=100$ global rounds is around $45\%$.
When all clients have the same $\upsilon_u$, in Fig. \ref{flRoundVsTestLossAcc_0.3} and Fig. \ref{flRoundVsTestLossAcc_0.7}, we also observe that a large $\upsilon_u$ leads to faster convergence.
Moreover, the test loss and test accuracy performances improve when $|\bar{\mathcal{U}}_b^{k,e}|$ increases regardless of the value of $\upsilon_u$'s.
For example, with $|\bar{\mathcal{U}}_b^{k,e}|=2$, the test accuracy fluctuates, \bblue{even for $\upsilon_u=0.7$. 
However, the test accuracy reaches a plateau after about $K=14$ global rounds with $|\mathcal{U}_b| \geq 4$ selected clients per \ac{bs} when $\upsilon_u=0.7$.}

\begin{figure*}
\begin{subfigure}{0.33\textwidth}
    \centering
    \includegraphics[trim=15 17 10 15, clip, width=\textwidth]{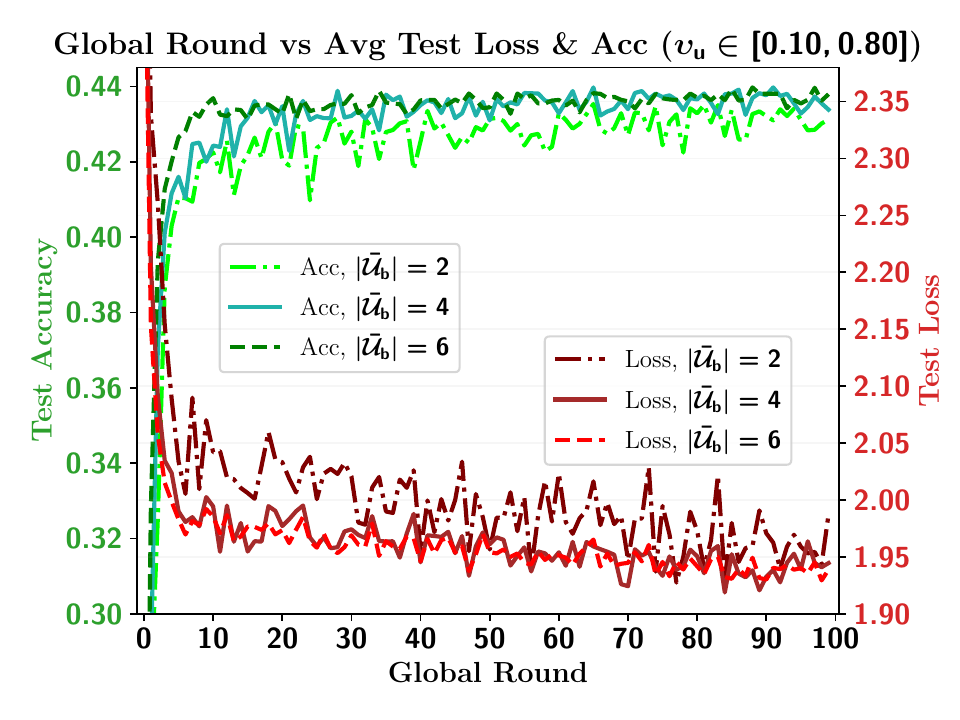}
    \caption{\bblue{Global round vs test metrics: $\upsilon_u \in [0.1, 0.8]$}}
    \label{flRoundVsTestLossAcc_0.1_0.8}
\end{subfigure}
\begin{subfigure}{0.33\textwidth}
\centering
    \includegraphics[trim=15 17 10 15, clip, width=\textwidth]{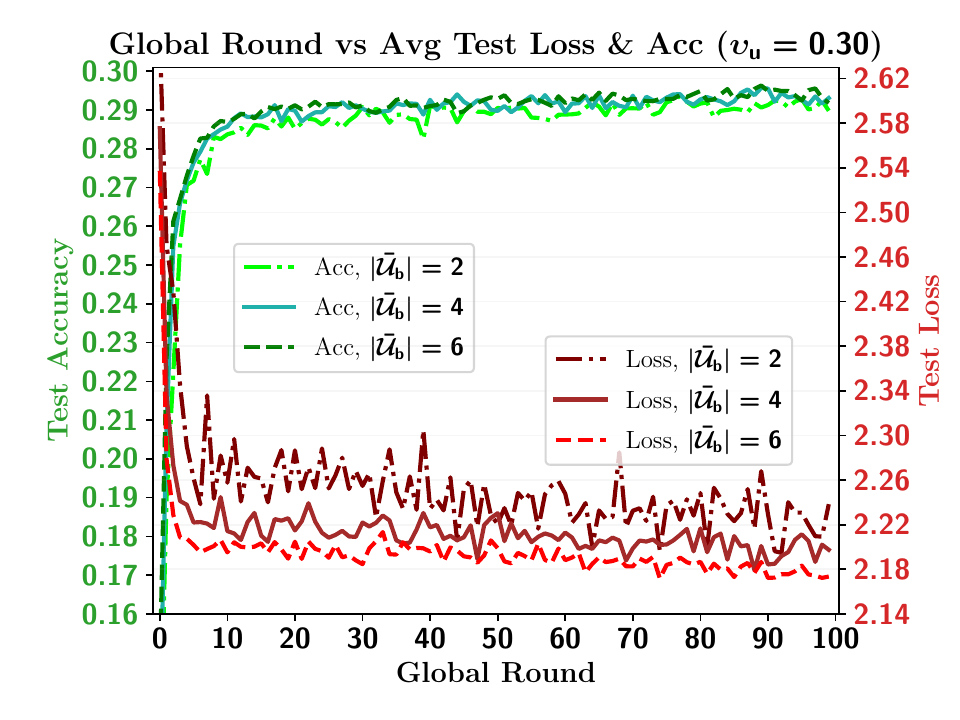}
    \caption{\bblue{Global round vs test metrics: $\upsilon_u = 0.3$}}  
    \label{flRoundVsTestLossAcc_0.3}
\end{subfigure}
\begin{subfigure}{0.33\textwidth}
    \centering
    \includegraphics[trim=15 17 10 15, clip, width=\textwidth]{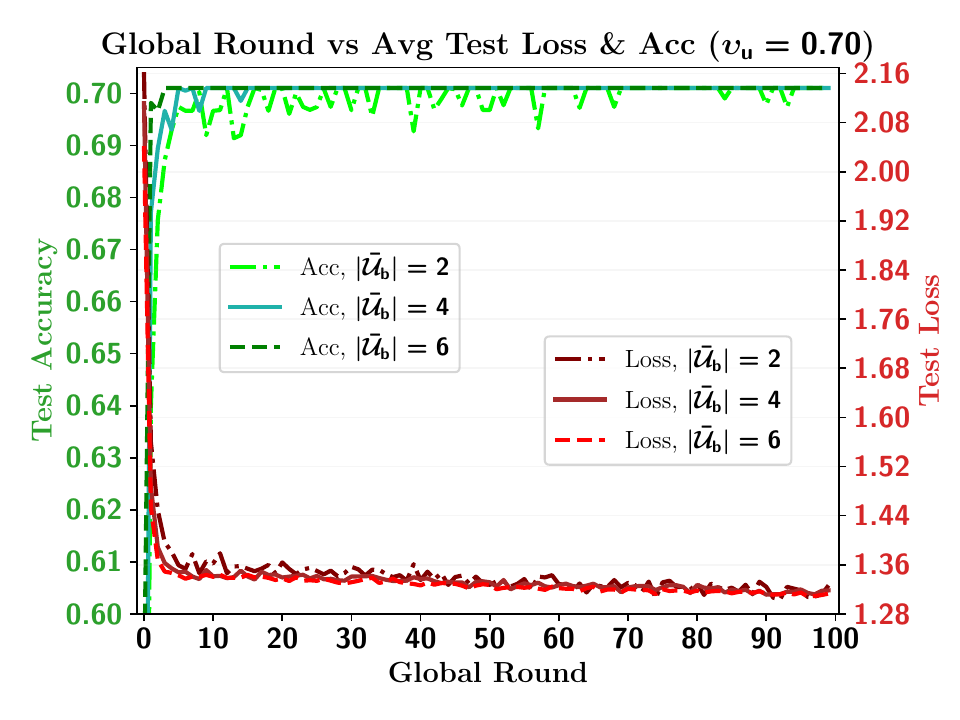}
    \caption{\bblue{Global round vs test metrics: $\upsilon_u = 0.7$}}
    \label{flRoundVsTestLossAcc_0.7}
\end{subfigure}
\caption{\bblue{Global round vs. test accuracy and test loss in \ac{rawhfl}}}
\end{figure*}


However, while more clients, i.e., large $|\bar{\mathcal{U}}_b^{k,e}|$, may help RawHFL converge faster, that comes with higher bandwidth and energy overheads.
On the one hand, the \ac{isp} will require more \acp{prb} to schedule the clients for wireless transmission.
On the other hand, the total energy cost per edge round also increases since more clients participate in the training process.
However, regardless of the number of participating clients, the per-client energy cost in each edge round can be non-trivial. 
More specifically, when $|\bar{\mathcal{U}}_b^{k,e}|$ increases, (\ref{transFormedOptimProb}) must choose the defined number of clients so that the utility function is minimized. 
In doing so, it may select some clients with higher energy expenses that can be a byproduct of poor wireless links. 
However, the total energy costs per client can still be similar due to varying wireless links, different local \ac{sgd} rounds and sub-optimal solutions of (\ref{transFormedOptimProb}) in different edge rounds.
Our simulation results in Fig. \ref{per_BS_cdfEnergyExpense} and Fig. \ref{per_UE_cdfEnergyExpense} also validate these claims.
Note that Fig. \ref{per_BS_cdfEnergyExpense} shows the \ac{cdf} of the total energy expenses per \ac{bs} during each edge round $e$ for different client set sizes.
Besides, Fig. \ref{per_UE_cdfEnergyExpense} shows the \ac{cdf} of per client total energy expense during each edge round for different $|\mathcal{U}_b|$'s.
For example, the probability of total energy expense of $50\%$ of \acp{bs} are more than \bblue{$3.7$} Joules, \bblue{$7.1$} Joules, and \bblue{$10.3$} Joules for $|\bar{\mathcal{U}}_b^{k,e}|=2$, $|\bar{\mathcal{U}}_b^{k,e}|=4$ and $|\bar{\mathcal{U}}_b^{k,e}|=6$, respectively.
Moreover, per client energy expense per edge round is quite similar in Fig. \ref{per_UE_cdfEnergyExpense}.

\begin{figure}[!t]
    \centering
    \includegraphics[trim=5 1 35 25, clip, width=0.4\textwidth]{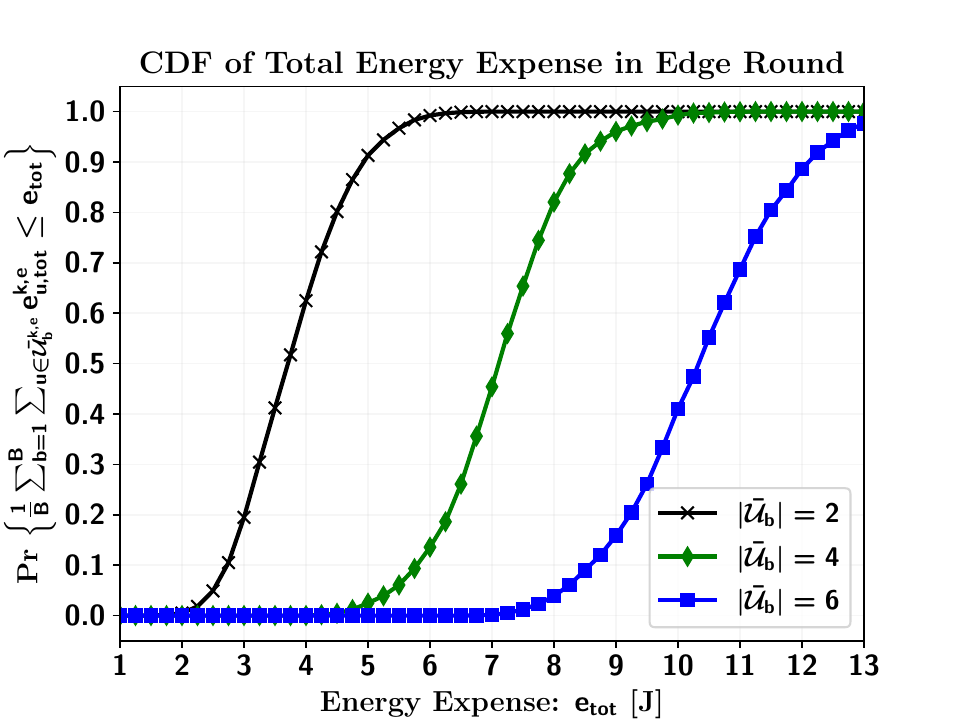}
    \caption{\bblue{CDF of total energy expense in the BS}}
    \label{per_BS_cdfEnergyExpense}
\end{figure}
\begin{figure}[!t]
    \centering
    \includegraphics[trim=5 1 25 25, clip, width=0.4\textwidth]{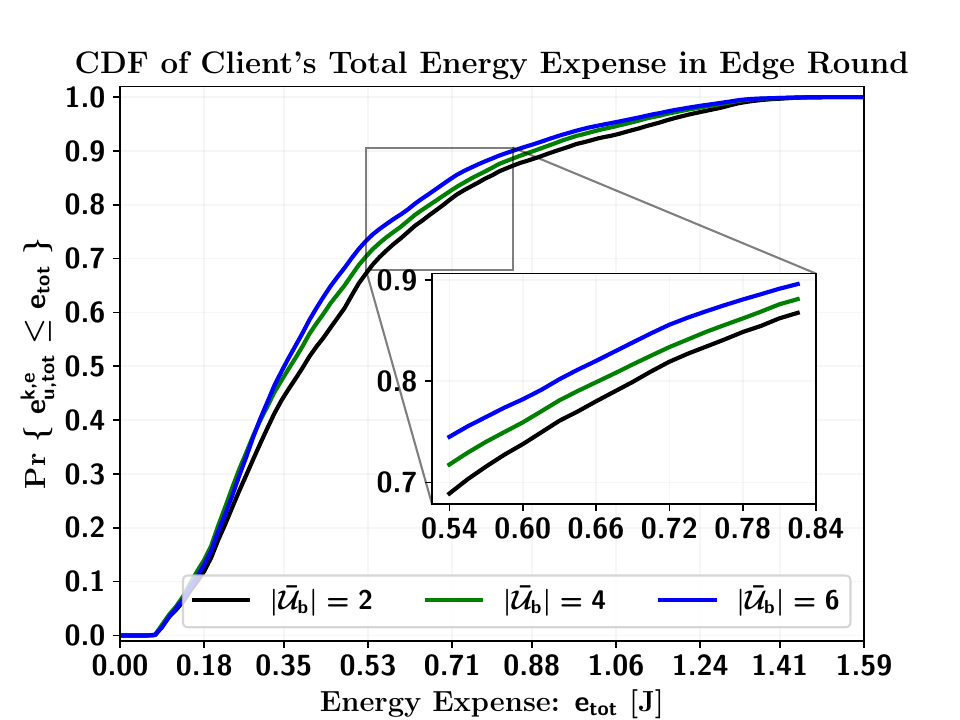}
    \caption{\bblue{CDF of client's total energy expense}}
    \label{per_UE_cdfEnergyExpense}
\end{figure}

\subsection{Performance Comparisons}


\noindent
To the best of our knowledge, no previous studies exactly considered our system design and used \ac{hfl} for video caching.
Therefore, we use the traditional \ac{hfedavg} algorithm \cite{liu2020client} and modify it to accommodate our system model for comparison. 
We consider two modifications. 
The first one calculates the smallest \ac{sgd} rounds that all clients can perform without violating any constraints.
We call this modified baseline \ac{hfedavg}-M$1$.
There are few clients who cannot even perform a single \ac{sgd} round.
These clients are dropped, and then the minimum \ac{sgd} rounds that the retained clients can all perform without violating any constraints are calculated in the second modification.
We call this baseline \ac{hfedavg}-M$2$.
Furthermore, we also consider the ideal case, where no constraints are enforced. 
This one is termed \ac{hfedavg}-UB and is indeed the upper bound of \ac{hfedavg} \cite{liu2020client}.
Finally, we also consider the traditional single server-based \ac{fl}, namely, \ac{fedavg} \cite{mcmahan2017communication}. 
Concretely, in this baseline, the \acp{bs} merely act as relays between the central server and the clients.
That is in every global round, the clients receive the same global model from the central server via their respective serving \acp{bs}.
Upon finishing the local \ac{sgd} steps, the clients offload the accumulated gradients to the central server via their serving \acp{bs}.  
Similar to \ac{hfedavg}-UB, we do not enforce any constraints in this baseline, and this baseline is termed as \ac{fedavg}-UB.
Note that these baselines require $Z=|\mathcal{U}_b|$ pRBS for clients' gradient sharing.
Besides, \ac{fedavg}-UB causes $s \times |\mathcal{U}_b|$ bits payloads in the \ac{bs}-\ac{cs} backhaul in every global round.
Note that the \ac{hfl} algorithm has only $s$ bits of payload for the \ac{bs}-\ac{cs} backhaul.


\ac{hfedavg} should work poorly in constrained cases, since some clients may be unable to participate in the training process due to time or energy overheads beyond the allowable budgets. 
This is particularly an impediment for the \ac{hfedavg}-M$1$ baseline, as all clients are required to train the same number of local \ac{sgd} rounds.
Therefore, if any clients fail to participate in model training during an edge round, we discard performing model training during that edge round.
This extreme case is alleviated in the \ac{hfedavg}-M$2$ baseline, where we first drop the stragglers.
As such, it is expected that \ac{hfedavg}-M$1$ will perform poorly compared to the \ac{hfedavg}-M$2$ baseline.
Furthermore, the \ac{hfedavg}-UB baseline should provide the best possible performance, as all clients would train their local models for the maximum allowable number of \ac{sgd} rounds in a perfect environment.
Moreover, \ac{fedavg}-UB is expected to take additional rounds to catch up on the performance of \ac{hfedavg}-UB since it does not have any intermediate edge aggregations. 
However, in reality, the corresponding energy overheads can be much higher than some clients' energy budgets, which we ignore to obtain the upper-bounded performances of \ac{fedavg}-UB and \ac{hfedavg}-UB. 
Compared to these baselines, the proposed \ac{rawhfl} selects a subset of the clients and minimizes the energy expense, since our joint objective function is designed as the weighted combination of energy expense and total local iterations of the clients.


Our simulation results in Fig. \ref{accEnergy_epsMin_0.1_epsMax_0.8} also reflect these trends. 
More specifically, we observe that the test accuracies of \ac{hfedavg} baselines in the constrained case lag significantly compared to the \ac{hfedavg}-UB.
Besides, \ac{hfedavg}-M$1$ and \ac{hfedavg}-M$2$ baselines do not reach a plateau even after $K=100$ global rounds.
The \ac{fedavg}-UB baseline, on the other hand, reaches the \ac{hfedavg}-UB's test accuracy after around \bblue{$60$} global rounds. 
Furthermore, the proposed \ac{rawhfl}---with $|\bar{\mathcal{U}}_b^{k,e}|=4$ clients selected per \ac{bs}---performs nearly identically to the \ac{hfedavg}-UB baselines in terms of test accuracy.
However, it is evident that the proposed \ac{rawhfl} requires significantly lower energy to finish $K=100$ global rounds.
Particularly, the energy expense grows linearly as the global round increases with significantly steep slopes for the \ac{hfedavg}-M$2$ and \ac{hfedavg}-UB baselines.
Besides, it is expected that \ac{hfedavg}-M$1$ has the lowest energy expense since there are many stragglers due to extreme resource constraints that lead to no model training in many edge rounds.
Furthermore, the energy expense of \ac{fedavg}-UB is also significantly lower than the \ac{hfedavg}-UB since \ac{fedavg}-UB does not have any edge aggregations.
However, recall that the \ac{fedavg}-UB causes additional communication burdens for the \ac{bs}-\ac{cs} backhaul. 
For example, when \bblue{$K=50$}, the test accuracies of \ac{hfedavg}-M$1$, \ac{hfedavg}-M$1$, \ac{hfedavg}-UB, \ac{fedavg}-UB and \ac{rawhfl} are about \bblue{$7.3\%$, $17.3\%$, $44.55\%$, $44.55\%$ and $44.30\%$}, respectively, with energy overheads of \bblue{$283$, $6845$, $8439$, $2109$ and $1741$} Joules.
Moreover, the proposed \ac{rawhfl} delivers nearly identical test accuracy as of the \ac{hfedavg}-UB after $K=100$ global rounds with \bblue{about $4.85$ times} lower energy overheads of \ac{hfedavg}-UB.


We now also compare the results with centralized \ac{sgd} and naive popularity-based Top-popular baselines.
Note that in centralized \ac{sgd}, we assume all datasets are available centrally to show the performance gap with \ac{rawhfl}.
In particular, we consider the Top-$M$ accuracy metric, which we define as the probability that the true requested content is within the Top-$M$ predicted content. 
Intuitively, the Top-$M$ accuracy should improve with the increase of $M$, as the actual requested content has a higher chance of being one of the Top-$M$ predicted content.
Moreover, while the \ac{fl} algorithms are expected to exhibit the performances as observed in the above discussions, the centralized \ac{sgd} is the ideal performance since the \ac{ml} model gets exposed to the datasets of all clients.
Furthermore, since the clients request content following a popularity-preference tradeoff, the Top-Popular baseline is expected to perform worse.


Our simulation results in Fig. \ref{topKacc_epsMin_0.1_epsMax_0.8} also validate the above claims. 
Note that the solid lines show the mean Top-$M$ accuracy while the shaded areas show the standard deviations of the test accuracies across all $U=48$ clients in Fig. \ref{topKacc_epsMin_0.1_epsMax_0.8}.
Particularly, our solution yields about \bblue{$36.94\%$, $25.67\%$ and $35.69\%$} higher (Top-$1$) test accuracies than \ac{hfedavg}-M$1$, \ac{hfedavg}-M$2$ and Top-Popular baselines, respectively.
Besides, \ac{rawhfl}, \ac{hfedavg}-UB and \ac{fedavg}-UB have nearly identical test accuracies for all $M$'s.
Furthermore, centralized \ac{sgd} has about \bblue{$6.95\%$} higher (Top-$1$) test accuracy than \ac{rawhfl}.
Moreover, for $M \leq 7$, we see the superior performance of the centralized \ac{sgd} as expected.
The \ac{rawhfl}, \ac{hfedavg}-UB and \ac{fedavg}-UB baselines can catch up with the performance of the centralized \ac{sgd} when $M \geq 8$.
The results also reveal that the actual content is almost surely within the Top-$8$ predicted contents, which can significantly help in designing the cache placement.
Moreover, for all $M$'s, the other \ac{hfedavg} baselines perform poorly as those are not suitable for a resource-constrained environment. 
Finally, naive popularity-based prediction does not work since the content requests are modeled following an exploration-exploitation tradeoff.

\begin{figure}[!t]
\centering
    \includegraphics[trim=15 15 15 15, clip, width=0.45\textwidth]{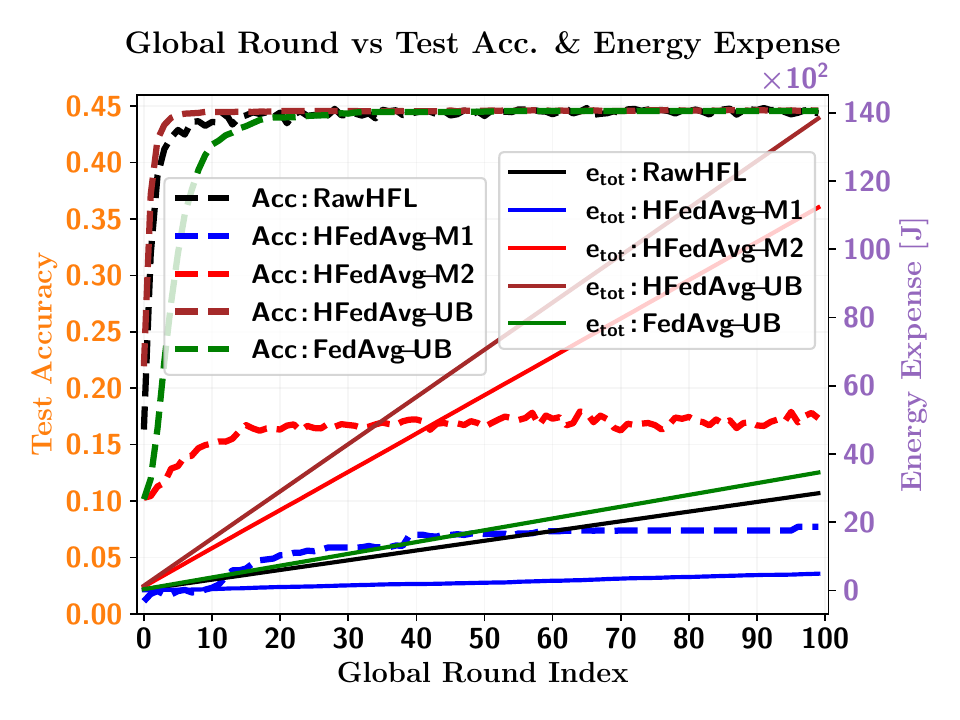}
    \caption{\bblue{Baseline comparisons: FL round vs. accuracy and energy expense ($\upsilon_u \in [0.1, 0.8]$)}}
    \label{accEnergy_epsMin_0.1_epsMax_0.8}
\end{figure}
\begin{figure}[!t]
\centering
    \includegraphics[trim=12 8 45 30, clip, width=0.42\textwidth]{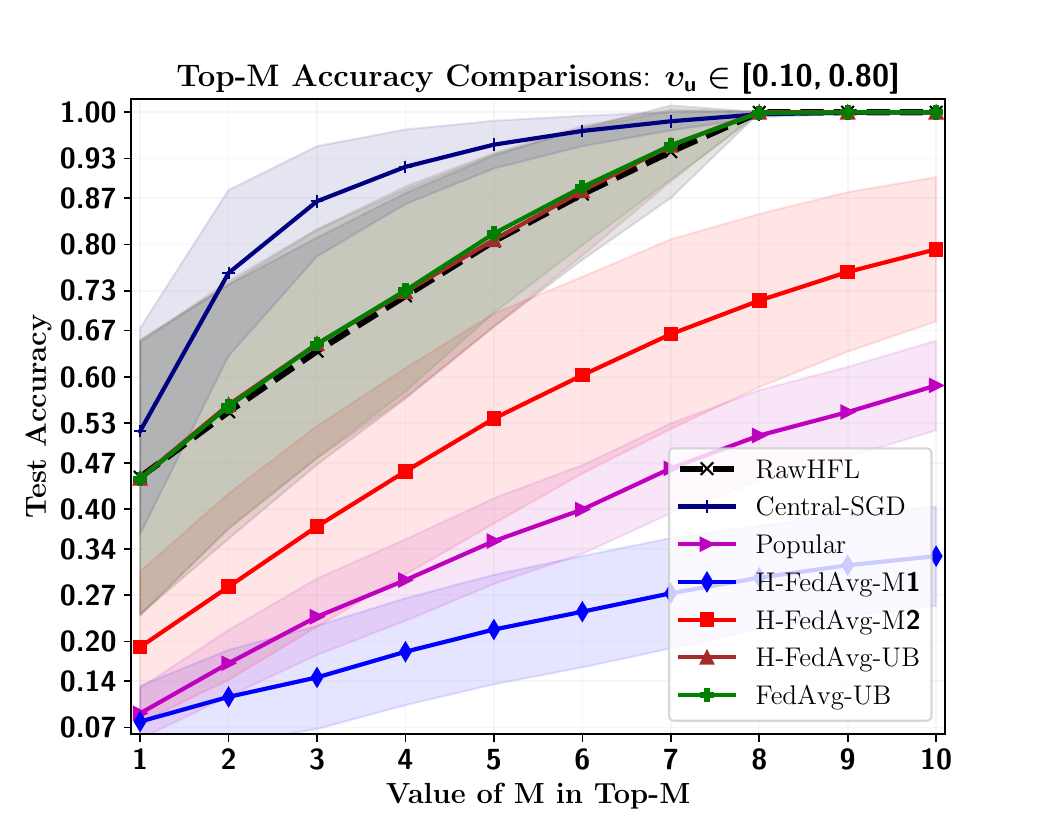}
    \caption{\bblue{Top-$M$ accuracy comparisons: probability that requested content is in the Top-$M$ predicted content ($\upsilon_u \in [0.1, 0.8]$)}}
    \label{topKacc_epsMin_0.1_epsMax_0.8}
\end{figure}
\begin{table}[!t]
\centering
\caption{Performance comparison: $\!K\!\!=\!400$, $E\!=\!4$, $\mathrm{t_{th}}\!\!=\!150$ s, $\upsilon_u=0.3$}
\small
\begin{tabular}{|C{2.5cm}|C{2.4cm}|C{2.7cm}|} \hline 
    \textbf{FL Algorithm} & \textbf{Test Accuracy} & \textbf{Energy Expense [J]}  \\ \hline 
    RawHFL-$|\Bar{\mathcal{U}}_b|=2$ & $0.2994 \pm 0.0143$ & $1508.33 \pm 205.47$ \\ \hline 
    RawHFL-$|\Bar{\mathcal{U}}_b|=4$ & $0.2998 \pm 0.0146$ &  $2848.16 \pm 273.85$ \\ \hline 
    RawHFL-$|\Bar{\mathcal{U}}_b|=6$ &  $0.2991 \pm 0.0155$ & $4137.94 \pm 401.28$ \\ \hline 
    Central-SGD & $0.3947 \pm 0.0856$ & N/A \\ \hline  
    H-FedAvg-M$1$ & $0.0663 \pm 0.0352$ & $490.23 \pm 1391.70$ \\ \hline 
    H-FedAvg-M$2$ & $0.1591 \pm 0.0906$ & $11226.51 \pm 413.97$ \\ \hline
    H-FedAvg-UB \cite{liu2020client} & $0.2977 \pm 0.0161$ & $13836.46 \pm 519.91$ \\ \hline
    FedAvg-UB \cite{mcmahan2017communication} & $0.2964 \pm 0.0170$ & $3455.56 \pm 131.95$ \\ \hline
    Top-Popular & $0.0978 \pm 0.0395$ & N/A  \\ \hline
\end{tabular}
\label{perCompTable_epsMin_0.3_epsMax_0.3}
\end{table}
\begin{table}[!t]
\centering
\caption{Performance comparison: $\!K\!\!=\!400$, $E\!=\!4$, $\mathrm{t_{th}}\!\!=\!150$ s, $\upsilon_u=0.7$}
\small
\begin{tabular}{|C{2.5cm}|C{2.4cm}|C{2.7cm}|} \hline 
    \textbf{FL Algorithm} & \textbf{Test Accuracy} & \textbf{Energy Expense [J]}  \\ \hline 
    RawHFL-$|\Bar{\mathcal{U}}_b|=2$ & $0.7013 \pm 0.0142$ &  $1508.33 \pm 205.47$ \\ \hline 
    RawHFL-$|\Bar{\mathcal{U}}_b|=4$ &  $0.7013 \pm 0.0142$  & $2848.16 \pm 273.85$ \\ \hline 
    RawHFL-$|\Bar{\mathcal{U}}_b|=6$ & $0.7013 \pm 0.0142$ & $4137.94 \pm 401.28$ \\ \hline 
    Central-SGD & $0.7013 \pm 0.0141$ & N/A \\ \hline  
    H-FedAvg-M$1$ & $0.1289 \pm 0.0309$ & $490.23 \pm 1391.70$ \\ \hline 
    H-FedAvg-M$2$ & $0.3483 \pm 0.0636$ & $11226.51 \pm 413.97$ \\ \hline
    H-FedAvg-UB \cite{liu2020client} & $0.7013 \pm 0.0142$ & $13836.46 \pm 519.91$ \\ \hline
    FedAvg-UB \cite{mcmahan2017communication} & $0.7013 \pm 0.0142$ & $3455.56 \pm 131.95$ \\ \hline
    Top-Popular & $0.0591 \pm 0.0322$ & N/A  \\ \hline
\end{tabular}
\label{perCompTable_epsMin_0.7_epsMax_0.7}
\end{table}
\begin{table}[!t]
\centering
\caption{Performance comparison: $\!K\!\!=\!400$, $E\!=\!4$, $\mathrm{t_{th}}\!\!=\!150$ s, $\upsilon_u \in [0.1, 0.8]$}
\small
\begin{tabular}{|C{2.5cm}|C{2.4cm}|C{2.7cm}|} \hline 
    \textbf{FL Algorithm} & \textbf{Test Accuracy} & \textbf{Energy Expense [J]}  \\ \hline 
    RawHFL-$|\Bar{\mathcal{U}}_b|=2$ & $0.4485 \pm 0.2072$ & $1508.33 \pm 205.47$ \\ \hline 
    RawHFL-$|\Bar{\mathcal{U}}_b|=4$ & $0.4485 \pm 0.2075$ & $2848.16 \pm 273.85$ \\ \hline 
    RawHFL-$|\Bar{\mathcal{U}}_b|=6$ & $0.4484 \pm 0.2075$ & $4137.94 \pm 401.28$ \\ \hline 
    Central-SGD & $0.5180 \pm 0.1558$ & N/A \\ \hline  
    H-FedAvg-M$1$ & $0.0791 \pm 0.0542$ & $490.23 \pm 1391.70$ \\ \hline 
    H-FedAvg-M$2$ & $0.1918 \pm 0.1142$ & $11226.51 \pm 413.97$\\ \hline
    H-FedAvg-UB \cite{liu2020client} & $0.4468 \pm 0.2069$ & $13836.46 \pm 519.91$ \\ \hline
    FedAvg-UB \cite{mcmahan2017communication} & $0.4457 \pm 0.2061$ & $3455.56 \pm 131.95$ \\ \hline
    Top-Popular &  $0.0916 \pm 0.0386$ & N/A  \\ \hline
\end{tabular}
\label{perCompTable_epsMin_0.1_epsMax_0.8}
\end{table}


To that end, we examine how these baselines perform for different $\upsilon_u$.
As such, we list the test accuracy and energy expense of these baselines in Tables \ref{perCompTable_epsMin_0.3_epsMax_0.3} -  \ref{perCompTable_epsMin_0.1_epsMax_0.8}.
These results reveal that while \ac{rawhfl} achieves nearly identical test accuracies compared to the \ac{hfedavg}-UB baseline, it is more energy efficient.
Besides, when $\vert \bar{\mathcal{U}}_b \vert$ is high, \ac{fedavg}-UB may require less energy than \ac{rawhfl}. 
However, recall that \ac{fedavg}-UB has $(\vert \mathcal{U}_b \vert - 1)$ times higher communication overhead in the \ac{bs}-\ac{cs} backhaul links. 
Moreover, although centralized \ac{sgd} outperforms the \ac{fl} algorithms in all examined scenarios, the test accuracy gap with our proposed \ac{rawhfl} depends on the value of the $\upsilon_u$.
Particularly, this gap is higher when $\upsilon_u$ is small and decreases as $\upsilon_u$ increases. 
This is due to the fact that a higher $\upsilon_u$ increases more personalized content requests from the client's preferred genre, leading to a clear trend in their datasets.
However, a smaller $\upsilon_u$ leads to more exploration of different genres, leading to irregular request patterns that are hard to predict.
Since the centralized \ac{sgd} can access all clients' training datasets, the severity of these irregularities can be predicted better with centralized \ac{sgd}. 
Furthermore, it is worth noting that the performance gaps between the centralized \ac{sgd} and the \ac{fl} algorithms in the above results indicate the impact of the clients' non-IID data distributions. 
Moreover, since the non-IID data distribution problem typically requires special treatments and belongs to a different research direction \cite{zhu2021datafree}, the solution to this problem in wireless video caching networks deserves a separate study, which we leave as our future works.

\section{Conclusion}
\label{conclusion}
\noindent
Considering a realistic dataset acquisition method, this paper proposed a privacy-preserving \ac{rawhfl} solution for predicting users' to-be-requested content in resource-constrained wireless video caching networks.
Our theoretical analysis revealed that the convergence bound of the proposed algorithm depends on the \ac{fl} and wireless networking parameters, as these parameters determine whether the \ac{es} receives the trained gradients from the clients. 
Besides, the formulated weighted utility function enabled energy-efficient training of the proposed \ac{rawhfl} by jointly selecting clients, local training rounds, \ac{cpu} frequencies and uplink wireless transmission power.
Furthermore, our extensive simulation results indicate that more clients yield better test performances at the expense of higher energy costs.
Moreover, our solution delivers nearly identical performance to the ideal case with significantly lower energy costs.
Finally, we will leverage the prediction results of \ac{rawhfl} in our future work and design an appropriate cache placement policy.

\section*{Acknowledgement}
\noindent
The authors thank Prof. Minseok Choi, Omer Gokalp Serbetci and Yijing Zhang for helpful discussions.

\medskip

\noindent
The authors acknowledge the Center for Advanced Research Computing (CARC) at the University of Southern California for providing computing resources that have contributed to the research. URL: \url{https://carc.usc.edu}.

\bibliography{Reference.bib}
\bibliographystyle{IEEEtran}

 \onecolumn
 \appendices
\section*{Supplementary Material}
\subsection{Key Equations}
\noindent
Recall that the \acp{ue} minimizes the following local loss function.
\begin{equation}
	\label{localLossFunc_apndx}
	\begin{aligned}
		f_u (\mathbf{w} | \mathcal{D}_{u,\mathrm{proc}}^t) \coloneqq [1/\mathrm{D}_{u,\mathrm{proc}}^t] \sum\nolimits_{(\mathbf{x}_{u}^a, \mathbf{y}_u^a) \in \mathcal{D}_{u,\mathrm{proc}}^t} \mathrm{l}\left(\mathbf{w} | (\mathbf{x}_{u}^a, \bblue{y_u^a}) \right),
	\end{aligned}
\end{equation}
where $\mathrm{l}\left(\mathbf{w} | (\mathbf{x}_{u}^a, \bblue{y_u^a})\right)$ is the loss associated with the $a^{\mathrm{th}}$ data sample.
The \ac{ue} minimizes (\ref{localLossFunc_apndx}) by taking $\mathrm{L}_u^{k,e}$ \ac{sgd} steps as 
\begin{equation}
	\label{localModUp_apndx}
	\begin{aligned}
		\mathbf{w}_{u}^{k,e,\mathrm{L}_{u}^{k,e}} 
		&= \mathbf{w}_{u}^{k, e, 0} - \eta \sum\nolimits_{l=0}^{\mathrm{L}_{u}^{k,e}-1} g_u(\mathbf{w}_{u}^{k,e,l}),
	\end{aligned}
\end{equation} 
where $\eta$ is the learning rate and $g_u(\mathbf{w}_{u}^{k,e,l})$ is the stochastic gradient.

Besides, the \acp{es} want to minimize the following loss
\begin{equation}
	\label{esLossRawHFL_apndx}
	f_{b} (\mathbf{w} | \cup_{u \in \bar{\mathcal{U}}_b^{k,e}} \mathcal{D}_{u, \mathrm{proc}}^t ) \coloneqq \sum\nolimits_{u \in \bar{\mathcal{U}}_b^{k,e}} \alpha_u f_u (\mathbf{w} | \mathcal{D}_{u,\mathrm{proc}}^t),
\end{equation}
where $\sum_{u \in \bar{\mathcal{U}}_b^{k,e}} \alpha_u = 1$.
Each \ac{es} follows the following aggregation rule.
\begin{equation}
	\label{edgeUpdateRule_apndx}
	\begin{aligned}
		\mathbf{w}_b^{k, e+1} 
		&= \mathbf{w}_{b}^{k,e} - \eta \sum\nolimits_{u \in \bar{\mathcal{U}}_b^{k,e}} \alpha_u \big[\mathrm{1}_{u,\mathrm{sc}}^{k,e}/\mathrm{p}_{u,\mathrm{sc}}^{k,e}\big]
		\tilde{g}_u^{k,e},
	\end{aligned}
\end{equation}
where $\tilde{g}_u^{k,e} \coloneqq \sum_{l=0}^{\mathrm{L}_u^{k,e}-1} g_u(\mathbf{w}_{u}^{k,e,l})$ and $\mathrm{1}_{u,\mathrm{sc}}^{k,e}$ is defined as 
\begin{equation}
	\begin{aligned}
		\mathrm{1}_{u,\mathrm{sc}}^{k,e} &= 
		\begin{cases}
			1, &\text{with probability } \mathrm{p}_{u, \mathrm{sc}}^{k,e},\\
			0, &\text{otherwise},
		\end{cases}.
	\end{aligned}
\end{equation}

The central server wants to minimize the following loss function
\begin{equation}
	\label{centralLossRawHFL_apndx}
	f(\mathbf{w}^k | \cup_{b=0}^{B-1} \cup_{u \in \bar{\mathcal{U}}_b^{k,e}} \mathcal{D}_{u,\mathrm{proc}}^t) \coloneqq \sum\nolimits_{b=0}^{B-1} \alpha_b f_{b} \big(\mathbf{w}^k | \cup_{u \in \bar{\mathcal{U}}_b^{k,e}} \mathcal{D}_{u, \mathrm{proc}}^t \big),
\end{equation}
where $\sum_{b=0}^{B-1}\alpha_b = 1$.
Moreover, the central server uses the following aggregation rule.
\begin{equation}
	\label{globalUpdateRule_apndx}
	\begin{aligned}
		\mathbf{w}^{k+1} 
		&= \sum_{b=0}^{B-1}  \alpha_b \mathbf{w}_b^{k, E}   
		= \mathbf{w}^k - \eta \sum\nolimits_{e=0}^{E-1} \sum\nolimits_{b=0}^{B-1} \alpha_b \sum\nolimits_{u \in \bar{\mathcal{U}}_b^{k,e}} \alpha_u \big[\mathrm{1}_{u,\mathrm{sc}}^{k,e} / \mathrm{p}_{u,\mathrm{sc}}^{k,e}\big] \tilde{g}_u^{k,e}.
	\end{aligned}   
\end{equation}

In the following, we have used $f_u(\mathbf{w})$, $f_b(\mathbf{w})$ and $f(\mathbf{w})$ to represent $f_u(\mathbf{w}| \mathcal{D}_{u,\mathrm{proc}}^t)$, $f_b(\mathbf{w} | \cup_{u \in \bar{\mathcal{U}}_b^{k,e}} \mathcal{D}_{u,\mathrm{proc}}^t)$ and $f(\mathbf{w} | \cup_{b=0}^{B-1} \cup_{u \in \bar{\mathcal{U}}_b^{k,e}} \mathcal{D}_{u,\mathrm{proc}}^t)$, respectively, for notational simplicity.
Moreover, from the definition of the loss functions, it is easy to notice that $\nabla f_b(\mathbf{w}) = \sum_{u \in \bar{\mathcal{U}}_b^{k,e}}\alpha_u \nabla f_u(\mathbf{w})$ and $\nabla f(\mathbf{w}) = \sum_{b=0}^{B-1} \alpha_b \nabla f_b(\mathbf{w}) = \sum_{b=0}^{B-1} \alpha_b \sum_{u \in \bar{\mathcal{U}}_b^{k,e}}\alpha_u \nabla f_u(\mathbf{w})$.

\subsection{Key Assumptions}
\noindent 
We make the following assumptions for our theoretical analysis, which are standard in the literature \cite{pervej2024hierarchical, wang2022demystifying, hosseinalipour2023parallel,pervej2023Resource}.
\begin{Assumption}
	$\beta$-smoothness: The loss functions in all nodes are $\beta$-smooth, i.e., $\Vert \nabla f(\mathbf{w}) - \nabla f(\mathbf{w}')\Vert \leq \beta \Vert \mathbf{w} - \mathbf{w}' \Vert$, where $\Vert \cdot \Vert$ is the $L_2$ norm. 
\end{Assumption}
\begin{Assumption}
	Unbiased \ac{sgd}: mini-batch gradients are unbiased, i.e., $\mathbb{E}_{\xi \sim \mathcal{D}_{u,\mathrm{proc}}} \left[ g_u(\mathbf{w}) \right] = \nabla f_u(\mathbf{w})$, where $\mathbb{E} [\cdot]$ is the expectation operator, and $\xi$ is client's randomly sampled mini-batch. 
\end{Assumption}
\begin{Assumption}
	Bounded variance: variance of the gradients is bounded, i.e., $\Vert g_u(\mathbf{w}) - \nabla f_u(\mathbf{w}) \Vert^2 \leq \sigma^2$.
\end{Assumption}
\begin{Assumption}
	Independence: a) the \textit{stochastic gradients} are independent of each other in different episodes and b) \textit{accumulated gradient offloading} is independent of client selection and each other in each edge round $e$. 
\end{Assumption}
\begin{Assumption}
	Bounded divergence: divergence between the a) local and edge and b) edge and global loss functions are bounded, i.e., for all $u$, $b$ and $\mathbf{w}$
	\begin{align}
		&\sum\nolimits_{u \in \bar{\mathcal{U}}_b^{k,e}} \alpha_u \Vert \nabla f_u(\mathbf{w}) - \nabla f_b(\mathbf{w}) \Vert^2 \leq \epsilon_0^2,\\
		&\sum_{b=0}^{B-1} \alpha_b \Big\Vert \sum_{u \in \bar{\mathcal{U}}_b^{k,e}} \rs \rs \alpha_u \nabla \tilde{f}_{u} (\mathbf{w}) - \sum_{b'=0}^{B-1} \alpha_{b'} \rs \rs \sum_{u' \in \bar{\mathcal{U}}_{b'}^{k,e}} \rs\rs \alpha_{u'} \nabla \tilde{f}_{u'} (\mathbf{w}) \Big\Vert^2 \leq \epsilon_1^2, \label{ES_Central_Loss_Divergence_apndx}
	\end{align}
	where $\nabla \tilde{f}_u (\mathbf{w}) \coloneqq \sum_{l=0}^{\mathrm{L}_u^{k,e} - 1} \nabla f_u (\mathbf{w})$.
\end{Assumption}

\section{Proof of Theorem \ref{theorem1_apndx}}
\label{theorem1Proof}

\setcounter{Theorem}{0}
\begin{Theorem}
	\label{theorem1_apndx}
	Suppose $\eta < \mathrm{min}\left\{\frac{1}{2\sqrt{5} \beta \mathrm{L}}, \frac{1}{\beta E \mathrm{L}} \right\}$ and the above assumptions hold.
	Then, the average global gradient norm from $K$ global rounds of \ac{rawhfl} is upper-bounded as
	\begin{align}
		\label{theorem1_eqn_apndx}
		& \frac{1}{K} \sum_{k=0}^{K-1} \mathbb{E} \big[\Vert \nabla f(\mathbf{w}^k) \Vert^2 \big]
		\leq \frac{2}{\eta K} \sum_{k=0}^{K-1} \frac{1}{\Omega^k} \Big\{ \mathbb{E} [ f(\mathbf{w}^k) ] - \mathbb{E} [ f(\mathbf{w}^{k+1}) ] \Big\} + \frac{2 \beta \eta \mathrm{L} \sigma^2}{K} \sum_{k=0}^{K-1} \frac{\mathrm{N}_1^k}{\Omega^k} + \frac{18 E \beta^2 \epsilon_0^2 \eta^2 \mathrm{L}^3}{K} \sum_{k=0}^{K-1} \frac{\mathrm{N}_2}{\Omega^k} + \nonumber\\
		&\squad \frac{20 \mathrm{L} \beta^2 \epsilon_1^2 \eta^2 E^3}{K} \sum_{k=0}^{K-1} \frac{1}{\Omega^k} + \frac{2 \beta \eta \mathrm{L}}{K} \sum_{k=0}^{K-1} \frac{1}{\Omega^k} \sum_{e=0}^{E-1} \sum_{b=0}^{B-1} \alpha_b  \sum\nolimits_{u \in \bar{\mathcal{U}}_b^{k,e}} \alpha_u \mathrm{N}_u  \big[(1/\mathrm{p}_{u,\mathrm{sc}}^{k,e}) - 1 \big] \sum\nolimits_{l=0}^{\mathrm{L}_u^{k,e} - 1} \mathbb{E} \big[\big\Vert g_u (\mathbf{w}_u^{k, e, l}) \big\Vert^2 \big],
	\end{align}
	where the expectations depend on clients' randomly selected mini-batches and $\mathrm{1}_{u,\mathrm{sc}}^{k,e}$'s. 
	Besides, $\Omega^k \coloneqq \sum_{e=0}^{E-1} \sum_{b=0}^{B-1} \alpha_b \sum_{u \in \bar{\mathcal{U}}_b^{k,e}} \alpha_u \mathrm{L}_{u}^{k,e}$, $\mathrm{N}_1^k \coloneqq  60 \beta^3 \eta^3 E^3 \mathrm{L}^3 + 3 \beta \eta E \mathrm{L} +  \sum_{e=0}^{E-1} \sum_{b=0}^{B-1} \alpha_b  \left(\alpha_b + 4 E \mathrm{L} \beta \eta \right) \sum_{u \in \bar{\mathcal{U}}_b^{k,e}} \left(\alpha_u\right)^2$, $\mathrm{N}_2 \coloneqq 1 + 20 \beta^2 \eta^2 E^2 \mathrm{L}^2$ and $\mathrm{N}_{u} \coloneqq  E + 3 \beta \eta \mathrm{L} + 4 \beta \eta E \left(\alpha_u + 15 E \beta^2 \eta^2 \mathrm{L}^3 \right)$.
\end{Theorem}

\begin{proof}
	\noindent
	From our aggregation rule in (\ref{globalUpdateRule_apndx}), we write the following
	\begin{equation}
		\label{convProofMainEq}
		\begin{aligned}
			&f(\mathbf{w}^{k+1}) 
			= f \bigg(\mathbf{w}^k - \eta \sum_{e=0}^{E-1} \sum_{b=0}^{B-1} \alpha_b\sum_{u \in \bar{\mathcal{U}}_b^{k,e}} \alpha_u \frac{\mathrm{1}_{u,\mathrm{sc}}^{k,e}}{\mathrm{p}_{u,\mathrm{sc}}^{k,e}} \tilde{g}_u^{k,e}\bigg)\\
			&\overset{(a)}{\leq} f(\mathbf{w}^k) + \bigg<\nabla f(\mathbf{w}^k), - \eta \sum_{e=0}^{E-1} \sum_{b=0}^{B-1} \alpha_b\sum_{u \in \bar{\mathcal{U}}_b^{k,e}} \alpha_u \frac{\mathrm{1}_{u,\mathrm{sc}}^{k,e}}{\mathrm{p}_{u,\mathrm{sc}}^{k,e}} \tilde{g}_u^{k,e}\bigg> + \frac{\beta \eta^2}{2} \bigg\Vert \sum_{e=0}^{E-1} \sum_{b=0}^{B-1} \alpha_b\sum_{u \in \bar{\mathcal{U}}_b^{k,e}} \alpha_u \frac{\mathrm{1}_{u,\mathrm{sc}}^{k,e}}{\mathrm{p}_{u,\mathrm{sc}}^{k,e}} \tilde{g}_u^{k,e} \bigg\Vert^2 \\
			&= f(\mathbf{w}^k) - \eta \bigg<\nabla f(\mathbf{w}^k), \sum_{e=0}^{E-1} \sum_{b=0}^{B-1} \alpha_b\sum_{u \in \bar{\mathcal{U}}_b^{k,e}} \alpha_u \frac{\mathrm{1}_{u,\mathrm{sc}}^{k,e}}{\mathrm{p}_{u,\mathrm{sc}}^{k,e}} \sum_{l=0}^{\mathrm{L}_{u}^{k,e}-1} g_u\big(\mathbf{w}_u^{k, e,l}\big) \bigg> + \frac{\beta \eta^2}{2} \bigg\Vert \sum_{e=0}^{E-1} \sum_{b=0}^{B-1} \alpha_b\sum_{u \in \bar{\mathcal{U}}_b^{k,e}} \alpha_u \frac{\mathrm{1}_{u,\mathrm{sc}}^{k,e}}{\mathrm{p}_{u,\mathrm{sc}}^{k,e}} \tilde{g}_u^{k,e} \bigg\Vert^2,
		\end{aligned}   
	\end{equation}
	where $(a)$ stems from $\beta$-smoothness property, i.e., $f(\mathbf{y}) \leq f(\mathbf{x}) + \left<\nabla f(\mathbf{x}), \mathbf{y} - \mathbf{x} \right> + \frac{\beta}{2}\Vert \mathbf{y} - \mathbf{x} \Vert^2$. 
	
	Now, taking expectation over both sides of (\ref{convProofMainEq}), we get 
	\begin{equation}
		\label{convProofMainEq_1}
		\begin{aligned}
			\mathbb{E} \left[ f(\mathbf{w}^{k+1}) \right] 
			&\leq ~ \mathbb{E} [f(\mathbf{w}^k)] - \underbrace{\eta \sum_{e=0}^{E-1} \sum_{b=0}^{B-1} \alpha_b\sum_{u \in \bar{\mathcal{U}}_b^{k,e}} \alpha_u \sum_{l=0}^{\mathrm{L}_{u}^{k,e}-1}  \mathbb{E} \left[\bigg<\nabla f(\mathbf{w}^k), \frac{\mathrm{1}_{u,\mathrm{sc}}^{k,e}}{\mathrm{p}_{u,\mathrm{sc}}^{k,e}} g_u\big(\mathbf{w}_u^{k, e,l}\big) \bigg>\right]}_{\mathrm{T}_1} + \\
			&\bquad\bquad\squad \underbrace{\frac{\beta \eta^2}{2} \mathbb{E} \bigg[\bigg\Vert \sum_{e=0}^{E-1} \sum_{b=0}^{B-1} \alpha_b\sum_{u \in \bar{\mathcal{U}}_b^{k,e}} \alpha_u \frac{\mathrm{1}_{u,\mathrm{sc}}^{k,e}}{\mathrm{p}_{u,\mathrm{sc}}^{k,e}} \tilde{g}_u^{k,e} \bigg\Vert^2\bigg]}_{\mathrm{T}_2}, 
		\end{aligned}
	\end{equation}
	where $\mathbb{E}[\cdot]$ depends on $\mathrm{1}_{u,\mathrm{sc}}^{k,e}$'s and clients mini-batches.
	
	Now, we simplify the $\mathrm{T}_1$ term as 
	\begin{align}
		\mathrm{T}_1 
		&\overset{(a)}{=} -\eta \sum_{e=0}^{E-1} \sum_{b=0}^{B-1} \alpha_b\sum_{u \in \bar{\mathcal{U}}_b^{k,e}} \alpha_u \sum_{l=0}^{\mathrm{L}_{u}^{k,e}-1}  \Bigg<\nabla f(\mathbf{w}^k), \frac{\mathbb{E} \left[\mathrm{1}_{u,\mathrm{sc}}^{k,e}\right]}{ \mathrm{p}_{u,\mathrm{sc}}^{k,e}} \mathbb{E} \left[g_u\big(\mathbf{w}_u^{k, e,l}\big)\right] \Bigg> \nonumber\\
		&\overset{(b)}{=} -\eta \sum_{e=0}^{E-1} \sum_{b=0}^{B-1} \alpha_b\sum_{u \in \bar{\mathcal{U}}_b^{k,e}} \alpha_u \sum_{l=0}^{\mathrm{L}_{u}^{k,e}-1} \left<\nabla f(\mathbf{w}^k),  \nabla f_u \big(\mathbf{w}_u^{k, e,l}\big) \right> \nonumber \\
		&= -\frac{\eta}{2} \sum_{e=0}^{E-1} \sum_{b=0}^{B-1} \alpha_b\sum_{u \in \bar{\mathcal{U}}_b^{k,e}} \alpha_u \sum_{l=0}^{\mathrm{L}_{u}^{k,e}-1} \bigg[\left\Vert \nabla f(\mathbf{w}^k)\right\Vert^2 + \left\Vert \nabla f_u \big(\mathbf{w}_u^{k, e,l}\big) \right\Vert^2 - \left\Vert \nabla f(\mathbf{w}^k) - \nabla f_u \big(\mathbf{w}_u^{k, e,l}\big) \right\Vert^2 \bigg] \nonumber\\
		&= \frac{\eta}{2} \sum_{e=0}^{E-1} \sum_{b=0}^{B-1} \alpha_b\sum_{u \in \bar{\mathcal{U}}_b^{k,e}} \alpha_u \sum_{l=0}^{\mathrm{L}_{u}^{k,e}-1} \left\Vert \nabla f(\mathbf{w}^k) - \nabla f_u \big(\mathbf{w}_u^{k, e,l}\big) \right\Vert^2 - \frac{\eta}{2} \sum_{e=0}^{E-1} \sum_{b=0}^{B-1} \alpha_b\sum_{u \in \bar{\mathcal{U}}_b^{k,e}} \alpha_u \mathrm{L}_{u}^{k,e} \left\Vert \nabla f(\mathbf{w}^k)\right\Vert^2 - \nonumber\\
		&\bquad \bquad \mquad \frac{\eta}{2} \sum_{e=0}^{E-1} \sum_{b=0}^{B-1} \alpha_b\sum_{u \in \bar{\mathcal{U}}_b^{k,e}} \alpha_u \sum_{l=0}^{\mathrm{L}_{u}^{k,e}-1} \left
		\Vert \nabla f_u \big(\mathbf{w}_u^{k, e,l}\big) \right\Vert^2,
	\end{align}
	where $(a)$ appears from the independence of successful offloading of $\tilde{g}_u^{k,e}$'s and stochastic gradients. 
	Besides, $(b)$ stems from the fact that $\mathbb{E} [\mathrm{1}_{u,\mathrm{sc}}^{k,e}] = \mathrm{p}_{u,\mathrm{sc}}^{k,e}$ and $\mathbb{E} [g_u(\mathbf{w}_u^{k,e,l})] = \nabla f_u (\mathbf{w}_u^{k,e,l})$.

	We also simplify the $\mathrm{T}_2$ term as follows:
	\begin{align}
		\mathrm{T}_2 
		&\overset{(a)} = \frac{\beta \eta^2}{2} \mathbb{E} \Bigg[\Bigg\Vert \sum_{e=0}^{E-1} \sum_{b=0}^{B-1} \alpha_b\sum_{u \in \bar{\mathcal{U}}_b^{k,e}} \alpha_u \frac{\mathrm{1}_{u,\mathrm{sc}}^{k,e}}{\mathrm{p}_{u,\mathrm{sc}}^{k,e}} \tilde{g}_u^{k,e} - \mathbb{E} \left[ \sum_{e=0}^{E-1} \sum_{b=0}^{B-1} \alpha_b\sum_{u \in \bar{\mathcal{U}}_b^{k,e}} \alpha_u \frac{\mathrm{1}_{u,\mathrm{sc}}^{k,e}}{\mathrm{p}_{u,\mathrm{sc}}^{k,e}} \tilde{g}_u^{k,e} \right] \Bigg\Vert^2\Bigg] + \nonumber \\
		&\bquad \bquad\mquad \frac{\beta \eta^2}{2} \left(\mathbb{E} \left[ \sum_{e=0}^{E-1} \sum_{b=0}^{B-1} \alpha_b\sum_{u \in \bar{\mathcal{U}}_b^{k,e}} \alpha_u \frac{\mathrm{1}_{u,\mathrm{sc}}^{k,e}}{\mathrm{p}_{u,\mathrm{sc}}^{k,e}} \tilde{g}_u^{k,e} \right] \right)^2 \nonumber\\
		&=\frac{\beta \eta^2}{2} \mathbb{E} \Bigg[\Bigg\Vert \sum_{e=0}^{E-1} \sum_{b=0}^{B-1} \alpha_b\sum_{u \in \bar{\mathcal{U}}_b^{k,e}} \alpha_u \frac{\mathrm{1}_{u,\mathrm{sc}}^{k,e}}{\mathrm{p}_{u,\mathrm{sc}}^{k,e}} \tilde{g}_u^{k,e} - \sum_{e=0}^{E-1} \sum_{b=0}^{B-1} \alpha_b\sum_{u \in \bar{\mathcal{U}}_b^{k,e}} \alpha_u \sum_{l=0}^{\mathrm{L}_{u}^{k,e} - 1}\nabla f_u (\mathbf{w}_u^{k,e,l})  \Bigg\Vert^2\Bigg] + \nonumber \\
		&\bquad \bquad\mquad \frac{\beta \eta^2}{2} \left\Vert  \sum_{e=0}^{E-1} \sum_{b=0}^{B-1} \alpha_b\sum_{u \in \bar{\mathcal{U}}_b^{k,e}} \alpha_u \sum_{l=0}^{\mathrm{L}_u^{k,e} - 1} \nabla f_u (\mathbf{w}_u^{k, e, l}) \right\Vert^2 \nonumber\\
		&\overset{(b)}{\leq} \frac{\beta \eta^2}{2} \mathbb{E} \Bigg[\Bigg\Vert \sum_{e=0}^{E-1} \sum_{b=0}^{B-1} \alpha_b\sum_{u \in \bar{\mathcal{U}}_b^{k,e}} \alpha_u \frac{\mathrm{1}_{u,\mathrm{sc}}^{k,e}}{\mathrm{p}_{u,\mathrm{sc}}^{k,e}} \sum_{l=0}^{\mathrm{L}_u^{k,e} - 1} g_u (\mathbf{w}_u^{k, e, l}) \pm \sum_{e=0}^{E-1} \sum_{b=0}^{B-1} \alpha_b\sum_{u \in \bar{\mathcal{U}}_b^{k,e}} \alpha_u \sum_{l=0}^{\mathrm{L}_u^{k,e} - 1} g_u (\mathbf{w}_u^{k, e, l}) - \nonumber\\
		&\bquad \sum_{e=0}^{E-1} \sum_{b=0}^{B-1} \alpha_b\sum_{u \in \bar{\mathcal{U}}_b^{k,e}} \alpha_u \sum_{l=0}^{\mathrm{L}_{u}^{k,e} - 1}\nabla f_u (\mathbf{w}_u^{k,e,l})  \Bigg\Vert^2\Bigg] + \frac{\beta E \eta^2}{2} \sum_{e=0}^{E-1} \sum_{b=0}^{B-1} \alpha_b \sum_{u \in \bar{\mathcal{U}}_b^{k,e}} \alpha_u \mathrm{L}_u^{k,e} \sum_{l=0}^{\mathrm{L}_u^{k,e} - 1} \left\Vert  \nabla f_u (\mathbf{w}_u^{k, e, l}) \right\Vert^2 \nonumber\\
		&\overset{(c)}{\leq} \frac{\beta \eta^2}{2} \mathbb{E} \Bigg[\Bigg\Vert \sum_{e=0}^{E-1} \sum_{b=0}^{B-1} \alpha_b\sum_{u \in \bar{\mathcal{U}}_b^{k,e}} \alpha_u  \sum_{l=0}^{\mathrm{L}_u^{k,e} - 1} g_u (\mathbf{w}_u^{k, e, l}) \left(\frac{\mathrm{1}_{u,\mathrm{sc}}^{k,e}} {\mathrm{p}_{u,\mathrm{sc}}^{k,e}} - 1 \right) + \sum_{e=0}^{E-1} \sum_{b=0}^{B-1} \alpha_b\sum_{u \in \bar{\mathcal{U}}_b^{k,e}} \alpha_u \sum_{l=0}^{\mathrm{L}_u^{k,e} - 1} \left( g_u (\mathbf{w}_u^{k, e, l}) - \nabla f_u (\mathbf{w}_u^{k,e,l}) \right) \Bigg\Vert^2\Bigg] \nonumber \\
		&\bquad \bquad \mquad + \frac{\beta E \mathrm{L} \eta^2}{2} \sum_{e=0}^{E-1} \sum_{b=0}^{B-1} \alpha_b \sum_{u \in \bar{\mathcal{U}}_b^{k,e}} \alpha_u \sum_{l=0}^{\mathrm{L}_u^{k,e} - 1} \left\Vert  \nabla f_u (\mathbf{w}_u^{k, e, l}) \right\Vert^2 \nonumber\\
		&\leq \beta \eta^2 \mathbb{E} \Bigg[\Bigg\Vert \sum_{e=0}^{E-1} \sum_{b=0}^{B-1} \alpha_b\sum_{u \in \bar{\mathcal{U}}_b^{k,e}} \alpha_u  \sum_{l=0}^{\mathrm{L}_u^{k,e} - 1} g_u (\mathbf{w}_u^{k, e, l}) \left(\frac{\mathrm{1}_{u,\mathrm{sc}}^{k,e}} {\mathrm{p}_{u,\mathrm{sc}}^{k,e}} - 1 \right) \Bigg\Vert^2 \Bigg] + \frac{\beta E \mathrm{L} \eta^2}{2} \sum_{e=0}^{E-1} \sum_{b=0}^{B-1} \alpha_b \sum_{u \in \bar{\mathcal{U}}_b^{k,e}} \alpha_u \sum_{l=0}^{\mathrm{L}_u^{k,e} - 1} \left\Vert  \nabla f_u (\mathbf{w}_u^{k, e, l}) \right\Vert^2 + \nonumber\\
		&\bquad \bquad \beta \eta^2 \mathbb{E} \Bigg[\Bigg\Vert \sum_{e=0}^{E-1} \sum_{b=0}^{B-1} \alpha_b\sum_{u \in \bar{\mathcal{U}}_b^{k,e}} \alpha_u \sum_{l=0}^{\mathrm{L}_u^{k,e} - 1} \left( g_u (\mathbf{w}_u^{k, e, l}) - \nabla f_u (\mathbf{w}_u^{k,e,l}) \right) \Bigg\Vert^2\Bigg]  \nonumber\\
		&\leq \beta E \eta^2 \sum_{e=0}^{E-1} \sum_{b=0}^{B-1} \alpha_b\sum_{u \in \bar{\mathcal{U}}_b^{k,e}} \alpha_u \mathrm{L}_u^{k,e} \sum_{l=0}^{\mathrm{L}_u^{k,e} - 1} \mathbb{E} \Bigg[\Bigg\Vert g_u (\mathbf{w}_u^{k, e, l}) \left(\frac{\mathrm{1}_{u,\mathrm{sc}}^{k,e}} {\mathrm{p}_{u,\mathrm{sc}}^{k,e}} - 1 \right) \Bigg\Vert^2 \Bigg] + \frac{\beta E \mathrm{L} \eta^2}{2} \sum_{e=0}^{E-1} \sum_{b=0}^{B-1} \alpha_b \sum_{u \in \bar{\mathcal{U}}_b^{k,e}} \alpha_u \sum_{l=0}^{\mathrm{L}_u^{k,e} - 1} \left\Vert  \nabla f_u (\mathbf{w}_u^{k, e, l}) \right\Vert^2 + \nonumber\\
		&\bquad \bquad \beta \eta^2 \mathbb{E} \Bigg[\Bigg\Vert \sum_{b=0}^{B-1} \alpha_b \sum_{e=0}^{E-1} \sum_{u \in \bar{\mathcal{U}}_b^{k,e}} \alpha_u \sum_{l=0}^{\mathrm{L}_u^{k,e} - 1} \left( g_u (\mathbf{w}_u^{k, e, l}) - \nabla f_u (\mathbf{w}_u^{k,e,l}) \right) \Bigg\Vert^2\Bigg] \nonumber \\
		&\overset{(d)}{=} \beta E \eta^2 \sum_{e=0}^{E-1} \sum_{b=0}^{B-1} \alpha_b\sum_{u \in \bar{\mathcal{U}}_b^{k,e}} \alpha_u \mathrm{L}_u^{k,e} \sum_{l=0}^{\mathrm{L}_u^{k,e} - 1} \mathbb{E} \left[\left\Vert g_u (\mathbf{w}_u^{k, e, l}) \right\Vert^2 \right] \times \mathbb{E} \Bigg[\bigg(\frac{\mathrm{1}_{u,\mathrm{sc}}^{k,e}} {\mathrm{p}_{u,\mathrm{sc}}^{k,e}} - 1 \bigg)^2\Bigg]  +  \nonumber\\
		& \frac{\beta E \mathrm{L} \eta^2}{2} \sum_{e=0}^{E-1} \sum_{b=0}^{B-1} \alpha_b \sum_{u \in \bar{\mathcal{U}}_b^{k,e}} \alpha_u \sum_{l=0}^{\mathrm{L}_u^{k,e} - 1} \left\Vert  \nabla f_u (\mathbf{w}_u^{k, e, l}) \right\Vert^2 + \beta \eta^2 \mathbb{E} \Bigg\{ \sum_{b=0}^{B-1} \left(\alpha_b\right)^2 \Bigg\Vert \sum_{e=0}^{E-1} \sum_{u \in \bar{\mathcal{U}}_b^{k,e}} \alpha_u \sum_{l=0}^{\mathrm{L}_u^{k,e} - 1} \left( g_u (\mathbf{w}_u^{k, e, l}) - \nabla f_u (\mathbf{w}_u^{k,e,l}) \right) \Bigg\Vert^2 \nonumber\\
		& + \underbrace{\sum_{b=0}^{B-1} \alpha_b \bigg[ \sum_{e=0}^{E-1} \sum_{u \in \bar{\mathcal{U}}_b^{k,e}} \alpha_u \sum_{l=0}^{\mathrm{L}_u^{k,e} - 1} \left( g_u (\mathbf{w}_u^{k, e, l}) - \nabla f_u (\mathbf{w}_u^{k,e,l}) \right) \bigg] \times \sum_{b'=1, b' \neq b}^B \alpha_{b'}^{k} \bigg[ \sum_{e=0}^{E-1} \sum_{u \in \mathcal{U}_{b'}^{k,e}} \alpha_u \sum_{l=0}^{\mathrm{L}_u^{k,e} - 1} \left( g_u (\mathbf{w}_u^{k, e, l}) - \nabla f_u (\mathbf{w}_u^{k,e,l}) \right) \bigg]}_{T_1}  \Bigg\}  \nonumber \\
		&\overset{(e)}{=} \beta E \eta^2 \sum_{e=0}^{E-1} \sum_{b=0}^{B-1} \alpha_b\sum_{u \in \bar{\mathcal{U}}_b^{k,e}} \alpha_u \mathrm{L}_u^{k,e} \sum_{l=0}^{\mathrm{L}_u^{k,e} - 1} \bigg[\frac{1} {\mathrm{p}_{u,\mathrm{sc}}^{k,e}} - 1 \bigg] \mathbb{E} \left[\left\Vert g_u (\mathbf{w}_u^{k, e, l}) \right\Vert^2 \right] + \frac{\beta E \mathrm{L} \eta^2}{2} \sum_{e=0}^{E-1} \sum_{b=0}^{B-1} \alpha_b \sum_{u \in \bar{\mathcal{U}}_b^{k,e}} \alpha_u \sum_{l=0}^{\mathrm{L} - 1} \left\Vert  \nabla f_u (\mathbf{w}_u^{k, e, l}) \right\Vert^2 + \nonumber\\
		&\bquad\bquad \beta \eta^2 \sum_{b=0}^{B-1} \left(\alpha_b\right)^2 \mathbb{E} \Bigg[ \Bigg\Vert \sum_{e=0}^{E-1} \sum_{u \in \bar{\mathcal{U}}_b^{k,e}} \alpha_u \sum_{l=0}^{\mathrm{L}_u^{k,e} - 1} \left( g_u (\mathbf{w}_u^{k, e, l}) - \nabla f_u (\mathbf{w}_u^{k,e,l}) \right) \Bigg\Vert^2 \Bigg]  \nonumber \\
		&\overset{(f)}{=} \beta E \eta^2 \sum_{e=0}^{E-1} \sum_{b=0}^{B-1} \alpha_b\sum_{u \in \bar{\mathcal{U}}_b^{k,e}} \alpha_u \mathrm{L}_u^{k,e} \sum_{l=0}^{\mathrm{L}_u^{k,e} - 1} \bigg[\frac{1} {\mathrm{p}_{u,\mathrm{sc}}^{k,e}} - 1 \bigg] \mathbb{E} \left[\left\Vert g_u (\mathbf{w}_u^{k, e, l}) \right\Vert^2 \right] + \frac{\beta E \mathrm{L} \eta^2}{2} \sum_{e=0}^{E-1} \sum_{b=0}^{B-1} \alpha_b \sum_{u \in \bar{\mathcal{U}}_b^{k,e}} \alpha_u \sum_{l=0}^{\mathrm{L}_u^{k,e} - 1} \left\Vert  \nabla f_u (\mathbf{w}_u^{k, e, l}) \right\Vert^2 + \nonumber\\
		&\bquad\bquad \beta \eta^2 \sum_{e=0}^{E-1} \sum_{b=0}^{B-1} \left(\alpha_b\right)^2 \sum_{u \in \bar{\mathcal{U}}_b^{k,e}}  \left(\alpha_u \right)^2 \sum_{l=0}^{\mathrm{L}_u^{k,e} - 1} \mathbb{E} \Bigg[ \Bigg\Vert g_u (\mathbf{w}_u^{k, e, l}) - \nabla f_u (\mathbf{w}_u^{k,e,l}) \Bigg\Vert^2 \Bigg] \nonumber \\
		&\overset{(g)}{\leq} \beta E \eta^2 \sum_{e=0}^{E-1} \sum_{b=0}^{B-1} \alpha_b\sum_{u \in \bar{\mathcal{U}}_b^{k,e}} \alpha_u \mathrm{L}_u^{k,e} \sum_{l=0}^{\mathrm{L}_u^{k,e} - 1} \bigg[\frac{1} {\mathrm{p}_{u,\mathrm{sc}}^{k,e}} - 1 \bigg] \mathbb{E} \left[\left\Vert g_u (\mathbf{w}_u^{k, e, l}) \right\Vert^2 \right] + \frac{\beta E \mathrm{L} \eta^2}{2} \sum_{e=0}^{E-1} \sum_{b=0}^{B-1} \alpha_b \sum_{u \in \bar{\mathcal{U}}_b^{k,e}} \alpha_u \sum_{l=0}^{\mathrm{L}_u^{k,e} - 1} \left\Vert  \nabla f_u (\mathbf{w}_u^{k, e, l}) \right\Vert^2 + \nonumber\\
		&\bquad\bquad \beta \eta^2 \sum_{e=0}^{E-1} \sum_{b=0}^{B-1} \left(\alpha_b\right)^2 \sum_{u \in \bar{\mathcal{U}}_b^{k,e}}  \left(\alpha_u \right)^2 \sum_{l=0}^{\mathrm{L}_u^{k,e} - 1} \sigma^2 \nonumber\\
		&= \beta \eta^2 \sigma^2 \sum_{e=0}^{E-1} \sum_{b=0}^{B-1} \left(\alpha_b\right)^2 \sum_{u \in \bar{\mathcal{U}}_b^{k,e}}  \left(\alpha_u \right)^2 \mathrm{L}_u^{k,e} + \beta E \eta^2 \sum_{e=0}^{E-1} \sum_{b=0}^{B-1} \alpha_b\sum_{u \in \bar{\mathcal{U}}_b^{k,e}} \alpha_u \mathrm{L}_u^{k,e} \sum_{l=0}^{\mathrm{L}_u^{k,e} - 1} \bigg[\frac{1} {\mathrm{p}_{u,\mathrm{sc}}^{k,e}} - 1 \bigg] \mathbb{E} \left[\left\Vert g_u (\mathbf{w}_u^{k, e, l}) \right\Vert^2 \right] + \nonumber\\
		& \bquad\bquad \frac{\beta E \mathrm{L} \eta^2}{2} \sum_{e=0}^{E-1} \sum_{b=0}^{B-1} \alpha_b \sum_{u \in \bar{\mathcal{U}}_b^{k,e}} \alpha_u \sum_{l=0}^{\mathrm{L}_u^{k,e} - 1} \left\Vert  \nabla f_u (\mathbf{w}_u^{k, e, l}) \right\Vert^2,  
	\end{align}
	where $(a)$ comes from the definition of variance.
	$(b)$ stems from the fact that $\Vert \sum_{j=1}^J \mathbf{x}_j \Vert^2 \leq J \sum_{j=1}^J \Vert \mathbf{x}_j \Vert^2$, and from the convexity of $\Vert \cdot \Vert$ and Jensen inequality $\Vert \sum_{j=1}^J \alpha_j \mathbf{x}_j \Vert^2 \leq \sum_{j=1}^J \alpha_j \Vert \mathbf{x}_j \Vert^2$.
	Besides, the notation $\pm~ \mathbf{x}$ is used to indicate that $(+\mathbf{x} - \mathbf{x})$ for brevity. 
	In $(c)$, we use the fact that $1 \leq \mathrm{L}_u^{k,e} \leq \mathrm{L}$ in the last term.
	Furthermore, we get $(d)$ using the fact that $\mathrm{1}_{u,\mathrm{sc}}^{k,e}$ and $g_u(\mathbf{w}_{u}^{k,e,l})$ are independent.
	$(e)$ is true since $\mathbb{E}[T_1] = 0$ from the unbiased mini-batch gradient assumption, i.e., $\mathbb{E} [g_u(\mathbf{w}_u^{k,e,l})] = \nabla f_u (\mathbf{w}_u^{k,e,l})$.
	Moreover, in $(f)$, we follow similar steps as in $(d)$ and $(e)$.
	Finally, $(g)$ is obtained from the bounded variance assumption of the gradients.

	Now, plugging $\mathrm{T}_1$ and $\mathrm{T}_2$ into (\ref{convProofMainEq_1}), we get
	\begin{align}
		\label{convProofMainEq_11}
		&\mathbb{E} \left[ f(\mathbf{w}^{k+1}) \right] 
		\leq \mathbb{E} [f(\mathbf{w}^k)] + \frac{\eta}{2} \sum_{e=0}^{E-1} \sum_{b=0}^{B-1} \alpha_b\sum_{u \in \bar{\mathcal{U}}_b^{k,e}} \alpha_u \sum_{l=0}^{\mathrm{L}_{u}^{k,e}-1} \left\Vert \nabla f(\mathbf{w}^k) - \nabla f_u \big(\mathbf{w}_u^{k, e,l}\big) \right\Vert^2 - \frac{\eta}{2} \sum_{e=0}^{E-1} \sum_{b=0}^{B-1} \alpha_b\sum_{u \in \bar{\mathcal{U}}_b^{k,e}} \alpha_u \mathrm{L}_{u}^{k,e} \left\Vert \nabla f(\mathbf{w}^k)\right\Vert^2 - \nonumber\\
		&\frac{\eta}{2} \sum_{e=0}^{E-1} \sum_{b=0}^{B-1} \alpha_b\sum_{u \in \bar{\mathcal{U}}_b^{k,e}} \alpha_u \sum_{l=0}^{\mathrm{L}_u^{k,e} - 1} \left
		\Vert \nabla f_u \big(\mathbf{w}_u^{k, e,l}\big) \right\Vert^2 + \beta \eta^2 \sigma^2 \sum_{e=0}^{E-1} \sum_{b=0}^{B-1} \left(\alpha_b\right)^2 \sum_{u \in \bar{\mathcal{U}}_b^{k,e}}  \left(\alpha_u \right)^2 \mathrm{L}_u^{k,e} + \nonumber\\
		&\beta E \eta^2 \sum_{e=0}^{E-1} \sum_{b=0}^{B-1} \alpha_b\sum_{u \in \bar{\mathcal{U}}_b^{k,e}} \alpha_u \mathrm{L}_u^{k,e} \sum_{l=0}^{\mathrm{L}_u^{k,e} - 1} \bigg[\frac{1} {\mathrm{p}_{u,\mathrm{sc}}^{k,e}} - 1 \bigg] \mathbb{E} \left[\left\Vert g_u (\mathbf{w}_u^{k, e, l}) \right\Vert^2 \right] + \frac{\beta E \mathrm{L} \eta^2}{2} \sum_{e=0}^{E-1} \sum_{b=0}^{B-1} \alpha_b \sum_{u \in \bar{\mathcal{U}}_b^{k,e}} \alpha_u \sum_{l=0}^{\mathrm{L}_u^{k,e} - 1} \left\Vert  \nabla f_u (\mathbf{w}_u^{k, e, l}) \right\Vert^2 \nonumber\\
		&= \mathbb{E} [f(\mathbf{w}^k)] + \beta \eta^2 \sigma^2 \sum_{e=0}^{E-1} \sum_{b=0}^{B-1} \left(\alpha_b\right)^2 \sum_{u \in \bar{\mathcal{U}}_b^{k,e}}  \left(\alpha_u \right)^2 \mathrm{L}_u^{k,e} + \beta E \eta^2 \sum_{e=0}^{E-1} \sum_{b=0}^{B-1} \alpha_b\sum_{u \in \bar{\mathcal{U}}_b^{k,e}} \alpha_u \mathrm{L}_u^{k,e} \sum_{l=0}^{\mathrm{L}_u^{k,e} - 1} \bigg[\frac{1} {\mathrm{p}_{u,\mathrm{sc}}^{k,e}} - 1 \bigg] \mathbb{E} \left[\left\Vert g_u (\mathbf{w}_u^{k, e, l}) \right\Vert^2 \right] + \nonumber\\
		&\mquad \frac{\eta}{2} \sum_{e=0}^{E-1} \sum_{b=0}^{B-1} \alpha_b\sum_{u \in \bar{\mathcal{U}}_b^{k,e}} \alpha_u \sum_{l=0}^{\mathrm{L}_{u}^{k,e}-1} \left\Vert \nabla f(\mathbf{w}^k) - \nabla f_u \big(\mathbf{w}_u^{k, e,l}\big) \right\Vert^2 - \frac{\eta}{2} \sum_{e=0}^{E-1} \sum_{b=0}^{B-1} \alpha_b\sum_{u \in \bar{\mathcal{U}}_b^{k,e}} \alpha_u \mathrm{L}_{u}^{k,e} \left\Vert \nabla f(\mathbf{w}^k)\right\Vert^2 - \nonumber\\
		&\mquad \frac{\eta}{2} \left(1 - \beta \eta E \mathrm{L} \right) \sum_{e=0}^{E-1} \sum_{b=0}^{B-1} \alpha_b\sum_{u \in \bar{\mathcal{U}}_b^{k,e}} \alpha_u \sum_{l=0}^{\mathrm{L}_u^{k,e} - 1} \left
		\Vert \nabla f_u \big(\mathbf{w}_u^{k, e,l}\big) \right\Vert^2,
	\end{align}
	where the last term is non-negative when $\eta \leq \frac{1}{\beta E \mathrm{L}}$, as $\Vert \cdot \Vert \geq 0$.
	
	To that end, as we are after an upper-bound, we drop the last term using $\eta \leq \frac{1}{\beta E \mathrm{L}}$.
	Then we divide both sides by $\frac{\eta}{2}$, rearrange the terms and take expectations on both sides, which yields 
	\begin{align}
		\label{convProofMainEq_12}
		&\sum_{e=0}^{E-1} \sum_{b=0}^{B-1} \alpha_b\sum_{u \in \bar{\mathcal{U}}_b^{k,e}} \alpha_u \mathrm{L}_{u}^{k,e} \mathbb{E} \bigg[ \left\Vert \nabla f(\mathbf{w}^k)\right\Vert^2 \bigg] 
		\leq \frac{2\left( \mathbb{E} [f(\mathbf{w}^k)] - \mathbb{E} \left[ f(\mathbf{w}^{k+1}) \right]\right)}{\eta} + 2 \beta \eta \sigma^2 \sum_{e=0}^{E-1} \sum_{b=0}^{B-1} \left(\alpha_b\right)^2 \sum_{u \in \bar{\mathcal{U}}_b^{k,e}}  \left(\alpha_u \right)^2 \mathrm{L}_u^{k,e} + \nonumber\\
		& 2 \beta \eta E \sum_{e=0}^{E-1} \sum_{b=0}^{B-1} \alpha_b\sum_{u \in \bar{\mathcal{U}}_b^{k,e}} \alpha_u \mathrm{L}_u^{k,e} \sum_{l=0}^{\mathrm{L}_u^{k,e} - 1} \bigg[\frac{1} {\mathrm{p}_{u,\mathrm{sc}}^{k,e}} - 1 \bigg] \mathbb{E} \left[\left\Vert g_u (\mathbf{w}_u^{k, e, l}) \right\Vert^2 \right] + \underbrace{\sum_{e=0}^{E-1} \sum_{b=0}^{B-1} \alpha_b\sum_{u \in \bar{\mathcal{U}}_b^{k,e}} \alpha_u \sum_{l=0}^{\mathrm{L}_{u}^{k,e}-1} \mathbb{E} \bigg[\left\Vert \nabla f(\mathbf{w}^k) - \nabla f_u \big(\mathbf{w}_u^{k, e,l}\big) \right\Vert^2 \bigg]}_{\mathrm{T}_3}.
	\end{align}
	We now bound the $\mathrm{T}_3$ term as 
	\begin{align}
		\mathrm{T}_3 &
		\overset{(a)}{=} \sum_{e=0}^{E-1} \sum_{b=0}^{B-1} \alpha_b\sum_{u \in \bar{\mathcal{U}}_b^{k,e}} \alpha_u \sum_{l=0}^{\mathrm{L}_{u}^{k,e}-1} \mathbb{E} \left[\left\Vert \sum_{b'=0}^{B-1} \alpha_{b'} \sum_{u' \in \bar{\mathcal{U}}_{b'}^{k,e}} \alpha_{u'}  \nabla f_{u'}(\mathbf{w}^k) - \nabla f_u \big(\mathbf{w}_u^{k, e,l}\big) \right] \right\Vert^2  \nonumber\\
		&\overset{(b)}{=} \sum_{e=0}^{E-1} \sum_{b=0}^{B-1} \alpha_b\sum_{u \in \bar{\mathcal{U}}_b^{k,e}} \alpha_u \sum_{l=0}^{\mathrm{L}_{u}^{k,e}-1} \mathbb{E} \left[\left\Vert \sum_{b'=0}^{B-1} \alpha_{b'} \sum_{u' \in \bar{\mathcal{U}}_{b'}^{k,e}} \alpha_{u'} \left[ \nabla f_{u'}(\mathbf{w}^k) - \nabla f_u \big(\mathbf{w}_u^{k, e,l}\big) \right] \right\Vert^2 \right] \nonumber\\
		&\overset{(c)}{\leq} \sum_{e=0}^{E-1} \sum_{b=0}^{B-1} \alpha_b\sum_{u \in \bar{\mathcal{U}}_b^{k,e}} \alpha_u \sum_{l=0}^{\mathrm{L}_{u}^{k,e}-1} \mathbb{E} \left[\left\Vert \nabla f_{u}(\mathbf{w}^k) \pm \nabla f_{u}(\mathbf{w}_{b}^{k,e}) - \nabla f_u \big(\mathbf{w}_u^{k, e,l}\big) \right\Vert^2 \right] \nonumber\\
		&\overset{(d)}{\leq} 2 \sum_{e=0}^{E-1} \sum_{b=0}^{B-1} \alpha_b\sum_{u \in \bar{\mathcal{U}}_b^{k,e}} \alpha_u \sum_{l=0}^{\mathrm{L}_{u}^{k,e}-1} \mathbb{E} \left[\left\Vert \nabla f_{u}(\mathbf{w}^k) - \nabla f_{u}(\mathbf{w}_{b}^{k,e}) \right\Vert^2 \right] + 2 \sum_{e=0}^{E-1} \sum_{b=0}^{B-1} \alpha_b\sum_{u \in \bar{\mathcal{U}}_b^{k,e}} \alpha_u \sum_{l=0}^{\mathrm{L}_{u}^{k,e}-1} \mathbb{E} \left[ \left\Vert \nabla f_{u}(\mathbf{w}_{b}^{k,e}) - \nabla f_{u}(\mathbf{w}_{u}^{k,e,l}) \right\Vert^2 \right] \nonumber\\
		&\overset{(e)}{\leq} 2\beta^2 \sum_{e=0}^{E-1} \sum_{b=0}^{B-1} \alpha_b\sum_{u \in \bar{\mathcal{U}}_b^{k,e}} \alpha_u \sum_{l=0}^{\mathrm{L}_{u}^{k,e}-1} \mathbb{E} \left[ \left\Vert \mathbf{w}_{b}^{k,e} - \mathbf{w}_{u}^{k,e,l} \right\Vert^2 \right] + 2 \beta^2 \sum_{e=0}^{E-1} \sum_{b=0}^{B-1} \alpha_b \sum_{u \in \bar{\mathcal{U}}_b^{k,e}} \alpha_u \mathrm{L}_{u}^{k,e} \mathbb{E} \left[ \left\Vert \mathbf{w}^k - \mathbf{w}_{b}^{k,e} \right\Vert^2 \right] \nonumber\\
		&\overset{(f)}{\leq} 2\beta^2 \sum_{e=0}^{E-1} \sum_{b=0}^{B-1} \alpha_b\sum_{u \in \bar{\mathcal{U}}_b^{k,e}} \alpha_u \sum_{l=0}^{\mathrm{L}_u^{k,e} - 1} \mathbb{E} \left[ \left\Vert \mathbf{w}_{b}^{k,e} - \mathbf{w}_{u}^{k,e,l} \right\Vert^2 \right] + 2 \mathrm{L} \beta^2 \sum_{e=0}^{E-1} \sum_{b=0}^{B-1} \alpha_b \mathbb{E} \left[ \left\Vert \mathbf{w}^k - \mathbf{w}_{b}^{k,e} \right\Vert^2 \right],
	\end{align}
	where, in step $(a)$, we use the fact that $\nabla f(\mathbf{w}^k) = \sum_{b'=0}^{B-1} \alpha_{b'} \sum_{u' \in \bar{\mathcal{U}}_{b'}^{k,e}} \alpha_{u'}  \nabla f_{u'}(\mathbf{w}^k)$ by definition of the global loss function in (\ref{centralLossRawHFL_apndx}).
	Besides, we write $(b)$ because $\sum_{u \in \bar{\mathcal{U}}_b^{k,e}} \alpha_u = 1$ and $\sum_{b=0}^{B-1} \alpha_b=1$.
	Furthermore, $(c)$ stems from Jensen inequality and $(d)$ appears from the fact $\Vert \sum_{j=1}^J \mathbf{x}_j \Vert^2 \leq J \sum_{j=1}^J \Vert \mathbf{x}_j \Vert^2$.
	Moreover, we get to $(e)$ by using $\beta$-smoothness assumption.
	Finally, we write $(f)$ since $1 \leq \mathrm{L}_u^{k,e} \leq \mathrm{L}$ and $\Vert \cdot \Vert \geq 0$.

	Now, plugging $\mathrm{T}_3$ into (\ref{convProofMainEq_12}) and averaging over $K$ global rounds, we get
	\begin{align} 
		\label{convProofMainEq_13}
		&\frac{1}{K} \sum_{k=0}^{K-1} \mathbb{E} \left[\left\Vert \nabla f(\mathbf{w}^k)\right\Vert^2 \right]
		\leq \frac{2}{\eta K} \sum_{k=0}^{K-1} \left[ \frac{\mathbb{E} [f(\mathbf{w}^k)] - \mathbb{E} \left[ f(\mathbf{w}^{k+1}) \right] }{\Omega^k} \right] + \frac{2 \beta \eta \sigma^2}{K} \sum_{k=0}^{K-1} \left[\frac{\sum_{e=0}^{E-1} \sum_{b=0}^{B-1} \left(\alpha_b\right)^2 \sum_{u \in \bar{\mathcal{U}}_b^{k,e}} \left(\alpha_u \right)^2 \mathrm{L}_u^{k,e} }{\Omega^k} \right] + \nonumber\\
		&\squad \frac{2 \beta \eta E}{K} \sum_{k=0}^K \frac{1}{\Omega^k} \sum_{e=0}^{E-1} \sum_{b=0}^{B-1} \alpha_b\sum_{u \in \bar{\mathcal{U}}_b^{k,e}} \alpha_u \mathrm{L}_u^{k,e} \sum_{l=0}^{\mathrm{L}_u^{k,e} - 1} \bigg[\frac{1} {\mathrm{p}_{u,\mathrm{sc}}^{k,e}} - 1 \bigg] \mathbb{E} \left[\left\Vert g_u (\mathbf{w}_u^{k, e, l}) \right\Vert^2 \right] + \nonumber\\
		&\squad \frac{2 \mathrm{L} \beta^2}{K} \sum_{k=0}^{K-1} \frac{1}{\Omega^k} \sum_{e=0}^{E-1} \sum_{b=0}^{B-1} \alpha_b \mathbb{E} \left[ \left\Vert \mathbf{w}^k - \mathbf{w}_{b}^{k,e} \right\Vert^2 \right] + \frac{2 \beta^2}{K} \sum_{k=0}^{K-1} \frac{1}{\Omega^k} \sum_{e=0}^{E-1} \sum_{b=0}^{B-1} \alpha_b \rs \rs \sum_{u \in \bar{\mathcal{U}}_b^{k,e}} \rs \rs \alpha_u \rs \sum_{l=0}^{\mathrm{L}_u^{k,e} - 1} \mathbb{E} \left[ \left\Vert \mathbf{w}_{b}^{k,e} - \mathbf{w}_{u}^{k,e,l} \right\Vert^2 \right].
	\end{align}
	where $\Omega^k \coloneqq \sum_{e=0}^{E-1} \sum_{b=0}^{B-1} \alpha_b \sum_{u \in \bar{\mathcal{U}}_b^{k,e}} \alpha_u \mathrm{L}_{u}^{k,e}$.

	\begin{Lemma}
		\label{lem_div_BS_UE_models}
		When learning rate $\eta < \mathrm{min} \left\{ \frac{1}{\beta E \mathrm{L}}, \frac{1}{3\sqrt{2} \beta \mathrm{L}}\right\}$, the mean square error between the local model and the ES model is upper-bounded as
		\begin{align}
			&\frac{1}{K} \sum_{k=0}^{K-1} \frac{1}{\Omega^k} \sum_{e=0}^{E-1} \sum_{b=0}^{B-1} \alpha_b\sum_{u \in \bar{\mathcal{U}}_b^{k,e}} \alpha_u \sum_{l=0}^{\mathrm{L}_u^{k,e} - 1} \mathbb{E} \left[ \left\Vert \mathbf{w}_{b}^{k,e} - \mathbf{w}_{u}^{k,e,l} \right\Vert^2 \right] 
			\leq \frac{3 E \mathrm{L}^2 \eta^2 \sigma^2}{K} \sum_{k=0}^{K-1} \frac{1}{\Omega^k} + \frac{9 E \epsilon_0^2 \eta^2 \mathrm{L}^3}{K} \sum_{k=0}^{K-1} \frac{1}{\Omega^k} + \nonumber\\
			&\bquad \frac{3 \mathrm{L}^2 \eta^2}{K} \sum_{k=0}^{K-1} \frac{1}{\Omega^k} \sum_{e=0}^{E-1} \sum_{b=0}^{B-1} \rs \alpha_b \rs \rs \sum_{u \in \bar{\mathcal{U}}_b^{k,e}} \rs\rs  \alpha_u \rs \bigg[\frac{1}{\mathrm{p}_{u,\mathrm{sc}}^{k,e}} - 1 \bigg]\sum_{l=0}^{\mathrm{L}_u^{k,e} - 1} \mathbb{E} \left[ \left\Vert g_u (\mathbf{w}_u^{k, e, l})  \right\Vert^2 \right].
		\end{align}
	\end{Lemma}
	\begin{proof}
		The proof is left to Appendix \ref{lem_div_BS_UE_models_Proof}.
	\end{proof}

	\begin{Lemma}
		\label{lem_div_central_BS_models}
		When learning rate $\eta < \mathrm{min}\left\{\frac{1}{2\sqrt{5} \beta \mathrm{L}}, \frac{1}{\beta E \mathrm{L}} \right\}$, the mean square error between the edge model at the ES and the global model at the central server is upper-bounded as
		\begin{align}
			&\frac{1}{K} \sum_{k=0}^{K-1} \frac{1}{\Omega^k} \sum_{e=0}^{E-1} \sum_{b=0}^{B-1} \alpha_b \mathbb{E} \left[ \left\Vert \mathbf{w}^k - \mathbf{w}_{b}^{k,e} \right\Vert^2 \right] \nonumber\\
			&\leq \frac{4 E \eta^2 \sigma^2} {K} \sum_{k=0}^{K-1} \frac{1}{\Omega^k} \sum_{e=0}^{E-1} \sum_{b=0}^{B-1} \alpha_b \rs\rs \sum_{u \in \bar{\mathcal{U}}_b^{k,e}} \rs \left(\alpha_u\right)^2 \mathrm{L}_u^{k,e} + \frac{60 \beta^2 \sigma^2 E^3 \mathrm{L}^3 \eta^4}{K} \sum_{k=0}^{K-1} \frac{1}{\Omega^k} + \frac{180 \beta^2 \epsilon_0^2 E^3 \mathrm{L}^4 \eta^4}{K} \sum_{k=0}^{K-1} \frac{1}{\Omega^k} + \frac{10 \epsilon_1^2 \eta^2 E^3}{K} \sum_{k=0}^{K-1} \frac{1}{\Omega^k} + \nonumber\\
			&\squad \frac{4 E \eta^2}{K} \sum_{k=0}^{K-1} \frac{1}{\Omega^k} \sum_{e=0}^{E-1} \sum_{b=0}^{B-1} \rs \alpha_b \rs\rs \sum_{u \in \bar{\mathcal{U}}_b^{k,e}} \alpha_u \left(\alpha_u + 15 E \beta^2 \eta^2 \mathrm{L}^3 \right) \rs \bigg[\frac{1}{\mathrm{p}_{u,\mathrm{sc}}^{k,e}} - 1 \bigg] \sum_{l=0}^{\mathrm{L}_u^{k,e} - 1} \mathbb{E} \left[\left\Vert g_u (\mathbf{w}_u^{k, e, l}) \right\Vert^2\right].
		\end{align}
	\end{Lemma}
	\begin{proof}
		The proof is left to Appendix \ref{lem_div_central_BS_models_Proof}.
	\end{proof}
	
	\noindent
	\textbf{Final Results:}
	Using Lemma \ref{lem_div_BS_UE_models} and Lemma \ref{lem_div_central_BS_models}, we write
	\begin{align}
		&\frac{1}{K} \sum_{k=0}^{K-1} \mathbb{E} \left[\left\Vert \nabla f(\mathbf{w}^k)\right\Vert^2 \right]
		\leq \frac{2}{\eta K} \sum_{k=0}^{K-1} \left[ \frac{\mathbb{E} [f(\mathbf{w}^k)] - \mathbb{E} \left[ f(\mathbf{w}^{k+1}) \right] }{\Omega^k} \right] + \frac{2 \beta \eta \sigma^2}{K} \sum_{k=0}^{K-1} \left[\frac{\sum_{e=0}^{E-1} \sum_{b=0}^{B-1} \left(\alpha_b\right)^2 \sum_{u \in \bar{\mathcal{U}}_b^{k,e}} \left(\alpha_u \right)^2 \mathrm{L}_u^{k,e} }{\Omega^k} \right] + \nonumber\\
		&\frac{2 \beta \eta E}{K} \sum_{k=0}^K \frac{1}{\Omega^k} \sum_{e=0}^{E-1} \sum_{b=0}^{B-1} \alpha_b\sum_{u \in \bar{\mathcal{U}}_b^{k,e}} \alpha_u \mathrm{L}_u^{k,e} \sum_{l=0}^{\mathrm{L}_u^{k,e} - 1} \bigg[\frac{1} {\mathrm{p}_{u,\mathrm{sc}}^{k,e}} - 1 \bigg] \mathbb{E} \left[\left\Vert g_u (\mathbf{w}_u^{k, e, l}) \right\Vert^2 \right] + \nonumber\\
		& \frac{8 E \mathrm{L} \beta^2 \eta^2 \sigma^2} {K} \sum_{k=0}^{K-1} \frac{1}{\Omega^k} \sum_{e=0}^{E-1} \sum_{b=0}^{B-1} \alpha_b \rs\rs \sum_{u \in \bar{\mathcal{U}}_b^{k,e}} \rs \left(\alpha_u\right)^2 \mathrm{L}_u^{k,e} + \frac{120 \sigma^2 E^3 \beta^4 \eta^4 \mathrm{L}^4}{K} \sum_{k=0}^{K-1} \frac{1}{\Omega^k} +  \frac{360 \epsilon_0^2 E^3 \beta^4 \eta^4 \mathrm{L}^5}{K} \sum_{k=0}^{K-1} \frac{1}{\Omega^k} + \frac{20 \mathrm{L} \beta^2 \epsilon_1^2 \eta^2 E^3}{K} \sum_{k=0}^{K-1} \frac{1}{\Omega^k} +  \nonumber\\
		&\frac{8 E \mathrm{L} \beta^2 \eta^2}{K} \sum_{k=0}^{K-1} \frac{1}{\Omega^k} \sum_{e=0}^{E-1} \sum_{b=0}^{B-1} \rs \alpha_b \rs\rs \sum_{u \in \bar{\mathcal{U}}_b^{k,e}} \alpha_u \left(\alpha_u + 15 E \beta^2 \eta^2 \mathrm{L}^3 \right) \rs \bigg[\frac{1}{\mathrm{p}_{u,\mathrm{sc}}^{k,e}} - 1 \bigg] \sum_{l=0}^{\mathrm{L}_u^{k,e} - 1} \mathbb{E} \left[\left\Vert g_u (\mathbf{w}_u^{k, e, l}) \right\Vert^2\right] + \frac{6 E \beta^2 \eta^2 \sigma^2 \mathrm{L}^2}{K} \sum_{k=0}^{K-1} \frac{1}{\Omega^k} + \nonumber\\
		& \frac{18 E \beta^2 \epsilon_0^2 \eta^2 \mathrm{L}^3}{K} \sum_{k=0}^{K-1} \frac{1}{\Omega^k} + \frac{6 \beta^2 \mathrm{L}^2 \eta^2}{K} \sum_{k=0}^{K-1} \frac{1}{\Omega^k} \sum_{e=0}^{E-1} \sum_{b=0}^{B-1} \rs \alpha_b \rs \rs \sum_{u \in \bar{\mathcal{U}}_b^{k,e}} \rs\rs  \alpha_u  \bigg[\frac{1}{\mathrm{p}_{u,\mathrm{sc}}^{k,e}} - 1 \bigg]\sum_{l=0}^{\mathrm{L}_u^{k,e} - 1} \mathbb{E} \left[ \left\Vert g_u (\mathbf{w}_u^{k, e, l}) \right\Vert^2 \right] \nonumber\\
		&\leq \frac{2}{\eta K} \sum_{k=0}^{K-1} \frac{1}{\Omega^k} \left(\mathbb{E} [f(\mathbf{w}^k)] - \mathbb{E} [f(\mathbf{w}^{k+1})] \right) + \frac{2 \beta \eta \mathrm{L} \sigma^2}{K} \sum_{k=0}^{K-1} \frac{1}{\Omega^k}\Bigg[ 60 \beta^3 \eta^3 E^3 \mathrm{L}^3 + 3 \beta \eta E \mathrm{L} +  \sum_{e=0}^{E-1} \sum_{b=0}^{B-1} \alpha_b  \left(\alpha_b + 4 E \mathrm{L} \beta \eta \right) \rs\rs \sum_{u \in \bar{\mathcal{U}}_b^{k,e}} \rs\rs \left(\alpha_u\right)^2 \Bigg] + \nonumber\\
		&\squad \frac{18 E \beta^2 \epsilon_0^2 \eta^2 \mathrm{L}^3}{K} \sum_{k=0}^{K-1} \frac{1}{\Omega^k} \big[1 + 20 \beta^2 \eta^2 E^2 \mathrm{L}^2 \big] + \frac{20 \mathrm{L} \beta^2 \epsilon_1^2 \eta^2 E^3}{K} \sum_{k=0}^{K-1} \frac{1}{\Omega^k} + \frac{2 \beta \eta \mathrm{L}}{K} \sum_{k=0}^{K-1} \frac{1}{\Omega^k} \sum_{e=0}^{E-1} \sum_{b=0}^{B-1} \alpha_b\sum_{u \in \bar{\mathcal{U}}_b^{k,e}} \alpha_u \big[ E + 3 \beta \eta \mathrm{L} + \nonumber\\
		&\squad  4 \beta \eta E \left(\alpha_u + 15 E \beta^2 \eta^2 \mathrm{L}^3 \right) \big] \sum_{l=0}^{\mathrm{L}_u^{k,e} - 1} \bigg[\frac{1} {\mathrm{p}_{u,\mathrm{sc}}^{k,e}} - 1 \bigg] \mathbb{E} \left[\left\Vert g_u (\mathbf{w}_u^{k, e, l}) \right\Vert^2 \right] \nonumber\\
		&=\frac{2}{\eta K} \sum_{k=0}^{K-1} \frac{1}{\Omega^k} \left(\mathbb{E} [f(\mathbf{w}^k)] - \mathbb{E} [f(\mathbf{w}^{k+1})] \right) + \frac{2 \beta \eta \mathrm{L} \sigma^2}{K} \sum_{k=0}^{K-1} \frac{\mathrm{N}_1^k}{\Omega^k} + \frac{18 E \beta^2 \epsilon_0^2 \eta^2 \mathrm{L}^3}{K} \sum_{k=0}^{K-1} \frac{\mathrm{N}_2}{\Omega^k} + \frac{20 \mathrm{L} \beta^2 \epsilon_1^2 \eta^2 E^3}{K} \sum_{k=0}^{K-1} \frac{1}{\Omega^k} + \nonumber\\
		&\squad \frac{2 \beta \eta \mathrm{L}}{K} \sum_{k=0}^{K-1} \frac{1}{\Omega^k} \sum_{e=0}^{E-1} \sum_{b=0}^{B-1} \alpha_b \rs \rs \sum_{u \in \bar{\mathcal{U}}_b^{k,e}} \rs\rs \alpha_u \mathrm{N}_u  \bigg[\frac{1} {\mathrm{p}_{u,\mathrm{sc}}^{k,e}} - 1 \bigg] \sum_{l=0}^{\mathrm{L}_u^{k,e} - 1} \mathbb{E} \left[\left\Vert g_u (\mathbf{w}_u^{k, e, l}) \right\Vert^2 \right],
	\end{align}
	where $\mathrm{N}_1^k \coloneqq  60 \beta^3 \eta^3 E^3 \mathrm{L}^3 + 3 \beta \eta E \mathrm{L} +  \sum_{e=0}^{E-1} \sum_{b=0}^{B-1} \alpha_b  \left(\alpha_b + 4 E \mathrm{L} \beta \eta \right) \sum_{u \in \bar{\mathcal{U}}_b^{k,e}} \left(\alpha_u\right)^2$, $\mathrm{N}_2 \coloneqq \big[1 + 20 \beta^2 \eta^2 E^2 \mathrm{L}^2 \big]$ and $\mathrm{N}_{u} \coloneqq  E + 3 \beta \eta \mathrm{L} + 4 \beta \eta E \left(\alpha_u + 15 E \beta^2 \eta^2 \mathrm{L}^3 \right)$.

	This concludes the proof of Theorem \ref{theorem1_apndx}.
	
\end{proof}

\section{Missing Proof of Lemma {\ref{lem_div_BS_UE_models}} and Lemma \ref{lem_div_central_BS_models}}
\label{missingProofOfLemmas}

\subsection{Proof of Lemma {\ref{lem_div_BS_UE_models}}}
\label{lem_div_BS_UE_models_Proof}
\begin{align}
	\label{proofLocal_BS_Div}
	&\frac{1}{K} \sum_{k=0}^{K-1} \frac{1}{\Omega^k} \sum_{e=0}^{E-1} \sum_{b=0}^{B-1} \alpha_b\sum_{u \in \bar{\mathcal{U}}_b^{k,e}} \alpha_u \sum_{l=0}^{\mathrm{L}_u^{k,e} - 1} \mathbb{E} \left[ \left\Vert \mathbf{w}_{b}^{k,e} - \mathbf{w}_{u}^{k,e,l} \right\Vert^2 \right] \nonumber\\
	&\overset{(a)}{=}\frac{\eta^2}{K} \sum_{k=0}^{K-1} \frac{1}{\Omega^k} \sum_{e=0}^{E-1} \sum_{b=0}^{B-1} \alpha_b\sum_{u \in \bar{\mathcal{U}}_b^{k,e}} \alpha_u \sum_{l=0}^{\mathrm{L}_u^{k,e} - 1} \mathbb{E} \left[ \left\Vert  \sum_{l'=0}^{l-1} g (\mathbf{w}_{u}^{k, e, l'}) - \sum_{u' \in \mathcal{U}_b} \alpha_{u'} \frac{\mathrm{1}_{u',\mathrm{sc}}^{k,e}}{\mathrm{p}_{u',\mathrm{sc}}^{k,e}} \sum_{l'=0}^{l-1} g (\mathbf{w}_{u'}^{k, e, l'}) \right\Vert^2 \right] \nonumber\\
	&=\frac{\eta^2}{K} \sum_{k=0}^{K-1} \frac{1}{\Omega^k} \sum_{e=0}^{E-1} \sum_{b=0}^{B-1} \alpha_b\sum_{u \in \bar{\mathcal{U}}_b^{k,e}} \alpha_u \sum_{l=0}^{\mathrm{L}_u^{k,e} - 1} \mathbb{E} \Bigg[ \Bigg\Vert \sum_{l'=0}^{l-1} g (\mathbf{w}_u^{k,e,l'}) \pm \sum_{l'=0}^{l-1} \nabla f_u (\mathbf{w}_u^{k,e,l'}) -  \sum_{u' \in \mathcal{U}_b} \alpha_{u'} \frac{\mathrm{1}_{u',\mathrm{sc}}^{k,e}}{\mathrm{p}_{u',\mathrm{sc}}^{k,e}} \sum_{l'=0}^{l-1} g (\mathbf{w}_{u'}^{k,e,l'}) \nonumber\\
	&\mquad \pm \sum_{u' \in \mathcal{U}_b} \alpha_{u'} \sum_{l'=0}^{l-1} \mathrm{g} (\mathbf{w}_{u'}^{k,e,l}) \pm \sum_{u' \in \mathcal{U}_b} \alpha_{u'} \sum_{l'=0}^{l-1} \nabla f_u (\mathbf{w}_{u'}^{k,e,l'}) \Bigg\Vert^2 \Bigg] \nonumber\\
	&\leq \frac{3\eta^2}{K} \sum_{k=0}^{K-1} \frac{1}{\Omega^k} \sum_{e=0}^{E-1} \sum_{b=0}^{B-1} \alpha_b\sum_{u \in \bar{\mathcal{U}}_b^{k,e}} \alpha_u \sum_{l=0}^{\mathrm{L}_u^{k,e} - 1} \mathbb{E} \Bigg[ \Bigg\Vert \sum_{l'=0}^{l-1} \left[ g (\mathbf{w}_u^{k,e,l'}) -  \nabla f_u (\mathbf{w}_u^{k,e,l'}) \right] - \nonumber\\
	&\mquad \sum_{u' \in \mathcal{U}_b} \alpha_{u'} \sum_{l'=0}^{l-1} \left[ \mathrm{g} (\mathbf{w}_{u'}^{k,e,l'}) - \nabla f_u (\mathbf{w}_{u'}^{k,e,l'}) \right] \Bigg\Vert^2 \Bigg] + \nonumber\\
	&\mquad \frac{3\eta^2}{K} \sum_{k=0}^{K-1} \frac{1}{\Omega^k} \sum_{e=0}^{E-1} \sum_{b=0}^{B-1} \alpha_b\sum_{u \in \bar{\mathcal{U}}_b^{k,e}} \alpha_u \sum_{l=0}^{\mathrm{L}_u^{k,e} - 1} \mathbb{E} \Bigg[ \Bigg\Vert \sum_{l'=0}^{l-1} \nabla f_u (\mathbf{w}_u^{k,e,l'}) - \sum_{u' \in \mathcal{U}_b} \alpha_{u'} \sum_{l'=0}^{l-1} \nabla f_{u'} (\mathbf{w}_{u'}^{k,e,l'}) \Bigg\Vert^2 \Bigg] + \nonumber\\
	&\mquad \frac{3\eta^2}{K} \sum_{k=0}^{K-1} \frac{1}{\Omega^k} \sum_{e=0}^{E-1} \sum_{b=0}^{B-1} \alpha_b\sum_{u \in \bar{\mathcal{U}}_b^{k,e}} \alpha_u \sum_{l=0}^{\mathrm{L}_u^{k,e} - 1} \mathbb{E} \Bigg[ \Bigg\Vert \sum_{u' \in \mathcal{U}_b} \alpha_{u'}  \sum_{l'=0}^{l-1} g (\mathbf{w}_{u'}^{k,e,l'}) \left[ 1 - \frac{\mathrm{1}_{u',\mathrm{sc}}^{k,e}}{\mathrm{p}_{u',\mathrm{sc}}^{k,e}} \right] \Bigg\Vert^2 \Bigg], 
\end{align}
where we get $(a)$ due to the fact that the client received the same model from the \ac{es} during the synchronization time.
Besides, the change in the \ac{es}' model is captured for all $0 \leq l'\leq l-1$ \ac{sgd} rounds, and the aggregation rule follows (\ref{edgeUpdateRule_apndx}) to indicate whether these gradients are received at the \ac{es}.

The first term of (\ref{proofLocal_BS_Div}) is upper bounded as 
\begin{align}
	\label{proofLocal_BS_Div_Term1}
	&\frac{3\eta^2}{K} \sum_{k=0}^{K-1} \frac{1}{\Omega^k} \sum_{e=0}^{E-1} \sum_{b=0}^{B-1} \alpha_b\sum_{u \in \bar{\mathcal{U}}_b^{k,e}} \alpha_u \sum_{l=0}^{\mathrm{L}_u^{k,e} - 1} \mathbb{E} \Bigg[ \Bigg\Vert \sum_{l'=0}^{l-1} \bigg[\left( g (\mathbf{w}_u^{k,e,l'}) - \nabla f_u (\mathbf{w}_u^{k,e,l'}) \right)  -  \sum_{u' \in \bar{\mathcal{U}}_b^{k,e}} \alpha_{u'} \left( g (\mathbf{w}_{u'}^{k,e,l'}) - \nabla f_u (\mathbf{w}_{u'}^{k,e,l'}) \right) \bigg] \Bigg\Vert^2 \Bigg] \nonumber\\
	&= \frac{3\eta^2}{K} \sum_{k=0}^{K-1} \frac{1}{\Omega^k} \sum_{e=0}^{E-1} \sum_{b=0}^{B-1} \alpha_b\sum_{u \in \bar{\mathcal{U}}_b^{k,e}} \alpha_u \sum_{l=0}^{\mathrm{L}_u^{k,e} - 1} \mathbb{E} \Bigg[ \sum_{l'=0}^{l-1} \Bigg\Vert \left( g (\mathbf{w}_u^{k,e,l'}) - \nabla f_u (\mathbf{w}_u^{k,e,l'}) \right) -  \sum_{u' \in \bar{\mathcal{U}}_b^{k,e}} \alpha_{u'} \left( g (\mathbf{w}_{u'}^{k,e,l'}) - \nabla f_u (\mathbf{w}_{u'}^{k,e,l'}) \right) \Bigg\Vert^2 + \nonumber\\
	&\sum_{l'=0}^{l-1} \bigg\{\left( g (\mathbf{w}_u^{k,e,l'}) - \nabla f_u (\mathbf{w}_u^{k,e,l'}) \right) -  \sum_{u' \in \bar{\mathcal{U}}_b^{k,e}} \alpha_{u'} \left( g (\mathbf{w}_{u'}^{k,e,l'}) - \nabla f_u (\mathbf{w}_{u'}^{k,e,l'}) \right) \bigg\} \times \sum_{l''=0, l' \neq l''}^{l-1} \bigg\{\left( g (\mathbf{w}_u^{k,e,l''}) - \nabla f_u (\mathbf{w}_u^{k,e,l''}) \right) - \nonumber\\
	&\bquad \bquad \bquad \qquad \sum_{u' \in \bar{\mathcal{U}}_b^{k,e}} \alpha_{u'} \left( g (\mathbf{w}_{u'}^{k,e,l''}) - \nabla f_u (\mathbf{w}_{u'}^{k,e,l''}) \right) \bigg\} \Bigg] \nonumber\\
	&\overset{(a)}{=} \frac{3\eta^2}{K} \sum_{k=0}^{K-1} \frac{1}{\Omega^k} \sum_{e=0}^{E-1} \sum_{b=0}^{B-1} \alpha_b\sum_{u \in \bar{\mathcal{U}}_b^{k,e}} \alpha_u \sum_{l=0}^{\mathrm{L}_u^{k,e} - 1} \sum_{l'=0}^{l-1} \mathbb{E} \Bigg[ \Bigg\Vert \left( g (\mathbf{w}_u^{k,e,l'}) - \nabla f_u (\mathbf{w}_u^{k,e,l'}) \right)  -  \sum_{u' \in \bar{\mathcal{U}}_b^{k,e}} \alpha_{u'} \left( g (\mathbf{w}_{u'}^{k,e,l'}) - \nabla f_u (\mathbf{w}_{u'}^{k,e,l'}) \right) \Bigg\Vert^2 \Bigg] \nonumber\\
	&\overset{(b)}{\leq} \frac{3 \mathrm{L} \eta^2}{K} \sum_{k=0}^{K-1} \frac{1}{\Omega^k} \sum_{e=0}^{E-1} \sum_{b=0}^{B-1} \alpha_b\sum_{u \in \bar{\mathcal{U}}_b^{k,e}} \alpha_u  \sum_{l=0}^{\mathrm{L}_u^{k,e} - 1} \mathbb{E} \bigg[ \bigg\Vert \left( g_u (\mathbf{w}_u^{k, e, l}) - \nabla f_u (\mathbf{w}_u^{k,e,l}) \right)  - \sum_{u' \in \bar{\mathcal{U}}_b^{k,e}} \alpha_{u'} \left( g (\mathbf{w}_{u'}^{k,e,l}) - \nabla f_u (\mathbf{w}_{u'}^{k,e,l}) \right) \bigg\Vert^2 \bigg] \nonumber\\
	&\overset{(c)}{=} \frac{3 \mathrm{L} \eta^2}{K} \sum_{k=0}^{K-1} \frac{1}{\Omega^k} \sum_{e=0}^{E-1} \sum_{b=0}^{B-1} \alpha_b \sum_{l=0}^{\mathrm{L}_u^{k,e} - 1} \sum_{u \in \bar{\mathcal{U}}_b^{k,e}} \alpha_u \mathbb{E} \left[ \left\Vert g_u (\mathbf{w}_u^{k, e, l}) - \nabla f_u (\mathbf{w}_u^{k,e,l}) \right\Vert^2 \right]  - \nonumber\\
	&\bquad \mquad \frac{3 \mathrm{L} \eta^2}{K} \sum_{k=0}^{K-1} \frac{1}{\Omega^k} \sum_{e=0}^{E-1} \sum_{b=0}^{B-1} \alpha_b \sum_{l=0}^{\mathrm{L}_u^{k,e} - 1} \mathbb{E} \left[ \left\Vert \sum_{u' \in \bar{\mathcal{U}}_b^{k,e}} \alpha_{u'} \left( g (\mathbf{w}_{u'}^{k,e,l}) - \nabla f_u (\mathbf{w}_{u'}^{k,e,l}) \right) \right\Vert^2 \right] \nonumber\\
	&\overset{(d)}{\leq} \frac{3 \mathrm{L} \eta^2}{K} \sum_{k=0}^{K-1} \frac{1}{\Omega^k} \sum_{e=0}^{E-1} \sum_{b=0}^{B-1} \alpha_b \sum_{l=0}^{\mathrm{L}_u^{k,e} - 1} \sum_{u \in \bar{\mathcal{U}}_b^{k,e}} \alpha_u \sigma^2 - \frac{3 \mathrm{L} \eta^2}{K} \sum_{k=0}^{K-1} \frac{1}{\Omega^k} \sum_{e=0}^{E-1} \sum_{b=0}^{B-1} \alpha_b \sum_{l=0}^{\mathrm{L}_u^{k,e} - 1} \sum_{u' \in \bar{\mathcal{U}}_b^{k,e}} \left(\alpha_{u'} \right)^2 \mathbb{E} \left[ \left\Vert g (\mathbf{w}_{u'}^{k,e,l}) - \nabla f_u (\mathbf{w}_{u'}^{k,e,l}) \right\Vert^2 \right] \nonumber\\
	&\overset{(e)}{\leq} \frac{3 E \mathrm{L}^2 \eta^2 \sigma^2}{K} \sum_{k=0}^{K-1} \frac{1}{\Omega^k}  - \frac{3 \mathrm{L}^2 \eta^2 \sigma^2}{K} \sum_{k=0}^{K-1} \frac{1}{\Omega^k} \sum_{e=0}^{E-1} \sum_{b=0}^{B-1} \alpha_b
	\nonumber\\
	&\leq \frac{3 E \mathrm{L}^2 \eta^2 \sigma^2}{K} \sum_{k=0}^{K-1} \frac{1}{\Omega^k},
\end{align}
where $(a)$ stems from time independence of the mini-batch gradients and $\mathbb{E} [g_u(\mathbf{w}_u^{k,e,l})] = \nabla f_u (\mathbf{w}_u^{k,e,l})$, which makes the expectation of the cross product terms $0$. 
Besides, we used the fact that $l \leq \mathrm{L}_u^{k,e}$ and $1 \leq \mathrm{L}_u^{k,e} \leq \mathrm{L}$ in $(b)$.
In $(c)$, $\sum_{j=1}^J \alpha_j \Vert \mathbf{x}_j - \bar{\mathbf{x}} \Vert^2 = \sum_{j=1}^J \alpha_j \Vert \mathbf{x}_j \Vert^2 - \Vert \bar{\mathbf{x}} \Vert^2$, where $\bar{\mathbf{x}} = \sum_{j=1}^J \alpha_j \mathbf{x}_j$.
In $(d)$ and $(e)$, we use the bounded variance of the SGD assumptions and independence of SGDs in different epochs. 

For the second term of (\ref{proofLocal_BS_Div}), we write 
\begin{align}
	\label{proofLocal_BS_Div_Term2}
	&\frac{3\eta^2}{K} \sum_{k=0}^{K-1} \frac{1}{\Omega^k} \sum_{e=0}^{E-1} \sum_{b=0}^{B-1} \alpha_b\sum_{u \in \bar{\mathcal{U}}_b^{k,e}} \alpha_u \sum_{l=0}^{\mathrm{L}_u^{k,e} - 1} \mathbb{E} \Bigg[ \Bigg\Vert \sum_{l'=0}^{l-1} \left(\nabla f_u (\mathbf{w}_u^{k,e,l'}) - \sum_{u' \in \mathcal{U}_b} \alpha_{u'} \nabla f_{u'} (\mathbf{w}_{u'}^{k,e,l'}) \right) \Bigg\Vert^2 \Bigg] \nonumber\\
	&\leq \frac{3 \mathrm{L} \eta^2}{K} \sum_{k=0}^{K-1} \frac{1}{\Omega^k} \sum_{e=0}^{E-1} \sum_{b=0}^{B-1} \alpha_b\sum_{u \in \bar{\mathcal{U}}_b^{k,e}} \alpha_u \sum_{l=0}^{\mathrm{L}_u^{k,e} - 1} \sum_{l'=0}^{l-1} \mathbb{E} \left[ \left\Vert \nabla f_u (\mathbf{w}_u^{k,e,l'}) - \sum\nolimits_{u' \in \mathcal{U}_b} \alpha_{u'} \nabla f_{u'} (\mathbf{w}_{u'}^{k,e,l'}) \right\Vert^2 \right] \nonumber \\
	&\leq \frac{3 \mathrm{L}^2 \eta^2}{K} \sum_{k=0}^{K-1} \frac{1}{\Omega^k} \sum_{e=0}^{E-1} \sum_{b=0}^{B-1} \alpha_b\sum_{u \in \bar{\mathcal{U}}_b^{k,e}} \alpha_u \sum_{l=0}^{\mathrm{L}_u^{k,e} - 1}  \mathbb{E} \left[ \left\Vert \nabla f_u (\mathbf{w}_u^{k,e,l}) \pm \nabla f_u (\mathbf{w}_b^{k,e}) \pm \sum_{u' \in \mathcal{U}_b} \alpha_{u'} \nabla f_{u'} (\mathbf{w}_{b}^{k,e}) - \sum_{u' \in \mathcal{U}_b} \alpha_{u'} \nabla f_{u'} (\mathbf{w}_{u'}^{k,e,l}) \right\Vert^2 \right] \nonumber\\
	&\leq \frac{9 \mathrm{L}^2 \eta^2}{K} \sum_{k=0}^{K-1} \frac{1}{\Omega^k} \sum_{e=0}^{E-1} \sum_{b=0}^{B-1} \alpha_b\sum_{u \in \bar{\mathcal{U}}_b^{k,e}} \alpha_u \sum_{l=0}^{\mathrm{L}_u^{k,e} - 1}  \mathbb{E} \left[ \left\Vert \nabla f_u (\mathbf{w}_u^{k,e,l}) - \nabla f_u (\mathbf{w}_b^{k,e}) \right\Vert^2 \right] + \nonumber\\
	&\mquad \frac{9 \mathrm{L}^2 \eta^2}{K} \sum_{k=0}^{K-1} \frac{1}{\Omega^k} \sum_{e=0}^{E-1} \sum_{b=0}^{B-1} \alpha_b\sum_{u \in \bar{\mathcal{U}}_b^{k,e}} \alpha_u \sum_{l=0}^{\mathrm{L}_u^{k,e} - 1}  \mathbb{E} \Bigg[ \Bigg\Vert \nabla f_u (\mathbf{w}_b^{k,e}) - \sum_{u' \in \mathcal{U}_b} \alpha_{u'} \nabla f_{u'} (\mathbf{w}_{b}^{k,e}) \Bigg\Vert^2\Bigg] + \nonumber\\
	&\mquad \frac{9 \mathrm{L}^2 \eta^2}{K} \sum_{k=0}^{K-1} \frac{1}{\Omega^k} \sum_{e=0}^{E-1} \sum_{b=0}^{B-1} \alpha_b\sum_{u \in \bar{\mathcal{U}}_b^{k,e}} \alpha_u \sum_{l=0}^{\mathrm{L}_u^{k,e} - 1}  \mathbb{E} \Bigg[ \Bigg\Vert \sum_{u' \in \mathcal{U}_b} \alpha_{u'} \nabla f_{u'} (\mathbf{w}_{b}^{k,e}) - \sum_{u' \in \mathcal{U}_b} \alpha_{u'} \nabla f_{u'} (\mathbf{w}_{u'}^{k,e,l}) \Bigg\Vert^2 \Bigg] \nonumber\\
	&\leq \frac{9 \mathrm{L}^2 \beta^2\eta^2}{K} \sum_{k=0}^{K-1} \frac{1}{\Omega^k} \sum_{e=0}^{E-1} \sum_{b=0}^{B-1} \alpha_b\sum_{u \in \bar{\mathcal{U}}_b^{k,e}} \alpha_u \sum_{l=0}^{\mathrm{L}_u^{k,e} - 1}  \mathbb{E} \left[ \left\Vert \mathbf{w}_b^{k,e} - \mathbf{w}_u^{k,e,l} \right\Vert^2 \right] + \nonumber\\
	&\mquad \frac{9 \mathrm{L}^2 \eta^2}{K} \sum_{k=0}^{K-1} \frac{1}{\Omega^k} \sum_{e=0}^{E-1} \sum_{b=0}^{B-1} \alpha_b\sum_{u \in \bar{\mathcal{U}}_b^{k,e}} \alpha_u \sum_{l=0}^{\mathrm{L}_u^{k,e} - 1}  \mathbb{E} \left[ \left\Vert \nabla f_u (\mathbf{w}_b^{k,e}) - \nabla f_b (\mathbf{w}_{b}^{k,e}) \right\Vert^2\right] + \nonumber\\
	&\mquad \frac{9 \mathrm{L}^2 \beta^2\eta^2}{K} \sum_{k=0}^{K-1} \frac{1}{\Omega^k} \sum_{e=0}^{E-1} \sum_{b=0}^{B-1} \alpha_b\sum_{u \in \bar{\mathcal{U}}_b^{k,e}} \alpha_u \sum_{l=0}^{\mathrm{L}_u^{k,e} - 1}  \mathbb{E} \left[ \left\Vert \mathbf{w}_{b}^{k,e} - \mathbf{w}_{u}^{k,e,l} \right\Vert^2 \right] \nonumber\\
	&\leq \frac{18 \mathrm{L}^2 \beta^2\eta^2}{K} \sum_{k=0}^{K-1} \frac{1}{\Omega^k} \sum_{e=0}^{E-1} \sum_{b=0}^{B-1} \alpha_b\sum_{u \in \bar{\mathcal{U}}_b^{k,e}} \alpha_u \sum_{l=0}^{\mathrm{L}_u^{k,e} - 1}  \mathbb{E} \left[ \left\Vert \mathbf{w}_b^{k,e} - \mathbf{w}_u^{k,e,l} \right\Vert^2 \right] + \frac{9 \mathrm{L}^2 \eta^2}{K} \sum_{k=0}^{K-1} \frac{1}{\Omega^k} \sum_{e=0}^{E-1} \sum_{b=0}^{B-1} \alpha_b\sum_{u \in \bar{\mathcal{U}}_b^{k,e}} \alpha_u \sum_{l=0}^{\mathrm{L}_u^{k,e} - 1} \epsilon_0^2 \nonumber\\
	&\leq \frac{9 E \epsilon_0^2 \eta^2 \mathrm{L}^3}{K} \sum_{k=0}^{K-1} \frac{1}{\Omega^k} + \frac{18 \mathrm{L}^2 \beta^2\eta^2}{K} \sum_{k=0}^{K-1} \frac{1}{\Omega^k} \sum_{e=0}^{E-1} \sum_{b=0}^{B-1} \alpha_b\sum_{u \in \bar{\mathcal{U}}_b^{k,e}} \alpha_u \sum_{l=0}^{\mathrm{L}_u^{k,e} - 1}  \mathbb{E} \left[ \left\Vert \mathbf{w}_b^{k,e} - \mathbf{w}_u^{k,e,l} \right\Vert^2 \right] .  
\end{align}

For the third term of (\ref{proofLocal_BS_Div}), we write 
\begin{align}
	\label{proofLocal_BS_Div_Term3}
	&\frac{3 \eta^2}{K} \sum_{k=0}^{K-1} \frac{1}{\Omega^k} \sum_{e=0}^{E-1} \sum_{b=0}^{B-1} \alpha_b\sum_{u \in \bar{\mathcal{U}}_b^{k,e}} \alpha_u \sum_{l=0}^{\mathrm{L}_u^{k,e} - 1} \mathbb{E} \Bigg[ \Bigg\Vert \sum_{u' \in \mathcal{U}_b} \alpha_{u'}  \sum_{l'=0}^{l-1} g (\mathbf{w}_{u'}^{k,e,l'}) \left[ 1 - \frac{\mathrm{1}_{u',\mathrm{sc}}^{k,e}}{\mathrm{p}_{u',\mathrm{sc}}^{k,e}} \right] \Bigg\Vert^2 \Bigg] \nonumber\\
	&\leq \frac{3 \mathrm{L} \eta^2}{K} \sum_{k=0}^{K-1} \frac{1}{\Omega^k} \sum_{e=0}^{E-1} \sum_{b=0}^{B-1} \alpha_b\sum_{u \in \bar{\mathcal{U}}_b^{k,e}} \alpha_u \sum_{l=0}^{\mathrm{L}_u^{k,e} - 1} \sum_{l'=0}^{l-1} \mathbb{E} \Bigg[ \Bigg\Vert   g (\mathbf{w}_{u}^{k,e,l'}) \left[ 1 - \frac{\mathrm{1}_{u,\mathrm{sc}}^{k,e}}{\mathrm{p}_{u,\mathrm{sc}}^{k,e}} \right] \Bigg\Vert^2 \Bigg] \nonumber\\
	&\leq \frac{3 \mathrm{L}^2 \eta^2}{K} \sum_{k=0}^{K-1} \frac{1}{\Omega^k} \sum_{e=0}^{E-1} \sum_{b=0}^{B-1} \alpha_b\sum_{u \in \bar{\mathcal{U}}_b^{k,e}} \alpha_u \sum_{l=0}^{\mathrm{L}_u^{k,e} - 1} \mathbb{E} \Bigg[ \Bigg\Vert g_u (\mathbf{w}_u^{k, e, l}) \left[ 1 - \frac{\mathrm{1}_{u,\mathrm{sc}}^{k,e}}{\mathrm{p}_{u,\mathrm{sc}}^{k,e}} \right] \Bigg\Vert^2 \Bigg] \nonumber\\
	&= \frac{3 \mathrm{L}^2 \eta^2}{K} \sum_{k=0}^{K-1} \frac{1}{\Omega^k} \sum_{e=0}^{E-1} \sum_{b=0}^{B-1} \alpha_b\sum_{u \in \bar{\mathcal{U}}_b^{k,e}} \alpha_u \bigg[\frac{1}{\mathrm{p}_{u,\mathrm{sc}}^{k,e}} - 1 \bigg]\sum_{l=0}^{\mathrm{L}_u^{k,e} - 1} \mathbb{E} \left[ \left\Vert g_u (\mathbf{w}_u^{k, e, l})  \right\Vert^2 \right]
\end{align}

To this end, plugging (\ref{proofLocal_BS_Div_Term1}), (\ref{proofLocal_BS_Div_Term2}) and (\ref{proofLocal_BS_Div_Term3}) into (\ref{proofLocal_BS_Div}), and rearranging the terms, we get
\begin{align}
	\label{proofLocal_BS_Div_1}
	&\frac{1}{K} \sum_{k=0}^{K-1} \frac{1}{\Omega^k} \sum_{e=0}^{E-1} \sum_{b=0}^{B-1} \alpha_b\sum_{u \in \bar{\mathcal{U}}_b^{k,e}} \alpha_u \sum_{l=0}^{\mathrm{L}_u^{k,e} - 1} \mathbb{E} \left[ \left\Vert \mathbf{w}_{b}^{k,e} - \mathbf{w}_{u}^{k,e,l} \right\Vert^2 \right] \nonumber\\
	&\leq \frac{ \frac{3 E \mathrm{L}^2 \eta^2 \sigma^2}{K} \sum_{k=0}^{K-1} \frac{1}{\Omega^k} + \frac{9 E \epsilon_0^2 \eta^2 \mathrm{L}^3}{K} \sum_{k=0}^{K-1} \frac{1}{\Omega^k} + \frac{3 \mathrm{L}^2 \eta^2}{K} \sum_{k=0}^{K-1} \frac{1}{\Omega^k} \sum_{e=0}^{E-1} \sum_{b=0}^{B-1} \alpha_b\sum_{u \in \bar{\mathcal{U}}_b^{k,e}} \alpha_u \left[\frac{1}{\mathrm{p}_{u,\mathrm{sc}}^{k,e}} - 1 \right]\sum_{l=0}^{\mathrm{L}_u^{k,e} - 1} \mathbb{E} \left[ \left\Vert g_u (\mathbf{w}_u^{k, e, l})  \right\Vert^2 \right]}{1 - 18 \beta^2 \eta^2 \mathrm{L}^2}
\end{align}
Notice that, when learning rate $\eta < \frac{1}{3\sqrt{2} \beta \mathrm{L}}$, we have $0 < (1 - 18\beta^2\eta^2\mathrm{L}^2) <1$. 
Besides, we can assume that $\eta < \mathrm{min}\left\{\frac{1}{3\sqrt{2} \beta \mathrm{L}}, \frac{1}{\beta E \mathrm{L}} \right\}$ to satisfy all assumptions for the learning rate. 
As such, we write the upper bound as
\begin{align}
	&\frac{1}{K} \sum_{k=0}^{K-1} \frac{1}{\Omega^k} \sum_{e=0}^{E-1} \sum_{b=0}^{B-1} \alpha_b\sum_{u \in \bar{\mathcal{U}}_b^{k,e}} \alpha_u \sum_{l=0}^{\mathrm{L}_u^{k,e} - 1} \mathbb{E} \left[ \left\Vert \mathbf{w}_{b}^{k,e} - \mathbf{w}_{u}^{k,e,l} \right\Vert^2 \right] \nonumber\\ 
	&\leq \frac{3 E \mathrm{L}^2 \eta^2 \sigma^2}{K} \sum_{k=0}^{K-1} \frac{1}{\Omega^k} + \frac{9 E \epsilon_0^2 \eta^2 \mathrm{L}^3}{K} \sum_{k=0}^{K-1} \frac{1}{\Omega^k} + \frac{3 \mathrm{L}^2 \eta^2}{K} \sum_{k=0}^{K-1} \frac{1}{\Omega^k} \sum_{e=0}^{E-1} \sum_{b=0}^{B-1} \alpha_b\sum_{u \in \bar{\mathcal{U}}_b^{k,e}} \alpha_u \left[\frac{1}{\mathrm{p}_{u,\mathrm{sc}}^{k,e}} - 1 \right]\sum_{l=0}^{\mathrm{L}_u^{k,e} - 1} \mathbb{E} \left[ \left\Vert g_u (\mathbf{w}_u^{k, e, l})  \right\Vert^2 \right].
\end{align}
This concludes the proof of Lemma \ref{lem_div_BS_UE_models}.

\subsection{Proof of Lemma \ref{lem_div_central_BS_models}}
\label{lem_div_central_BS_models_Proof}

\begin{align}
	\label{lemma_bs_server_eqn_0}
	&\frac{1}{K} \sum_{k=0}^{K-1} \frac{1}{\Omega^k} \sum_{e=0}^{E-1} \sum_{b=0}^{B-1} \alpha_b \mathbb{E} \left[ \left\Vert \mathbf{w}^k - \mathbf{w}_{b}^{k,e} \right\Vert^2 \right] \nonumber\\
	&=\frac{ \eta^2}{K} \sum_{k=0}^{K-1} \frac{1}{\Omega^k} \sum_{e=0}^{E-1} \sum_{b=0}^{B-1} \alpha_b \mathbb{E} \Bigg[\Bigg\Vert \sum_{e'=0}^{e-1} \sum_{u \in \bar{\mathcal{U}}_b^{k,e'}} \alpha_u \frac{\mathrm{1}_{u,\mathrm{sc}}^{k,e'} }{\mathrm{p}_{u,\mathrm{sc}}^{k,e'}} \mathrm{g}_u^{k,e'} - \sum_{e'=0}^{e-1} \sum_{b'=0}^B \alpha_{b'} \sum_{u' \in \bar{\mathcal{U}}_{b'}^{k,e'}} \alpha_{u'} \frac{\mathrm{1}_{u',\mathrm{sc}}^{k,e'} }{\mathrm{p}_{u',\mathrm{sc}}^{k,e'}} \mathrm{g}_{u'}^{k,e'} \Bigg\Vert^2 \Bigg] \nonumber\\ 
	&=\frac{\eta^2}{K} \sum_{k=0}^{K-1} \frac{1}{\Omega^k} \sum_{e=0}^{E-1} \sum_{b=0}^{B-1} \alpha_b \mathbb{E} \Bigg[\Bigg\Vert \sum_{e'=0}^{e-1} \Bigg\{\sum_{u \in \bar{\mathcal{U}}_b^{k,e'}} \alpha_u \sum_{l=0}^{\mathrm{L}_u^{k,e'} - 1} \bigg( \frac{\mathrm{1}_{u,\mathrm{sc}}^{k,e'} }{\mathrm{p}_{u,\mathrm{sc}}^{k,e'}}  g_u (\mathbf{w}_u^{k, e', l}) \pm \nabla f_u( \mathbf{w}_u^{k,e',l}) \bigg) - \nonumber \\
	&\bquad \bquad \sum_{b'=0}^B \alpha_{b'} \sum_{u' \in \bar{\mathcal{U}}_{b'}^{k,e'}} \alpha_{u'} \sum_{l=0}^{\mathrm{L}_{u'}^{k,e'} - 1} \bigg(\frac{\mathrm{1}_{u',\mathrm{sc}}^{k,e'} }{\mathrm{p}_{u',\mathrm{sc}}^{k,e'}} g (\mathbf{w}_{u'}^{k,e',l}) \pm \nabla f_{u'} (\mathbf{w}_{u'}^{k,e',l}) \bigg) \Bigg\} \Bigg\Vert^2 \Bigg] \nonumber\\
	&\leq \frac{2 \eta^2}{K} \sum_{k=0}^{K-1} \frac{1}{\Omega^k} \sum_{e=0}^{E-1} \sum_{b=0}^{B-1} \alpha_b \mathbb{E} \Bigg[\Bigg\Vert \sum_{e'=0}^{e-1} \Bigg\{\sum_{u \in \bar{\mathcal{U}}_b^{k,e'}} \alpha_u \sum_{l=0}^{\mathrm{L}_u^{k,e'} - 1} \bigg( \frac{\mathrm{1}_{u,\mathrm{sc}}^{k,e'} }{\mathrm{p}_{u,\mathrm{sc}}^{k,e'}}  g_u (\mathbf{w}_u^{k, e', l}) - \nabla f_u( \mathbf{w}_u^{k,e',l}) \bigg) - \nonumber\\
	&\bquad \mquad \sum_{b'=0}^B \alpha_{b'} \sum_{u' \in \bar{\mathcal{U}}_{b'}^{k,e'}} \alpha_{u'} \sum_{l=0}^{\mathrm{L}_{u'}^{k,e'} - 1} \bigg(\frac{\mathrm{1}_{u',\mathrm{sc}}^{k,e'} }{\mathrm{p}_{u',\mathrm{sc}}^{k,e'}} g (\mathbf{w}_{u'}^{k,e',l}) - \nabla f_{u'} (\mathbf{w}_{u'}^{k,e',l}) \bigg) \Bigg\} \Bigg]\Bigg\Vert^2 + \nonumber\\
	&\qquad \quad \frac{2 \eta^2}{K} \sum_{k=0}^{K-1} \frac{1}{\Omega^k} \sum_{e=0}^{E-1} \sum_{b=0}^{B-1} \alpha_b \mathbb{E} \Bigg[\Bigg\Vert \sum_{e'=0}^{e-1} \Bigg\{\sum_{u \in \bar{\mathcal{U}}_b^{k,e'}} \alpha_u \sum_{l=0}^{\mathrm{L}_u^{k,e'} - 1} \nabla f_{u} (\mathbf{w}_{u}^{k,e',l}) - \sum_{b'=0}^B \alpha_{b'} \sum_{u' \in \bar{\mathcal{U}}_{b'}^{k,e'}} \alpha_{u'} \sum_{l=0}^{\mathrm{L}_{u'}^{k,e'} - 1} \nabla f_{u'} (\mathbf{w}_{u'}^{k,e',l}) \Bigg\}\Bigg\Vert^2 \Bigg],
\end{align}  

The first part of (\ref{lemma_bs_server_eqn_0}) is simplified as
\begin{align}
	\label{lemma_bs_server_eqn_0_0}
	&\frac{2 \eta^2}{K} \sum_{k=0}^{K-1} \frac{1}{\Omega^k} \sum_{e=0}^{E-1} \sum_{b=0}^{B-1} \alpha_b \mathbb{E} \Bigg[\Bigg\Vert \sum_{e'=0}^{e-1} \Bigg\{\sum_{u \in \bar{\mathcal{U}}_b^{k,e'}} \alpha_u \sum_{l=0}^{\mathrm{L}_u^{k,e'} - 1} \bigg( \frac{\mathrm{1}_{u,\mathrm{sc}}^{k,e'} }{\mathrm{p}_{u,\mathrm{sc}}^{k,e'}}  g_u (\mathbf{w}_u^{k, e', l}) - \nabla f_u( \mathbf{w}_u^{k,e',l}) \bigg) - \nonumber\\
	&\bquad \mquad \sum_{b'=0}^B \alpha_{b'} \sum_{u' \in \bar{\mathcal{U}}_{b'}^{k,e'}} \alpha_{u'} \sum_{l=0}^{\mathrm{L}_{u'}^{k,e'} - 1} \bigg(\frac{\mathrm{1}_{u',\mathrm{sc}}^{k,e'} }{\mathrm{p}_{u',\mathrm{sc}}^{k,e'}} g (\mathbf{w}_{u'}^{k,e',l}) - \nabla f_{u'} (\mathbf{w}_{u'}^{k,e',l}) \bigg) \Bigg\} \Bigg]\Bigg\Vert^2 \nonumber\\
	&\overset{(a)}{=} \frac{2 \eta^2}{K} \sum_{k=0}^{K-1} \frac{1}{\Omega^k} \sum_{e=0}^{E-1} \sum_{b=0}^{B-1} \alpha_b \mathbb{E} \Bigg[\Bigg\Vert \sum_{e'=0}^{e-1} \Bigg\{\sum_{u \in \bar{\mathcal{U}}_b^{k,e'}} \alpha_u \sum_{l=0}^{\mathrm{L}_u^{k,e'} - 1} \bigg( \frac{\mathrm{1}_{u,\mathrm{sc}}^{k,e'} }{\mathrm{p}_{u,\mathrm{sc}}^{k,e'}}  g_u (\mathbf{w}_u^{k, e', l}) - \nabla f_u( \mathbf{w}_u^{k,e',l}) \bigg) \Bigg\Vert^2\Bigg] - \nonumber\\
	&\bquad \frac{2 \eta^2}{K} \sum_{k=0}^{K-1} \frac{1}{\Omega^k} \sum_{e=0}^{E-1} \mathbb{E} \Bigg[\Bigg\Vert \sum_{e'=0}^{e-1} \Bigg\{\sum_{b=0}^B \alpha_b \sum_{u \in \bar{\mathcal{U}}_{b}^{k,e'}} \alpha_u \sum_{l=0}^{\mathrm{L}_{u}^{k,e'} - 1} \bigg(\frac{\mathrm{1}_{u,\mathrm{sc}}^{k,e'} }{\mathrm{p}_{u,\mathrm{sc}}^{k,e'}} g_u (\mathbf{w}_u^{k, e', l}) - \nabla f_{u} (\mathbf{w}_{u}^{k,e',l}) \bigg) \Bigg\} \Bigg]\Bigg\Vert^2 \nonumber\\
	&\overset{(b)}{=} \frac{2 \eta^2}{K} \sum_{k=0}^{K-1} \frac{1}{\Omega^k} \sum_{e=0}^{E-1} \sum_{b=0}^{B-1} \alpha_b \sum_{e'=0}^{e-1} \mathbb{E} \Bigg[\Bigg\Vert \sum_{u \in \bar{\mathcal{U}}_b^{k,e'}} \alpha_u \sum_{l=0}^{\mathrm{L}_u^{k,e'} - 1} \bigg( \frac{\mathrm{1}_{u,\mathrm{sc}}^{k,e'} }{\mathrm{p}_{u,\mathrm{sc}}^{k,e'}}  g_u (\mathbf{w}_u^{k, e', l}) - \nabla f_u( \mathbf{w}_u^{k,e',l}) \bigg) \Bigg\Vert^2\Bigg] -  \nonumber\\
	&\bquad \frac{2 \eta^2}{K} \sum_{k=0}^{K-1} \frac{1}{\Omega^k} \sum_{e=0}^{E-1} \sum_{e'=0}^{e-1} \mathbb{E} \Bigg[\Bigg\Vert \sum_{b=0}^B \alpha_b \sum_{u \in \bar{\mathcal{U}}_{b}^{k,e'}} \alpha_u \sum_{l=0}^{\mathrm{L}_{u}^{k,e'} - 1} \bigg(\frac{\mathrm{1}_{u,\mathrm{sc}}^{k,e'} }{\mathrm{p}_{u,\mathrm{sc}}^{k,e'}} g_u (\mathbf{w}_u^{k, e', l}) - \nabla f_{u} (\mathbf{w}_{u}^{k,e',l}) \bigg) \Bigg]\Bigg\Vert^2  \nonumber\\
	&\overset{(c)}{\leq} \frac{2 E \eta^2}{K} \sum_{k=0}^{K-1} \frac{1}{\Omega^k} \sum_{e=0}^{E-1} \sum_{b=0}^{B-1} \alpha_b \mathbb{E} \Bigg[\Bigg\Vert \sum_{u \in \bar{\mathcal{U}}_b^{k,e}} \alpha_u \sum_{l=0}^{\mathrm{L}_u^{k,e} - 1} \bigg( \frac{\mathrm{1}_{u,\mathrm{sc}}^{k,e} }{\mathrm{p}_{u,\mathrm{sc}}^{k,e}}  g_u (\mathbf{w}_u^{k, e, l}) - \nabla f_u( \mathbf{w}_u^{k,e,l}) \bigg) \Bigg\Vert^2\Bigg] -  \nonumber\\
	&\bquad \frac{2 E \eta^2}{K} \sum_{k=0}^{K-1} \frac{1}{\Omega^k} \sum_{e=0}^{E-1} \mathbb{E} \Bigg[\Bigg\Vert \sum_{b=0}^B \alpha_b \sum_{u \in \bar{\mathcal{U}}_{b}^{k,e}} \alpha_u \sum_{l=0}^{\mathrm{L}_{u}^{k,e} - 1} \bigg(\frac{\mathrm{1}_{u,\mathrm{sc}}^{k,e} }{\mathrm{p}_{u,\mathrm{sc}}^{k,e}} g_u (\mathbf{w}_u^{k, e, l}) - \nabla f_{u} (\mathbf{w}_{u}^{k,e,l}) \bigg) \Bigg]\Bigg\Vert^2  \nonumber\\
	&\overset{(d)}{=} \frac{2 E \eta^2}{K} \sum_{k=0}^{K-1} \frac{1}{\Omega^k} \sum_{e=0}^{E-1} \sum_{b=0}^{B-1} \alpha_b \sum_{u \in \bar{\mathcal{U}}_b^{k,e}} \left(\alpha_u\right)^2 \sum_{l=0}^{\mathrm{L}_u^{k,e} - 1} \mathbb{E} \Bigg[\Bigg\Vert \frac{\mathrm{1}_{u,\mathrm{sc}}^{k,e} }{\mathrm{p}_{u,\mathrm{sc}}^{k,e}}  g_u (\mathbf{w}_u^{k, e, l}) - \nabla f_u( \mathbf{w}_u^{k,e,l}) \Bigg\Vert^2\Bigg] - \nonumber\\
	&\bquad \frac{2 E \eta^2}{K} \sum_{k=0}^{K-1} \frac{1}{\Omega^k} \sum_{e=0}^{E-1} \sum_{b=0}^B \left(\alpha_b\right)^2 \sum_{u \in \bar{\mathcal{U}}_{b}^{k,e}} \left(\alpha_u\right)^2 \sum_{l=0}^{\mathrm{L}_{u}^{k,e} - 1}  \mathbb{E} \Bigg[\Bigg\Vert \frac{\mathrm{1}_{u,\mathrm{sc}}^{k,e} }{\mathrm{p}_{u,\mathrm{sc}}^{k,e}} g_u (\mathbf{w}_u^{k, e, l}) - \nabla f_{u'} (\mathbf{w}_{u}^{k,e,l}) \Bigg]\Bigg\Vert^2 \nonumber\\
	&= \frac{2 E \eta^2}{K} \sum_{k=0}^{K-1} \frac{1}{\Omega^k} \sum_{e=0}^{E-1} \sum_{b=0}^{B-1} \alpha_b \sum_{u \in \bar{\mathcal{U}}_b^{k,e}} \left(\alpha_u\right)^2 \sum_{l=0}^{\mathrm{L}_u^{k,e} - 1} \mathbb{E} \left[\left\Vert \frac{\mathrm{1}_{u,\mathrm{sc}}^{k,e} }{\mathrm{p}_{u,\mathrm{sc}}^{k,e}}  g_u (\mathbf{w}_u^{k, e, l}) \pm g_u (\mathbf{w}_u^{k, e, l}) - \nabla f_u( \mathbf{w}_u^{k,e,l}) \right\Vert^2\right] - \nonumber\\
	&\bquad \frac{2 E \eta^2}{K} \sum_{k=0}^{K-1} \frac{1}{\Omega^k} \sum_{e=0}^{E-1} \sum_{b=0}^B \left(\alpha_b\right)^2 \sum_{u \in \bar{\mathcal{U}}_{b}^{k,e}} \left(\alpha_u \right)^2 \sum_{l=0}^{\mathrm{L}_{u}^{k,e} - 1} \mathbb{E} \left[\left\Vert \frac{\mathrm{1}_{u,\mathrm{sc}}^{k,e} }{\mathrm{p}_{u,\mathrm{sc}}^{k,e}} g_u (\mathbf{w}_u^{k, e, l}) \pm g_u (\mathbf{w}_u^{k, e, l}) - \nabla f_{u} (\mathbf{w}_{u}^{k,e,l}) \right]\right\Vert^2 \nonumber\\
	&\overset{(e)}{\leq} \frac{4 E \eta^2}{K} \sum_{k=0}^{K-1} \frac{1}{\Omega^k} \sum_{e=0}^{E-1} \sum_{b=0}^{B-1} \alpha_b \sum_{u \in \bar{\mathcal{U}}_b^{k,e}} \left(\alpha_u\right)^2 \sum_{l=0}^{\mathrm{L}_u^{k,e} - 1} \mathbb{E} \left[\left\Vert \left(\frac{\mathrm{1}_{u,\mathrm{sc}}^{k,e} }{\mathrm{p}_{u,\mathrm{sc}}^{k,e}} - 1 \right) g_u (\mathbf{w}_u^{k, e, l}) \right\Vert^2\right] + \nonumber\\
	&\bquad \frac{4 E \eta^2}{K} \sum_{k=0}^{K-1} \frac{1}{\Omega^k} \sum_{e=0}^{E-1} \sum_{b=0}^{B-1} \alpha_b \sum_{u \in \bar{\mathcal{U}}_b^{k,e}} \left(\alpha_u\right)^2 \sum_{l=0}^{\mathrm{L}_u^{k,e} - 1} \mathbb{E} \left[\left \Vert  g_u (\mathbf{w}_u^{k, e, l}) - \nabla f_u( \mathbf{w}_u^{k,e,l}) \right\Vert^2\right] - \nonumber\\
	&\bquad \frac{4 E \eta^2}{K} \sum_{k=0}^{K-1} \frac{1}{\Omega^k} \sum_{e=0}^{E-1} \sum_{b=0}^B \left(\alpha_b\right)^2 \sum_{u \in \bar{\mathcal{U}}_{b}^{k,e}} \left(\alpha_u \right)^2 \sum_{l=0}^{\mathrm{L}_{u}^{k,e} - 1} \mathbb{E} \left[\left\Vert \left(\frac{\mathrm{1}_{u,\mathrm{sc}}^{k,e} }{\mathrm{p}_{u,\mathrm{sc}}^{k,e}} - 1 \right) g_u (\mathbf{w}_u^{k, e, l}) \right\Vert^2 \right] - \nonumber \\
	&\bquad \frac{4 E \eta^2}{K} \sum_{k=0}^{K-1} \frac{1}{\Omega^k} \sum_{e=0}^{E-1} \sum_{b=0}^B \left(\alpha_b\right)^2 \sum_{u \in \bar{\mathcal{U}}_{b}^{k,e}} \left(\alpha_u \right)^2 \sum_{l=0}^{\mathrm{L}_{u}^{k,e} - 1} \mathbb{E} \left[\left\Vert g_u (\mathbf{w}_u^{k, e, l}) - \nabla f_{u} (\mathbf{w}_{u}^{k,e,l}) \right]\right\Vert^2 \nonumber \\
	&\overset{(f)}{\leq} \frac{4 E \eta^2 \sigma^2} {K} \sum_{k=0}^{K-1} \frac{1}{\Omega^k} \sum_{e=0}^{E-1} \sum_{b=0}^{B-1} \alpha_b \sum_{u \in \bar{\mathcal{U}}_b^{k,e}} \left(\alpha_u\right)^2 \mathrm{L}_u^{k,e} + \frac{4 E \eta^2}{K} \sum_{k=0}^{K-1} \frac{1}{\Omega^k} \sum_{e=0}^{E-1} \sum_{b=0}^{B-1} \alpha_b \sum_{u \in \bar{\mathcal{U}}_b^{k,e}} \left(\alpha_u\right)^2 \bigg[\frac{1}{\mathrm{p}_{u,\mathrm{sc}}^{k,e}} - 1 \bigg] \sum_{l=0}^{\mathrm{L}_u^{k,e} - 1} \mathbb{E} \left[\left\Vert g_u (\mathbf{w}_u^{k, e, l}) \right\Vert^2\right] - \nonumber\\
	&\frac{4 E \eta^2 \sigma^2}{K} \sum_{k=0}^{K-1} \frac{1}{\Omega^k} \sum_{e=0}^{E-1} \sum_{b=0}^B \left(\alpha_b\right)^2 \sum_{u \in \bar{\mathcal{U}}_{b}^{k,e}} \left(\alpha_u \right)^2 \mathrm{L}_{u}^{k,e} - \frac{4 E \eta^2}{K} \sum_{k=0}^{K-1} \frac{1}{\Omega^k} \sum_{e=0}^{E-1} \sum_{b=0}^B \left(\alpha_b\right)^2 \sum_{u \in \bar{\mathcal{U}}_{b}^{k,e}} \left(\alpha_u \right)^2 \bigg[\frac{1} {\mathrm{p}_{u,\mathrm{sc}}^{k,e}} - 1 \bigg] \sum_{l=0}^{\mathrm{L}_{u}^{k,e} - 1}  \mathbb{E} \left[\left\Vert g_u (\mathbf{w}_u^{k, e, l}) \right\Vert^2 \right] \nonumber \\
	&\leq \frac{4 E \eta^2 \sigma^2} {K} \sum_{k=0}^{K-1} \frac{1}{\Omega^k} \sum_{e=0}^{E-1} \sum_{b=0}^{B-1} \alpha_b \rs\rs \sum_{u \in \bar{\mathcal{U}}_b^{k,e}} \rs \left(\alpha_u\right)^2 \mathrm{L}_u^{k,e} + \frac{4 E \eta^2}{K} \sum_{k=0}^{K-1} \frac{1}{\Omega^k} \sum_{e=0}^{E-1} \sum_{b=0}^{B-1} \rs \alpha_b \rs\rs \sum_{u \in \bar{\mathcal{U}}_b^{k,e}} \rs \left(\alpha_u\right)^2 \rs \bigg[\frac{1}{\mathrm{p}_{u,\mathrm{sc}}^{k,e}} - 1 \bigg] \sum_{l=0}^{\mathrm{L}_u^{k,e} - 1} \mathbb{E} \left[\left\Vert g_u (\mathbf{w}_u^{k, e, l}) \right\Vert^2\right],
\end{align}
where $(a)$ stems from the fact that $\sum_{j=1}^J \alpha_j \Vert \mathbf{x}_j - \bar{\mathbf{x}} \Vert^2 = \sum_{j=1}^J \alpha_j \Vert \mathbf{x}_j \Vert^2 - \Vert \bar{\mathbf{x}} \Vert^2$, where $\bar{\mathbf{x}} = \sum_{j=1}^J \alpha_j \mathbf{x}_j$.
Besides, $(b)$ appears from the independent mini-batch gradients and $\mathbb{E} [g_u(\mathbf{w}_u^{k,e,l})] = \nabla f_u (\mathbf{w}_u^{k,e,l})$ assumptions, which make the expectation of the cross product terms $0$.
Besides, $(c)$ is true since $(e-1)- e' \leq E$.
Furthermore, we use the independent mini-batch gradients and $\mathbb{E} [g_u(\mathbf{w}_u^{k,e,l})] = \nabla f_u (\mathbf{w}_u^{k,e,l})$ assumptions in step $(d)$.
In $(e)$, we use the fact that $\Vert \sum_{j=1}^J \mathbf{x}_j \Vert^2 \leq J \sum_{j=1}^J \Vert \mathbf{x}_j \Vert^2$. 
Moreover, $(f)$ appears from bounded variance of the gradients.
Finally, we drop the negative second term of $(f)$ in the last inequality.

The second part of (\ref{lemma_bs_server_eqn_0}) is simplified as
\begin{align}
	\label{lemma_bs_server_eqn_0_1}
	&\frac{2 \eta^2}{K} \sum_{k=0}^{K-1} \frac{1}{\Omega^k} \sum_{e=0}^{E-1} \sum_{b=0}^{B-1} \alpha_b \mathbb{E} \Bigg[\Bigg\Vert \sum_{e'=0}^{e-1} \Bigg\{\sum_{u \in \bar{\mathcal{U}}_b^{k,e'}} \alpha_u \sum_{l=0}^{\mathrm{L}_u^{k,e'} - 1} \nabla f_{u} (\mathbf{w}_{u}^{k,e',l}) - \sum_{b'=0}^B \alpha_{b'} \sum_{u' \in \bar{\mathcal{U}}_{b'}^{k,e'}} \alpha_{u'} \sum_{l=0}^{\mathrm{L}_{u'}^{k,e'} - 1} \nabla f_{u'} (\mathbf{w}_{u'}^{k,e',l}) \Bigg\}\Bigg\Vert^2 \Bigg] \nonumber\\
	&\overset{(a)}{\leq}\frac{2 E \eta^2}{K} \sum_{k=0}^{K-1} \frac{1}{\Omega^k} \sum_{e=0}^{E-1} \sum_{b=0}^{B-1} \alpha_b \sum_{e'=0}^{e-1} \mathbb{E} \Bigg[\Bigg\Vert  \sum_{u \in \bar{\mathcal{U}}_b^{k,e'}} \alpha_u \sum_{l=0}^{\mathrm{L}_u^{k,e'} - 1} \nabla f_{u} (\mathbf{w}_{u}^{k,e',l}) - \sum_{b'=0}^B \alpha_{b'} \sum_{u' \in \bar{\mathcal{U}}_{b'}^{k,e'}} \alpha_{u'} \sum_{l=0}^{\mathrm{L}_{u'}^{k,e'} - 1} \nabla f_{u'} (\mathbf{w}_{u'}^{k,e',l}) \Bigg\Vert^2 \Bigg] \nonumber \\
	&\overset{(b)}{\leq} \frac{2 \eta^2 E^2}{K} \sum_{k=0}^{K-1} \frac{1}{\Omega^k} \sum_{e=0}^{E-1} \sum_{b=0}^{B-1} \alpha_b \mathbb{E} \Bigg[\Bigg\Vert  \sum_{u \in \bar{\mathcal{U}}_b^{k,e}} \alpha_u \sum_{l=0}^{\mathrm{L}_u^{k,e} - 1} \nabla f_{u} (\mathbf{w}_{u}^{k,e,l}) - \sum_{b'=0}^B \alpha_{b'} \sum_{u' \in \bar{\mathcal{U}}_{b'}^{k,e}} \alpha_{u'} \sum_{l=0}^{\mathrm{L}_{u'}^{k,e} - 1} \nabla f_{u'} (\mathbf{w}_{u'}^{k,e,l}) \Bigg\Vert^2 \Bigg] \nonumber\\
	&= \frac{2 \eta^2 E^2}{K} \sum_{k=0}^{K-1} \frac{1}{\Omega^k} \sum_{e=0}^{E-1} \sum_{b=0}^{B-1} \alpha_b \mathbb{E} \Bigg[\Bigg\Vert  \sum_{u \in \bar{\mathcal{U}}_b^{k,e}} \alpha_u \sum_{l=0}^{\mathrm{L}_u^{k,e} - 1} \big[\nabla f_{u} (\mathbf{w}_{u}^{k,e,l}) - \nabla f_{u} (\mathbf{w}_b^{k,e}) \big] + \sum_{u \in \bar{\mathcal{U}}_b^{k,e}} \alpha_u \sum_{l=0}^{\mathrm{L}_u^{k,e} - 1} \big[\nabla f_{u} (\mathbf{w}_{b}^{k,e}) - \nabla f_{u} (\mathbf{w}^{k}) \big] +\nonumber\\
	&\squad \sum_{u \in \bar{\mathcal{U}}_b^{k,e}} \alpha_u \sum_{l=0}^{\mathrm{L}_u^{k,e} - 1} \nabla f_{u} (\mathbf{w}^k) - \sum_{b'=0}^B \alpha_{b'} \sum_{u' \in \bar{\mathcal{U}}_{b'}^{k,e}} \alpha_{u'} \sum_{l=0}^{\mathrm{L}_{u'}^{k,e} - 1} \nabla f_{u'} (\mathbf{w}^k) + \sum_{b'=0}^B \alpha_{b'} \sum_{u' \in \bar{\mathcal{U}}_{b'}^{k,e}} \alpha_{u'} \sum_{l=0}^{\mathrm{L}_{u'}^{k,e} - 1} \big[\nabla f_{u'} (\mathbf{w}^k) - \nabla f_{u'} (\mathbf{w}_{b'}^{k,e}) \big] + \nonumber\\
	&\squad \sum_{b'=0}^B \alpha_{b'} \sum_{u' \in \bar{\mathcal{U}}_{b'}^{k,e}} \alpha_{u'} \sum_{l=0}^{\mathrm{L}_{u'}^{k,e} - 1} \big[ \nabla f_{u'} (\mathbf{w}_{b'}^{k,e}) - \nabla f_{u'} (\mathbf{w}_{u'}^{k,e,l}) \big] \Bigg\Vert^2 \Bigg] \nonumber\\
	&\overset{(c)}{\leq} \frac{10 \eta^2 E^2}{K} \sum_{k=0}^{K-1} \frac{1}{\Omega^k} \sum_{e=0}^{E-1} \sum_{b=0}^{B-1} \alpha_b \mathbb{E} \Bigg[\Bigg\Vert \sum_{u \in \bar{\mathcal{U}}_b^{k,e}} \alpha_u \sum_{l=0}^{\mathrm{L}_u^{k,e} - 1} \big[\nabla f_{u} (\mathbf{w}_{u}^{k,e,l}) - \nabla f_{u} (\mathbf{w}_b^{k,e}) \big] \Bigg\Vert^2\Bigg] + \nonumber\\
	&\squad \frac{10 \eta^2 E^2}{K} \sum_{k=0}^{K-1} \frac{1}{\Omega^k} \sum_{e=0}^{E-1} \sum_{b=0}^{B-1} \alpha_b \mathbb{E} \Bigg[\Bigg\Vert \sum_{u \in \bar{\mathcal{U}}_b^{k,e}} \alpha_u \sum_{l=0}^{\mathrm{L}_u^{k,e} - 1} \big[\nabla f_{u} (\mathbf{w}_{b}^{k,e}) - \nabla f_{u} (\mathbf{w}^{k}) \big] \Bigg\Vert^2 \Bigg] + \nonumber\\
	&\squad \frac{10 \eta^2 E^2}{K} \sum_{k=0}^{K-1} \frac{1}{\Omega^k} \sum_{e=0}^{E-1} \sum_{b=0}^{B-1} \alpha_b \mathbb{E} \Bigg[\Bigg\Vert \sum_{u \in \bar{\mathcal{U}}_b^{k,e}} \alpha_u \sum_{l=0}^{\mathrm{L}_u^{k,e} - 1} \nabla f_{u} (\mathbf{w}^k) - \sum_{b'=0}^B \alpha_{b'} \sum_{u' \in \bar{\mathcal{U}}_{b'}^{k,e}} \alpha_{u'} \sum_{l=0}^{\mathrm{L}_{u'}^{k,e} - 1} \nabla f_{u'} (\mathbf{w}^k) \Bigg\Vert^2 \Bigg] + \nonumber\\
	&\squad \frac{10 \eta^2 E^2}{K} \sum_{k=0}^{K-1} \frac{1}{\Omega^k} \sum_{e=0}^{E-1} \sum_{b=0}^{B-1} \alpha_b \mathbb{E} \Bigg[\Bigg\Vert \sum_{b'=0}^B \alpha_{b'} \sum_{u' \in \bar{\mathcal{U}}_{b'}^{k,e}} \alpha_{u'} \sum_{l=0}^{\mathrm{L}_{u'}^{k,e} - 1} \big[\nabla f_{u'} (\mathbf{w}^k) - \nabla f_{u'} (\mathbf{w}_{b'}^{k,e}) \big] \Bigg\Vert^2\Bigg] + \nonumber\\
	&\squad \frac{10 \eta^2 E^2}{K} \sum_{k=0}^{K-1} \frac{1}{\Omega^k} \sum_{e=0}^{E-1} \sum_{b=0}^{B-1} \alpha_b \mathbb{E} \Bigg[\Bigg\Vert \sum_{b'=0}^B \alpha_{b'} \sum_{u' \in \bar{\mathcal{U}}_{b'}^{k,e}} \alpha_{u'} \sum_{l=0}^{\mathrm{L}_{u'}^{k,e} - 1} \big[ \nabla f_{u'} (\mathbf{w}_{b'}^{k,e}) - \nabla f_{u'} (\mathbf{w}_{u'}^{k,e,l}) \big] \Bigg\Vert^2 \Bigg] \nonumber\\
	&\overset{(d)}{\leq} \frac{10 \eta^2 E^2}{K} \sum_{k=0}^{K-1} \frac{1}{\Omega^k} \sum_{e=0}^{E-1} \sum_{b=0}^{B-1} \alpha_b \sum_{u \in \bar{\mathcal{U}}_b^{k,e}} \alpha_u \mathbb{E} \Bigg[\Bigg\Vert \sum_{l=0}^{\mathrm{L}_u^{k,e} - 1} \big[\nabla f_{u} (\mathbf{w}_{u}^{k,e,l}) - \nabla f_{u} (\mathbf{w}_b^{k,e}) \big] \Bigg\Vert^2\Bigg] + \nonumber\\
	&\squad \frac{10 \eta^2 E^2}{K} \sum_{k=0}^{K-1} \frac{1}{\Omega^k} \sum_{e=0}^{E-1} \sum_{b=0}^{B-1} \alpha_b \sum_{u \in \bar{\mathcal{U}}_b^{k,e}} \alpha_u \mathbb{E} \Bigg[\Bigg\Vert \sum_{l=0}^{\mathrm{L}_u^{k,e} - 1} \big[\nabla f_{u} (\mathbf{w}_{b}^{k,e}) - \nabla f_{u} (\mathbf{w}^{k}) \big] \Bigg\Vert^2 \Bigg] + \nonumber\\
	&\squad \frac{10 \eta^2 E^2}{K} \sum_{k=0}^{K-1} \frac{1}{\Omega^k} \sum_{e=0}^{E-1} \sum_{b=0}^{B-1} \alpha_b \mathbb{E} \Bigg[\Bigg\Vert \sum_{u \in \bar{\mathcal{U}}_b^{k,e}} \alpha_u \nabla \tilde{f}_{u} (\mathbf{w}^k) - \sum_{b'=0}^B \alpha_{b'} \sum_{u' \in \bar{\mathcal{U}}_{b'}^{k,e}} \alpha_{u'} \nabla \tilde{f}_{u'} (\mathbf{w}^k) \Bigg\Vert^2 \Bigg] + \nonumber\\
	&\squad \frac{10 \eta^2 E^2}{K} \sum_{k=0}^{K-1} \frac{1}{\Omega^k} \sum_{e=0}^{E-1} \sum_{b=0}^{B-1} \alpha_b \sum_{u \in \bar{\mathcal{U}}_b^{k,e}} \alpha_u \mathbb{E} \Bigg[\Bigg\Vert \sum_{l=0}^{\mathrm{L}_u^{k,e} - 1} \big[\nabla f_u (\mathbf{w}^k) - \nabla f_u (\mathbf{w}_b^{k,e}) \big] \Bigg\Vert^2\Bigg] + \nonumber\\
	&\squad \frac{10 \eta^2 E^2}{K} \sum_{k=0}^{K-1} \frac{1}{\Omega^k} \sum_{e=0}^{E-1} \sum_{b=0}^{B-1} \alpha_b  \sum_{u \in \bar{\mathcal{U}}_b^{k,e}} \alpha_u \mathbb{E} \Bigg[\Bigg\Vert \sum_{l=0}^{\mathrm{L}_u^{k,e} - 1} \big[ \nabla f_u (\mathbf{w}_b^{k,e}) - \nabla f_u (\mathbf{w}_u^{k,e,l}) \big] \Bigg\Vert^2 \Bigg] \nonumber\\
	&\overset{(e)}{\leq} \frac{10 \mathrm{L} \eta^2 E^2}{K} \sum_{k=0}^{K-1} \frac{1}{\Omega^k} \sum_{e=0}^{E-1} \sum_{b=0}^{B-1} \alpha_b \sum_{u \in \bar{\mathcal{U}}_b^{k,e}} \alpha_u \sum_{l=0}^{\mathrm{L}_u^{k,e} - 1} \mathbb{E} \Big[\left\Vert \nabla f_{u} (\mathbf{w}_{u}^{k,e,l}) - \nabla f_{u} (\mathbf{w}_b^{k,e}) \right\Vert^2\Big] + \nonumber\\
	&\squad \frac{10 \mathrm{L} \eta^2 E^2}{K} \sum_{k=0}^{K-1} \frac{1}{\Omega^k} \sum_{e=0}^{E-1} \sum_{b=0}^{B-1} \alpha_b \sum_{u \in \bar{\mathcal{U}}_b^{k,e}} \alpha_u \sum_{l=0}^{\mathrm{L}_u^{k,e} - 1} \mathbb{E} \Big[\left\Vert \nabla f_{u} (\mathbf{w}_{b}^{k,e}) - \nabla f_{u} (\mathbf{w}^{k}) \right\Vert^2 \Big] + \frac{10 \eta^2 E^2}{K} \sum_{k=0}^{K-1} \frac{1}{\Omega^k} \sum_{e=0}^{E-1} \sum_{b=0}^{B-1} \alpha_b \cdot \epsilon_1^2 + \nonumber\\
	&\squad \frac{10 \mathrm{L} \eta^2 E^2}{K} \sum_{k=0}^{K-1} \frac{1}{\Omega^k} \sum_{e=0}^{E-1} \sum_{b=0}^{B-1} \alpha_b \sum_{u \in \bar{\mathcal{U}}_b^{k,e}} \alpha_u \sum_{l=0}^{\mathrm{L}_u^{k,e} - 1} \mathbb{E} \Big[\left\Vert \nabla f_u (\mathbf{w}^k) - \nabla f_u (\mathbf{w}_b^{k,e}) \right\Vert^2\Big] + \nonumber\\
	&\squad \frac{10 \mathrm{L}\eta^2 E^2}{K} \sum_{k=0}^{K-1} \frac{1}{\Omega^k} \sum_{e=0}^{E-1} \sum_{b=0}^{B-1} \alpha_b  \sum_{u \in \bar{\mathcal{U}}_b^{k,e}} \alpha_u \sum_{l=0}^{\mathrm{L}_u^{k,e} - 1}  \mathbb{E} \Big[\left\Vert \nabla f_u (\mathbf{w}_b^{k,e}) - \nabla f_u (\mathbf{w}_u^{k,e,l}) \right\Vert^2 \Big] \nonumber\\
	&\overset{(f)}{\leq} \frac{10 \epsilon_1^2 \eta^2 E^3}{K} \sum_{k=0}^{K-1} \frac{1}{\Omega^k} + \frac{20 \mathrm{L} \beta^2 \eta^2 E^2}{K} \sum_{k=0}^{K-1} \frac{1}{\Omega^k} \sum_{e=0}^{E-1} \sum_{b=0}^{B-1} \alpha_b \sum_{u \in \bar{\mathcal{U}}_b^{k,e}} \alpha_u \sum_{l=0}^{\mathrm{L}_u^{k,e} - 1} \mathbb{E} \Big[\left\Vert  \mathbf{w}_b^{k,e} - \mathbf{w}_{u}^{k,e,l} \right\Vert^2\Big] + \nonumber\\
	&\squad \frac{20 \beta^2 \eta^2 E^2 \mathrm{L}^2}{K} \sum_{k=0}^{K-1} \frac{1}{\Omega^k} \sum_{e=0}^{E-1} \sum_{b=0}^{B-1} \alpha_b \mathbb{E} \Big[\left\Vert \mathbf{w}^k - \mathbf{w}_b^{k,e} \right\Vert^2\Big] \nonumber\\
	&\overset{(g)}{\leq} \frac{10 \epsilon_1^2 \eta^2 E^3}{K} \sum_{k=0}^{K-1} \frac{1}{\Omega^k} + \frac{20 \beta^2 \eta^2 E^2 \mathrm{L}^2}{K} \sum_{k=0}^{K-1} \frac{1}{\Omega^k} \sum_{e=0}^{E-1} \sum_{b=0}^{B-1} \alpha_b \mathbb{E} \Big[\left\Vert \mathbf{w}^k - \mathbf{w}_b^{k,e} \right\Vert^2\Big] + 20 \mathrm{L} \beta^2 \eta^2 E^2 \times \Bigg[ \frac{3 E \mathrm{L}^2 \eta^2 \sigma^2}{K} \sum_{k=0}^{K-1} \frac{1}{\Omega^k} + \nonumber\\
	&\squad \frac{9 E \epsilon_0^2 \eta^2 \mathrm{L}^3}{K} \sum_{k=0}^{K-1} \frac{1}{\Omega^k} + \frac{3 \mathrm{L}^2 \eta^2}{K} \sum_{k=0}^{K-1} \frac{1}{\Omega^k} \sum_{e=0}^{E-1} \sum_{b=0}^{B-1} \alpha_b\sum_{u \in \bar{\mathcal{U}}_b^{k,e}} \alpha_u \left[\frac{1}{\mathrm{p}_{u,\mathrm{sc}}^{k,e}} - 1 \right]\sum_{l=0}^{\mathrm{L}_u^{k,e} - 1} \mathbb{E} \left[ \left\Vert g_u (\mathbf{w}_u^{k, e, l})  \right\Vert^2 \right] \Bigg] \nonumber\\
	&= \frac{180 \beta^2 \epsilon_0^2 E^3 \mathrm{L}^4 \eta^4}{K} \sum_{k=0}^{K-1} \frac{1}{\Omega^k} + \frac{10 \epsilon_1^2 \eta^2 E^3}{K} \sum_{k=0}^{K-1} \frac{1}{\Omega^k} + \frac{60 \beta^2 \sigma^2 E^3 \mathrm{L}^3 \eta^4}{K} \sum_{k=0}^{K-1} \frac{1}{\Omega^k} + \nonumber\\
	&\squad \frac{60 \beta^2 E^2 \mathrm{L}^3 \eta^4}{K} \sum_{k=0}^{K-1} \frac{1}{\Omega^k} \sum_{e=0}^{E-1} \sum_{b=0}^{B-1} \alpha_b\sum_{u \in \bar{\mathcal{U}}_b^{k,e}} \alpha_u \left[\frac{1}{\mathrm{p}_{u,\mathrm{sc}}^{k,e}} - 1 \right]\sum_{l=0}^{\mathrm{L}_u^{k,e} - 1} \mathbb{E} \left[ \left\Vert g_u (\mathbf{w}_u^{k, e, l})  \right\Vert^2 \right] + \nonumber\\
	&\squad \frac{20 \beta^2 \eta^2 E^2 \mathrm{L}^2}{K} \sum_{k=0}^{K-1} \frac{1}{\Omega^k} \sum_{e=0}^{E-1} \sum_{b=0}^{B-1} \alpha_b \mathbb{E} \Big[\left\Vert \mathbf{w}^k - \mathbf{w}_b^{k,e} \right\Vert^2\Big],
\end{align}
where $(a)$ and $c$ stem from the fact that $\Vert \sum_{j=1}^J \mathbf{x}_j \Vert^2 \leq J \sum_{j=1}^J \Vert \mathbf{x}_j \Vert^2$.
Besides, we used $(e-1) - e' \leq E$ in step $(a)$ and step $(b)$.
In $(d)$, we use the fact that $\Vert \sum_{j=1}^J \alpha_j \mathbf{x}_j \Vert^2 \leq \sum_{j=1}^J \alpha_j \Vert \mathbf{x}_j \Vert^2$ and the definition of $\nabla \tilde{f}_u(\mathbf{w}^k) \coloneqq \sum_{l=0}^{\mathrm{L}_u^{k,e} - 1} \nabla f_u(\mathbf{w}^k)$. 
Furthermore, $\Vert \sum_{j=1}^J \mathbf{x}_j \Vert^2 \leq J \sum_{j=1}^J \Vert \mathbf{x}_j \Vert^2$ and the bounded divergence assumption in (\ref{ES_Central_Loss_Divergence_apndx}) yield $(e)$.
Moreover, the inequality in $(f)$ stems from the $\beta$-smoothness property. 
Finally, we use Lemma \ref{lem_div_BS_UE_models} in step $(g)$.

Now, plugging (\ref{lemma_bs_server_eqn_0_0}) and (\ref{lemma_bs_server_eqn_0_1}) into (\ref{lemma_bs_server_eqn_0}), we get
\begin{align}
	\label{lemma_bs_server_eqn_1}
	&\left(1 - 20 \beta^2 \eta^2 E^2 \mathrm{L}^2 \right) \frac{1}{K} \sum_{k=0}^{K-1} \frac{1}{\Omega^k} \sum_{e=0}^{E-1} \sum_{b=0}^{B-1} \alpha_b \mathbb{E} \left[ \left\Vert \mathbf{w}^k - \mathbf{w}_{b}^{k,e} \right\Vert^2 \right] \nonumber\\
	&\leq \frac{4 E \eta^2 \sigma^2} {K} \sum_{k=0}^{K-1} \frac{1}{\Omega^k} \sum_{e=0}^{E-1} \sum_{b=0}^{B-1} \alpha_b \rs\rs \sum_{u \in \bar{\mathcal{U}}_b^{k,e}} \rs \left(\alpha_u\right)^2 \mathrm{L}_u^{k,e} + \frac{60 \beta^2 \sigma^2 E^3 \mathrm{L}^3 \eta^4}{K} \sum_{k=0}^{K-1} \frac{1}{\Omega^k} + \frac{180 \beta^2 \epsilon_0^2 E^3 \mathrm{L}^4 \eta^4}{K} \sum_{k=0}^{K-1} \frac{1}{\Omega^k} + \frac{10 \epsilon_1^2 \eta^2 E^3}{K} \sum_{k=0}^{K-1} \frac{1}{\Omega^k} + \nonumber\\
	&\squad \frac{4 E \eta^2}{K} \sum_{k=0}^{K-1} \frac{1}{\Omega^k} \sum_{e=0}^{E-1} \sum_{b=0}^{B-1} \rs \alpha_b \rs\rs \sum_{u \in \bar{\mathcal{U}}_b^{k,e}} \alpha_u \left(\alpha_u + 15 E \beta^2 \eta^2 \mathrm{L}^3 \right) \rs \bigg[\frac{1}{\mathrm{p}_{u,\mathrm{sc}}^{k,e}} - 1 \bigg] \sum_{l=0}^{\mathrm{L}_u^{k,e} - 1} \mathbb{E} \left[\left\Vert g_u (\mathbf{w}_u^{k, e, l}) \right\Vert^2\right].
\end{align}

It is worth noting that when $\eta < \frac{1}{2\sqrt{5} \beta E \mathrm{L}}$, we have $0 < (1 - 20\beta^2\eta^2E^2 \mathrm{L}^2) < 1$. 
In order to satisfy the previous assumptions on the learning rate, we let $\eta < \mathrm{min}\left\{\frac{1}{2\sqrt{5} \beta \mathrm{L}}, \frac{1}{\beta E \mathrm{L}} \right\}$. 
As such, we re-write (\ref{lemma_bs_server_eqn_1}) as follows:
\begin{align}
	\label{lemma_bs_server_eqn_pf}
	&\frac{1}{K} \sum_{k=0}^{K-1} \frac{1}{\Omega^k} \sum_{e=0}^{E-1} \sum_{b=0}^{B-1} \alpha_b \mathbb{E} \left[ \left\Vert \mathbf{w}^k - \mathbf{w}_{b}^{k,e} \right\Vert^2 \right] \nonumber\\
	&\leq \frac{4 E \eta^2 \sigma^2} {K} \sum_{k=0}^{K-1} \frac{1}{\Omega^k} \sum_{e=0}^{E-1} \sum_{b=0}^{B-1} \alpha_b \rs\rs \sum_{u \in \bar{\mathcal{U}}_b^{k,e}} \rs \left(\alpha_u\right)^2 \mathrm{L}_u^{k,e} + \frac{60 \beta^2 \sigma^2 E^3 \mathrm{L}^3 \eta^4}{K} \sum_{k=0}^{K-1} \frac{1}{\Omega^k} + \frac{180 \beta^2 \epsilon_0^2 E^3 \mathrm{L}^4 \eta^4}{K} \sum_{k=0}^{K-1} \frac{1}{\Omega^k} + \frac{10 \epsilon_1^2 \eta^2 E^3}{K} \sum_{k=0}^{K-1} \frac{1}{\Omega^k} + \nonumber\\
	&\squad \frac{4 E \eta^2}{K} \sum_{k=0}^{K-1} \frac{1}{\Omega^k} \sum_{e=0}^{E-1} \sum_{b=0}^{B-1} \rs \alpha_b \rs\rs \sum_{u \in \bar{\mathcal{U}}_b^{k,e}} \alpha_u \left(\alpha_u + 15 E \beta^2 \eta^2 \mathrm{L}^3 \right) \rs \bigg[\frac{1}{\mathrm{p}_{u,\mathrm{sc}}^{k,e}} - 1 \bigg] \sum_{l=0}^{\mathrm{L}_u^{k,e} - 1} \mathbb{E} \left[\left\Vert g_u (\mathbf{w}_u^{k, e, l}) \right\Vert^2\right].
\end{align}
This concludes the proof of Lemma \ref{lem_div_central_BS_models}.

\end{document}